\documentclass[a4paper,UKenglish,cleveref, autoref, thm-restate]{lipics-v2021}
\nolinenumbers
\hideLIPIcs
\title{When Does Sparsity Help for $k$-Independent Set in Hypergraphs and Other Boolean CSPs?}

\author{Timo Fritsch}{Karlsruhe Institute of Technology, Germany}{timo.fritsch@hu-berlin.de}{https://orcid.org/0009-0006-3768-1382}{}
\author{Marvin Künnemann}{Karlsruhe Institute of Technology, Germany}{}{[orcid]}{Research partially supported by the Deutsche Forschungsgemeinschaft (DFG, German Research Foundation) -- 462679611.}
\author{Mirza Redzic}{Karlsruhe Institute of Technology, Germany}{mirza.redzic@kit.edu}{https://orcid.org/0009-0001-7509-1686}
{Research partially supported by the Deutsche Forschungsgemeinschaft (DFG, German Research Foundation) -- 462679611.}
\author{Julian Stie{\ss}}{Karlsruhe Institute of Technology, Germany}{julian.stiess@kit.edu}{[orcid]}{Research partially supported by the Deutsche Forschungsgemeinschaft (DFG, German Research Foundation) -- 462679611.}

\authorrunning{T. Fritsch, M. Künnemann, M. Redzic and J. Stie{\ss}} %TODO mandatory. First: Use abbreviated first/middle names. Second (only in severe cases): Use first author plus 'et al.'

\Copyright{Timo Fritsch, Marvin Künnemann, Mirza Redzic and Julian Stie{\ss}} %TODO mandatory, please use full first names. LIPIcs license is "CC-BY";  http://creativecommons.org/licenses/by/3.0/

%TODO mandatory: Please choose ACM 2012 classifications from https://dl.acm.org/ccs/ccs_flat.cfm 

 %TODO mandatory; please add comma-separated list of keywords

\category{} %optional, e.g. invited paper

\relatedversion{} %optional, e.g. full version hosted on arXiv, HAL, or other respository/website
\usepackage{tikz}
\usetikzlibrary{shapes, positioning, backgrounds, calc}

\tikzset{
    std_node/.style={circle, fill=black, inner sep=1.5pt, outer sep=0pt},
    highlighted_node/.style={circle, fill=red!20, draw=red, very thick, inner sep=2pt, outer sep=0pt},
    hyperedge_e/.style={fill=blue!20, opacity=0.7, draw=blue, thick},
    hyperedge_ep/.style={fill=green!20, opacity=0.7, draw=green!60!black, thick},
    label_style/.style={fill=none, font=\small}
}

\usepackage{relsize}
\usepackage{mathtools}
\usepackage{todonotes}
\usepackage{calc}
\usepackage{hyperref}

\newcommand{\smallbold}[1]{{\textbf{#1}}}

\newcommand{\satg}[3]{\textsc{Csp$_{#1}^{#3}$}(#2) }
\newcommand{\IMPL}[0]{\smallbold{IMPL}}
\newcommand{\NAND}{{\textbf{NAND}}}
\newcommand{\EQ}[0]{\smallbold{EQ}}
\newcommand{\NOR}{{\smallbold{NOR}}}

\newcommand{\CSP}{\ensuremath{\textsc{Csp}}}

\renewcommand{\O}{\mathcal{O}}
\newcommand{\Fam}{\mathcal{F}}

\newcommand{\julian}[1]{{\color{red}J: }{\color{orange}{#1}}}
\newcommand{\marvin}[1]{{\color{red}MK: }{\color{orange}{#1}}}
\newcommand{\timo}[1]{{\color{red}T: }{\color{orange}{#1}}}
\newcommand{\mirza}[1]{{\color{red}MR: }{\color{orange}{#1}}}

\renewcommand{\julian}[1]{}
\renewcommand{\marvin}[1]{}
\renewcommand{\timo}[1]{}
\renewcommand{\mirza}[1]{}

\newcommand{\ponder}[1]{%
    \phantomsection%
    \textcolor{orange}{#1}%
    \addcontentsline{toc}{subsection}{\textcolor{orange}{\textbf{TODO:}} #1}%
}

\renewcommand{\ponder}[1]{}

\newcommand{\triv}{{\gamma_{\mathrm{triv}}}}
\newcommand{\umin}{{u_{\mathrm{min}}}}
\newcommand{\smin}{{s_{\mathrm{min}}}}

\newtheorem{hypothesis}{Hypothesis}

\usepackage{algorithm}
\usepackage[noend]{algpseudocode}

\begin{document}

\maketitle

%TODO mandatory: add short abstract of the document
\begin{abstract}
Consider the fundamental task of finding independent sets of (constant) size $k$ in a given $n$-node \emph{hyper}graph. How much is the time complexity affected by the sparsity of the input, i.e., the number of hyperedges $m$? Tur\'{a}n's theorem implies that the problem is trivial if $m=O(n^{2-\epsilon})$ for some $\epsilon> 0$.  Above that threshold (i.e., if $m=\Theta(n^\gamma)$ for some $\gamma \ge 2$), we give a perhaps surprising algorithm with running time $O\left(\min\left\{n^{\frac{\omega}{3}k} + m^{k/3}, n^k\right\}\right)$ (for $k$ divisible by 3), which is essentially \emph{conditionally optimal} for all $\gamma\ge 2$, assuming the $k$-clique and 3-uniform hyperclique hypotheses (here, $\omega \le 2.372$ denotes the matrix multiplication exponent). In fact, we obtain a more detailed time complexity that is sensitive to the arity distribution of the hyperedges.%\marvin{rephrased to highlight the more detailed time bound}

To study such phenomena in more generality, we study the time complexity of finding solutions of (constant) size $k$ in sparse instances of Boolean constraint satisfaction problems, where $n$ and $m$ denote the number of variables and constraints, respectively. %Here, we can fully classify interesting subclasses and observe a diverse set of time complexities. 
Our results include, among others:
\begin{itemize}
    \item an essentially full classification of the influence of sparsity for Boolean constraint families of binary arity. Of particular technical interest is a conditionally tight algorithm for the family consisting of the binary NAND and the binary Implication constraints, with a running time of $\Theta(m^{\omega k/6 \pm  c})$. %whenever $m=\Omega(n)$.
    \item the identification of a large class of constraint families $\Fam$ that exhibits a sharp phase transition: there is a threshold $\gamma_\Fam$ such that the problem is trivial for $m=O(n^{\gamma_\Fam-\epsilon})$, but requires essentially brute-force running time $\Theta(n^{k\pm c})$ for $m=\Omega(n^{\gamma_\Fam})$, assuming the 3-uniform hyperclique hypothesis. 
\end{itemize}
In general, we observe a rich landscape of time complexities. Notably, in many cases the \emph{combination} of constraints display higher time complexity than either constraint alone.
%\marvin{claims in the abstract correct?}
\end{abstract}

%\marvin{TODO: read Appendix B}
%\marvin{make Proof of Triviality Cutoff its own section?}
%\marvin{make $k$-IS in arity-$3$ hypergraphs its own section?}

%\begin{itemize}
%    \item should mention connection between Turan's theorem and the fact that sparse independent set instances are trivial. An extension of Turan's theorem to higher arity appears to be a difficult subject in the combinatorics community. (see, e.g., here: https://people.maths.ox.ac.uk/keevash/papers/hyp-reg-turan-journal.pdf) 
%   \timo{I believe a sufficiently general result for hypergraphs can be found here: 
%   \url{https://en.wikipedia.org/wiki/Tur\%C3\%A1n_number}. 
%   Wikipedia quotes a book by Dinitz, for which i wasnt able to find a digital copy, probably one can proof a similar result themselves? The precise formulation should quite instantly yield that $k$-indSet in $r$-uni HGs with $\gamma < r$ must contain a solution.}
%   \julian{The handbook cites this paper on \url{https://link.springer.com/article/10.1007/BF01375476} i suppose.}
%    \item in the conference version, we might want to give only a sketch for the case of graphs, and prove the statement for hypergraphs with minimum arity $r$ in a later technical section?
%    \item Notation: Timo's name of "sparsification" for the main reduction type should be called differently, e.g., "sparse embeddings"
%    \item kernel $\to$ core? possible names for general method: ``peeling off'', ...?
%\end{itemize}
%\tableofcontents

\section{Introduction}
%\textbf{Outline:} 
%\begin{itemize}
%    \item Hook on sparsity (i.e. many interesting graphs are indeed sparse) \marvin{cite, e.g., FKR SODA'24 for DomSet}
%    \item $k$-independent set is easy in sparse graphs
%    \item Switching lens to parametrized CSP, as NAND-like constraints are indeed the hardest class there
%    \item Sparsity indeed as a notable effect on sparsity for CSPs (if NAND appears)
%    \item We discuss the limitations that apply in the field of CSPs when considering sparsity (i.e. restrictions and implementations are now very limited)
%    \item More interestingly, this leads to very nuanced changes in hardness when considering constraint families containing NAND as well
%    \item We fully resolve the complexity 
 %       of binary constraint families under the current matrix multiplication exponent $\omega$ value.
%        (\autoref{sec:binary})
%    \item For higher arity constraint families, 
%        we discover a complex landscape of lowerbounds
%        for which we can prove tightness %for some particular constraint families.
%\end{itemize}

How does \emph{sparsity} of the input influence a problem's time complexity? From an algorithmic perspective, such questions have been well studied particularly for graph-theoretic problems, as many real-world graphs are rather far from dense. Indeed, for a problem solvable in time $n^{c\pm o(1)}$ on $n$-node graphs, a natural target to shoot for is an $m^{c/2\pm o(1)}$-time algorithm, where $m$ denotes the number of edges in the graph. Such an algorithm recovers the time bound of $n^{c\pm o(1)}$  in the dense case $m=\Theta(n^2)$, while significantly improving the running time if $m = O(n^{2-\epsilon})$ -- we shall refer to such a situation as a \emph{natural interpolation}. For many graph problems, natural interpolation is indeed achievable. %\marvin{this passage is perhaps a bit close to our SODA paper?} \marvin{cite well-known examples here?} \marvin{add a sentence like: Natural interpolation is indeed possible for problems such as global min-cut, max flow; watch out: needs time $m^{c/2+o(1)}$ rather than $O(m^{c/2})$.}

%From the hardness side, the precise influence of sparsity has only been established for a few select problems:
Much less is known on the hardness side.
Nevertheless, a growing body of work establishes conditional lower bounds ruling out natural interpolation for several fundamental problems. Notable examples include: approximate Diameter and Radius~\cite{RodittyVW13,AbboudVWW16},  All-Edges Triangle Detection~\cite{Patrascu10,WilliamsX20}, APSP and related problems~\cite{DBLP:conf/soda/LincolnWW18}, $k$-Dominating Set for $k\ge 3$ ~\cite{fischer2024effect} and more; see also~\cite{AgarwalR18}. %\marvin{this paragraph could use some work}

The goal of this work is to investigate the rich interplay  of the input sparsity and the resulting time complexity. Particularly, we focus on cases in which instead of achieving a natural interpolation, the running time exhibits a complex relationship with the sparsity -- possibly overshooting or undershooting natural interpolation, or both. Predominantly, we investigate problems beyond graphs, specifically, hypergraph problems as well as constraint satisfaction problems. In such problems, $n$ objects (nodes or variables) interact in a combination of relationships (hyperedges of different arity or constraints of different types). If each relation or constraint has arity at most $r$, then the sparsity $m$ (i.e., the total size of these relationships) is bounded by $O(n^r)$. In these cases, for an $n^{c\pm o(1)}$-time solvable problem, we denote the natural interpolation as a running time of the form $m^{c/r\pm o(1)}$. We shall see several cases in which this running time is \emph{partially} achievable, using \emph{conditionally optimal} algorithms that carefully consider the combination of relationships.

%From a hardness perspective, such questions are less well investigated.
%\marvin{where to discuss previous conditional lower bounds for sparse graphs? Such as Diameter LB, Lincoln et al., FischerKR SODA + subsequent works? only later in related work section?}

\subsection{Result I: Conditionally Optimal $k$-Independent Set in Sparse Hypergraphs}

As our first focus, we consider the $k$-Independent Set problem ($k$-IS) in hypergraphs: Given a (hyper)graph $G=(V,E)$, determine if there is a set $S\subseteq V$ of size $k$ that contains no (hyper)edge of $G$, i.e., $e\not\subseteq S$ for all $e\in E$. Already in graphs, $k$-IS is central to algorithmic graph theory: by complementing the graph, we obtain the classic $k$-clique problem, perhaps the best known $W[1]$-complete problem.\footnote{Note however, that complementing a sparse graph generally yields a dense graph, so that the influence of sparsity differs between $k$-clique and $k$-IS.} Generalizing to hypergraphs, we obtain a significantly more expressive problem. Indeed, for $3$-uniform hypergraphs, it is equivalent (up to sparsity) to the 3-uniform hyperclique problem. This problem is generally considered to be more difficult than $k$-clique; the corresponding 3-uniform hyperclique hypothesis has seen a surge of applications recently, see ~\cite{DBLP:conf/soda/LincolnWW18,AbboudBDN18,BringmannFK19,KunnemannM20,BringmannS21,DalirrooyfardJW22,MathialaganWX22,FischerKRS25,FuLR25}, among others (see below for further discussion). As a further case in point, note that the fairly recent hypergraph container method (see, e.g.,~\cite{BaloghMS18survey}) gives a combinatorial tool for independent sets in well-behaved hypergraphs, which has found various applications, including in algorithm design~\cite{Zamir23}.

We ask:
\begin{centering}
\itshape 
Let $k\ge 2$ and $\gamma\ge 0$ be constants. What is the time complexity of $k$-Independent Set in $n$-node hypergraphs with $m=\Theta(n^\gamma)$ hyperedges?
\end{centering}

Let us first consider the case of graphs rather than hypergraphs. Here, $0 \le \gamma \le 2$, and without taking sparsity into account, $k$-IS is well known to be solvable in time $O(n^{(\omega/3)k})$~\cite{cliqueFast}\footnote{Whenever $k$ is divisible by 3; in other cases, the running time is only slightly higher.}; the $k$-clique hypothesis postulates that this running time is essentially optimal. However, taking sparsity into account, the problem becomes \emph{trivial} when $\gamma < 2$, i.e., $m=O(n^{2-\epsilon})$: This already follows from Tur\'an's Theorem, which implies that any $n$-node graph with at most $(\frac{1}{k-1}-o(1))\frac{n^2}{2}$ edges contains an independent set of size $k$. 

The case for $h$-uniform hypergraphs with $r\ge 3$ can be resolved analogously: Without taking sparsity into account, no substantial improvement over brute-force running time $O(n^k)$ is known (see~\cite{AbboudFS24} for subpolynomial-factor improvements), and the 3-uniform hyperclique hypothesis postulates that this is best possible. While generalizing Tur\'an's theorem to hypergraphs is challenging (see, e.g.,~\cite{Keevash11-survey}), a simple greedy argument establishes that $k$-IS in $h$-uniform hypergraphs with $\gamma < h$ (i.e., the hypergraph contains $m\le O(n^{h-\epsilon})$ hyperedges) is again trivial. (We shall prove this formally in Section~\ref{sec:triviality-cutoff}). Thus, for any arity $2 \le h \le k-1$, the $k$-IS problem in $h$-uniform hypergraphs is non-trivial only for the hardest, dense case of $\gamma = h$.\footnote{The case of $h=k$ is trivial for all $\gamma < k$ and trivially solvable in linear time in the input if $\gamma=k$.}

We thus turn to the general case of hypergraphs with possibly \emph{mixed-arity} hyperedges. We may assume without loss of generality that $\gamma \le k$, since we may simply ignore any edge of arity at least $k+1$. The problem is trivial only for $\gamma < 2$. How does the problem's complexity behave in between? E.g., how quickly can we solve $k$-IS when we have, say, $\Theta(n^2)$ edges (of binary arity) and $\Theta(n^{2.5})$ hyperedges of arity 3? 

In this case, we determine a conditionally optimal running time of $O(n^{(5/6)k\pm c})$, which is intermediate between $O(n^{(\omega/3)k\pm c})$ and $O(n^{k \pm c})$. More generally, we obtain the following result.
\begin{theorem}
    Assuming the Clique and the 3-uniform Hyperclique Hypothesis, the optimal running time for $k$-Independent Set in $n$-node hypergraphs with $m$ hyperedges, each of arity at least $2$, is
    \[ \min\{ n^{\frac{\omega}{3}k} + m^{k/3}, n^k \},\]
    up to a factor of the form $O(n^c)$ for some $c$ independent of $k$.
\end{theorem}
For an illustration, we refer to Figure~\ref{fig:beyond_binary}(a) -- note that for $\omega \le \gamma\le 3$, and only there, natural interpolation is obtained. In fact, we give a much more detailed statement that takes the arity distribution into account, see Section~\ref{sec:technical-overview}.

Technically, our result relies on a very careful combination of ideas, which overcomes the challenge of incorporating different ways to handle edges of different arity. The main technical obstacle here is that the well known $k$-clique algorithm in graphs due to Nešetřil and Poljak~\cite{cliqueFast} is not well compatible with including \emph{any} hyperedge of arity at least $3$. However, exploiting that it can be used to \emph{count} all $k$-cliques as well, we seek to count all $k$-cliques on the (binary-arity) edges, and subtract from this count the number of solutions including any of the $\le m$ hyperedges of larger arity. To solve this task, we present a careful argument based on inclusion-exclusion and sparse witness listing in Section~\ref{sec:technical-overview}.\footnote{We remark that the idea of using inclusion-exclusion in combination with fast clique/triangle counting has been used in prior work, e.g., to obtain a $O(m^{2\omega/(\omega+1)})$-time algorithm for counting 3-ISes~\cite[Footnote 2]{VWilliamsWWY15}. In our work, we exploit these ingredients further to overcome a different challenge, i.e., handling the presence of mixed-arity edges.}

%Implicit in the above results is the insight that to understand the influence of sparsity for mixed-arity hyperedges, we have to closely consider both the presence of arity-2 edges as well as hyperedges of arity at least $3$.

We now turn towards a vastly more general setting, specifically, the class of Boolean constraint satisfaction problems, which includes the $k$-IS problem in hypergraphs as just one of many interesting examples.  

\subsection{Result II: On the Influence of Sparsity for Boolean CSPs}

Boolean constraint satisfaction has long served as a challenge for understanding the precise limits of our algorithmic and complexity-theoretic methods. The extensive list of such works includes classifications of tractability in terms of P vs NP-complete~\cite{Schaefer78,DBLP:conf/focs/Zhuk17,DBLP:conf/focs/Bulatov17}, counting complexity~\cite{Bulatov13}, parameterized complexity~\cite{Marx05,BulatovM14}, approximability~\cite{DBLP:conf/stoc/KhannaSW97} and many more. We wish to explore how sparsity affects the fine-grained time complexity investigated in~\cite{KunnemannM20}.

Formally, we define, for any finite constraint family $\Fam$ and $\gamma > 0$, the algorithmic problem $\CSP_k^\gamma(\Fam)$: given a set of $m=\Theta(n^\gamma)$ constraints, each formed by applying some function $f:\{0,1\}^r \to \{0,1\} \in \Fam$ on a set of $r$ pairwise distinct Boolean variables chosen from $x_1,\dots, x_n$, determine whether there exists an assignment that sets precisely $k$ variables to true and satisfies all constraints. For a singleton family $\Fam = \{f\}$, we also write $\CSP_k^\gamma(f)$.
Throughout the paper, we consider $k$ as a  constant independent of $n$.

Note that this class contains $k$-IS in $h$-uniform hypergraphs with $m=\Theta(n^\gamma)$ hyperedges as $\CSP_k^\gamma(\NAND_h)$, where $\NAND_h(y_1,\dots,  y_h) = \overline{y_1\wedge \cdots \wedge y_h}$. It also contains $k$-IS in mixed-arity hypergraphs as $\CSP_k^\gamma(\{\NAND_2,\dots, \NAND_k\})$.

We ask:
\emph{What determines the influence of sparsity for detecting size-$k$ solutions of Boolean CSPs?}

Generally speaking, we obtain a full\marvin{add discussion in footnote}\footnote{Full with regard to the effects of sparsity, as the classification is not tight for $\CSP_k^\gamma(\IMPL)$, for which the precise complexity is unclear in the dense case already.} classification for constraint families consisting exclusively of functions of binary arity. For higher-arity constraint functions, we give a large subclass of CSPs exhibiting a precise cutoff point: for sparser instances, the problem is trivial, for denser instances, it is essentially as hard as the dense case.

\paragraph*{Classification for Binary Constraint Families}

Without taking sparsity into account, previous works~\cite{Marx05,KunnemannM20} establish the following regimes within the Boolean constraint families: (1) problems fine-grained equivalent to $\CSP_k(\NAND)$, (2) problems fine-grained equivalent to $\CSP_k(\IMPL)$ where $\IMPL(x,y) = x \rightarrow y$, and (3) FPT problems.

For the two central hard CSPs, it is not too difficult to establish the following baselines (see Section~\ref{sec:binary}):
\begin{enumerate}
\item  $\CSP_k^{\gamma}(\NAND)$, i.e., $k$-IS in graphs, is trivial if $\gamma < 2$ and has complexity $O(n^{\omega k/3 \pm c})$ if $\gamma = 2$, assuming the clique hypothesis.
\item $\CSP_k^\gamma(\IMPL)$ is trivial if $\gamma < 1$ and has time complexity $n^{g(k)}$ uniformly for all $1\le \gamma \le 2$, where $\Omega(\sqrt[3]{k}) \le g(k) \le O(\sqrt{k})$, assuming the clique hypothesis.
\end{enumerate}
    
We find that surprisingly, the \emph{combined} constraint family $\{\IMPL, \NAND\}$ has \emph{higher} time complexity for all $1 \le \gamma < 2$ than either $\IMPL$ or $\NAND$ individually.

\begin{theorem}
Assuming that the clique hypothesis holds. The optimal time complexity for $\CSP^\gamma_k(\{\NAND, \IMPL\})$ is $\Theta(m^{\omega k/6 \pm c})$ for some $c$ independent of $k$.\footnote{By slight use of notation, in the introduction we write $\Theta(n^{f(k)\pm c})$ to express existence of an algorithm with running time $O(n^{f(k)+c})$ and a conditional lower bound of $\Omega(n^{f(k)-c})$.}
\end{theorem}

Thus, for $1\le \gamma < 2$, while $\CSP_k^\gamma(\NAND)$ is trivial and $\CSP_k^\gamma(\IMPL)$ has a subexponential time complexity $n^{g(k)}$, the combination of both constraints has an exponential time complexity that is a natural interpolation of the $k$-IS running time.

The above theorem consists of two parts: (1) a conditional lower bound  and (2) a matching algorithm. The conditional lower bound is strikingly simple: We can introduce $\approx n-n_0$ $\IMPL$-constraints enforcing that any satisfying assignment with $k$ nonzeroes chooses its nonzero variables from a set $V_0\subseteq V$ of size $n_0$. On $V_0$, we can thus embed an arbitrary, possibly dense instance as long as $n_0^2 = O(m)$. Thus, by choosing $n_0 \approx \sqrt{m}$, we can reduce from a $k$-clique instance in an $\sqrt{m}$-vertex graph, yielding a conditional lower bound of $m^{k\omega/6-o(1)}$, as desired. %\marvin{introduce the name for this?}

The corresponding upper bound does not follow as easily. To beat  time $O(n^{\omega k/3})$ when the input contains only few \NAND- and \IMPL-constraints, one might hope for a ``reversal'' of the reduction in the conditional lower bound: is it possible to quickly ``peel off'' most of the variables so that a small set of $O(\sqrt{m})$ variables survives? Unfortunately, possibly intricate interactions between \IMPL- and \NAND-constraints prevent simple preprocessing schemes. Instead, we perform a careful combination of arguments:
%\begin{itemize}
    %\item 
    We first preprocess the instance to obtain a structured setting in which crucially each variable implies at most one other variable.
     We use this structured setting to group the variables, yielding a win-win argument: Either (1) there exists a ``large'' group with few $\NAND$ constraints, and we can detect a solution within this single group greedily, or (2) all groups are ``small''.
    In the remaining case in which all groups are ``small'', we can carefully reduce to an almost-balanced Triangle Detection instance.
%\end{itemize}
For more details, we refer to Section~\ref{sec:technical-overview} and the the full proof in Section~\ref{sec:NAND_IMPL}.

Interestingly, we observe that among the binary constraint families, the regime of $\CSP^{\gamma}_k(\{\NAND,\IMPL\})$ is the only new regime that emerges. Specifically, we arrive at the the following classification, illustrated in Figure~\ref{fig:binary_regime_overview}.

\begin{figure}[h]
    \centering
    \begin{subfigure}[c]{0.48\textwidth}
        \centering
        \includegraphics[width=0.85\textwidth]{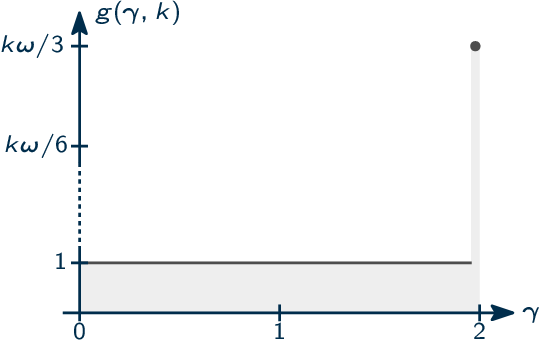}
        \caption{ $\CSP_k^\gamma(\NAND)$, equivalent to $k$-\textsc{Independent Set}}
    \end{subfigure}
    \hfill
    \begin{subfigure}[c]{0.48\textwidth}
        \centering
        \includegraphics[width=0.85\textwidth]{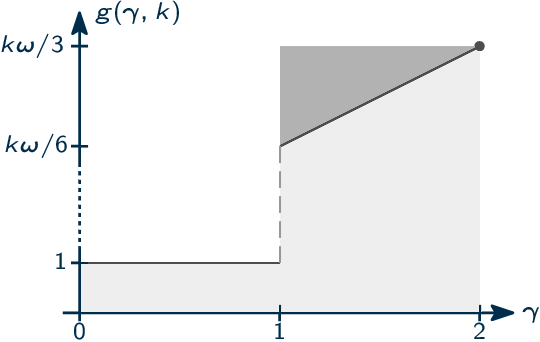}
        \caption{$\CSP_k^\gamma(\{\NAND\} \cup \Fam')$ for any binary, non-trivial family $\Fam'$}    
    \end{subfigure}

    \vspace*{0.5cm}
    
    \begin{subfigure}[c]{0.48\textwidth}
        \centering
        \includegraphics[width=0.85\textwidth]{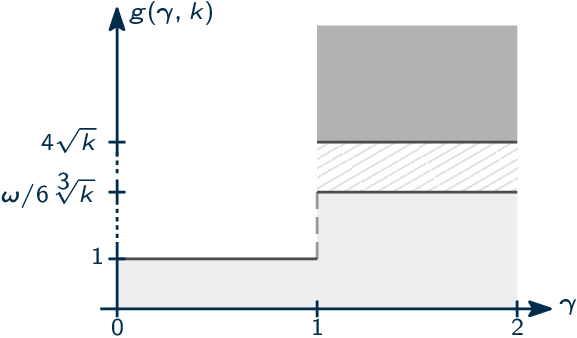}
        \caption{$\CSP_k^\gamma(\{\IMPL\} \cup \Fam')$ for any binary family with $\NAND \notin \Fam'$} 
    \end{subfigure}
    \hfill
    \begin{subfigure}[c]{0.48\textwidth}
        \centering
        \includegraphics[width=0.85\textwidth]{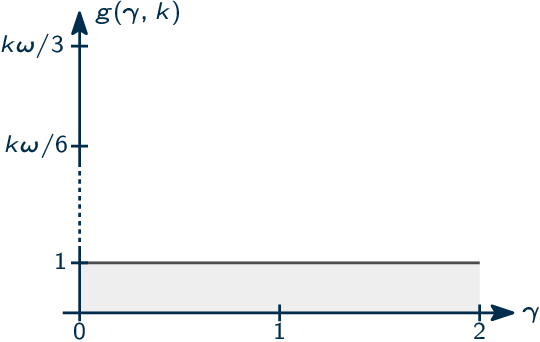}
        \caption{$\CSP_k^\gamma(\Fam)$ for any binary family with $\NAND, \IMPL \notin \Fam$}    
    \end{subfigure}

    \caption{Relationship between runtime and density for our different regimes: Given some density $\gamma$ and family $\Fam$ we obtain upper and conditional lower bounds of the form $n^{g(\gamma,k)}$. For improved readability, we state the exponent ${g(\gamma,k)}$ up to an additive constant independent of $k,\gamma$. Under this tolerance, these regimes are tight, except for case \textit{(c)} where the knowledge gap between $n^{\Omega(\sqrt[3]{k})}$ and $n^{O(\sqrt{k})}$ from the dense case transfers.}
    \label{fig:binary_regime_overview}
\end{figure}

\begin{theorem}[Sparsity Classification for Binary Constraint Families]\label{thm:classification-informal}
Let $\Fam$ be a family consisting exclusively of binary constraint functions. For $0\le \gamma < 1$, $\CSP_k^\gamma(\Fam)$ can be solved in $\O(n)$. Furthermore,
\begin{itemize}
    \item If $\Fam = \{\NAND\} \cup \Fam'$ for some non-empty family $\Fam'$, the optimal time complexity of $\CSP_k^\gamma(\Fam)$ for $1\le \gamma\le 2$ is $\Theta(m^{\omega k/6 \pm c})$, assuming the clique hypothesis.
    \item If $\Fam = \{\NAND\} $, then $\CSP^\gamma_k(\Fam)$ is trivial for $\gamma <2$. For $\gamma=2$, the optimal time complexity is $\Theta(n^{\omega k/3 \pm c})$, assuming the clique hypothesis.
    \item If $\NAND\notin \Fam$, but $\IMPL\in \Fam$, then the optimal time complexity of $\CSP^\gamma_k(\Fam)$  for $1\le \gamma \le 2$  is $n^{g(k)}$ with $\Omega(\sqrt[3]{k}) \le g(k) \le O(\sqrt{k})$, assuming the clique hypothesis.
    \item Finally, if $\IMPL,\NAND \notin \Fam$, then $\CSP^\gamma_k(\Fam)$ can be solved in time $O(n+m)$ for all $1\le \gamma \le 2$.
\end{itemize}
\end{theorem}

We prove this Theorem in~\autoref{sec:binary}.

\paragraph*{Towards Higher-Arity Constraint Families: Phase Transition At Triviality Cutoff}

Classifying the influence of sparsity for higher-arity constraint families becomes significantly more challenging. %As suggested by our previous results, it is essential to study the \emph{interplay} of different constraint functions. We shall thus focus on families $\Fam$ that combine one of the central functions $\NAND$ -- more generally, the $r$-ary $\NAND_r$ for $r\ge 2$ -- or $\IMPL$ with an additional constraint function $f$. \marvin{last sentence probably to delete}
Already the question how many constraints are minimally required to define a non-trivial instance is not obvious. We shall however see, that for a large class of constraint families, this quantity is decisive to understand the time complexity, as we observe a phase transition at this \emph{triviality cutoff}.
%\paragraph*{Triviality cutoff}

%As we observed for binary constraint families $\Fam$, it is crucial to analyze for which values of the sparsity parameter $\gamma$ the problem $\CSP_k^\gamma(\Fam)$ remains trivial.

Specifically, we say that $\CSP_k^\gamma(\Fam)$ is \emph{non-trivial} if there are infinitely many NO instances in $\CSP_k^\gamma(\Fam)$; conversely, it is \emph{trivial}, if for all sufficiently large $n$, any instance of $\CSP_k^\gamma(\Fam)$ with $n$ variables is a YES instance. For any constraint family $\Fam$, we define the \emph{triviality cutoff} $\triv$ as the smallest $\gamma\ge 0$ such that $\CSP_k^\gamma(\Fam)$ is non-trivial. Thus, every $n$-variable instance with sufficiently large $n$ and $O(n^{\triv-\epsilon})$ constraints from $\Fam$ contains a solution, while there exists infinitely many instances with $\Theta(n^\triv)$ constraints from $\Fam$ that do not contain a solution with $k$ nonzeroes.%\marvin{special cases: $\triv = 0$ if no value of $\gamma$ is trivial}.

We first introduce a parameter $\umin(\Fam)$ that describes this triviality cutoff. Specifically, let $\umin(f)$ denote the smallest weight $\|x\|_1$ of an assignment $x$ violating $f$, i.e., $\umin(f) = \min \{ \|x\|_1 \mid f(x) = 0\}$. Defining $\umin(\Fam)\coloneqq \min_{f\in \Fam} \umin(f)$, we can establish the triviality cutoff $\triv$ as $\umin(\Fam)$. 

\begin{theorem}[Triviality Cutoff]
\label{thm:triviality_cutoff}
Let $\Fam$ be any finite constraint family. The problem $\CSP_k^\gamma(\Fam)$ with $\gamma = \umin(\Fam)$ is non-trivial. Conversely, if $\gamma< \umin(\Fam)$, then there exists $n_0$ such that any instance of $\CSP_k^\gamma(\Fam)$ with at least $n_0$ variables admits a satisfying assignment of weight~$k$.
\end{theorem}

%\paragraph*{Phase Transition}

For a large class of constraint families, we obtain a sharp phase transition from trivial instances to requiring brute force running time, assuming the 3-uniform hyperclique hypothesis. For an illustration, see Figure~\ref{fig:beyond_binary}.

\begin{figure}
    \centering
    \begin{subfigure}[c]{0.45\textwidth}
        \centering
        \includegraphics[width=\textwidth]{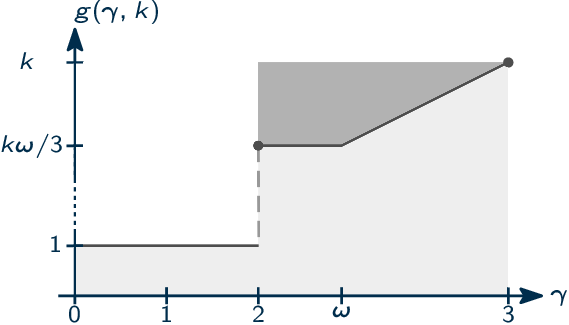}
        \caption{$\CSP_k^\gamma\big(\{\NAND_2, \NAND_3\}\big)$ equivalent to $k$-\textsc{Independent Set} in $3$-hypergraphs}
    \end{subfigure}
    \hfill
    \begin{subfigure}[c]{0.50\textwidth}
        \centering
        \includegraphics[width=\textwidth]{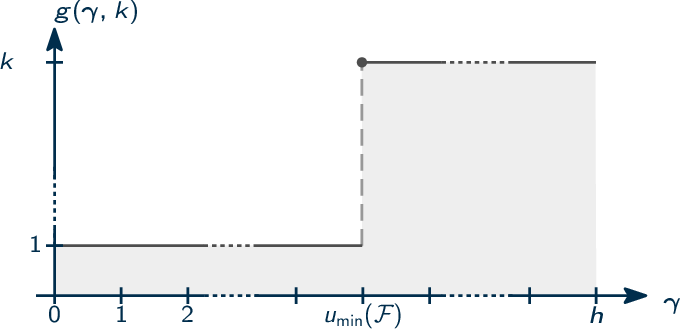}
        \caption{$\CSP_k^\gamma(\Fam)$ for any family $\Fam$ with $\umin(\Fam)\geq 3$}
    \end{subfigure}
    \hfill
    \caption{Relationship between runtime and sparsity for $\CSP_k^\gamma(\Fam)$ with higher-arity $\Fam$: (a) For $\Fam= \{\NAND_2,\NAND_3\}$, the complexity displays a nuanced relationship (a kinked slope) in the regime $2\le \gamma \le 3$; note that $\umin(\Fam)=2$. (b) showcases the sharp phase transition at the triviality cutoff for or families $\Fam$ with $\umin(\Fam) \geq 3$}
    \label{fig:beyond_binary}
\end{figure}

\begin{theorem}[Phase Transition]
\label{thm:phase_transition}
Let $\Fam$ be a constraint family. If $\umin(\Fam) \ge 3$, then $\CSP_k^\gamma(\Fam)$ is trivial for $\gamma < \umin(\Fam)$. For $\gamma \ge \umin(\Fam)$, it requires time $n^{k-o(1)}$ assuming the $\umin(\Fam)$-uniform hyperclique hypothesis.
\end{theorem}

To future work, we leave the challenge of settling the influence of sparsity for constraint families $\Fam$ with $\umin(\Fam)\in \{0,1,2\}$. As witnessed by the mixed-arity $k$-IS family $\Fam= \{\NAND_2,\NAND_3\}$ (which has $\umin(\Fam)=2$), settling such families can become technically quite demanding.

%Classifies, e.g., ...

%\paragraph*{Linear slopes}
%$\Fam = \{\NAND_{10}, \NOR\}$ has complexity $m^{k/10 \pm o(1)}$.
%\begin{theorem}
%Let $\Fam$ be a constraint family and $f,g \in \Fam$. 
%If $1 \leq \umin(f)\leq 2 < \umin(g)$ then $\CSP_k^\gamma(\Fam)$ is trivial for $\gamma \leq \umin(f)$.
%For $\umin(f) \leq \gamma \leq \umin(g)$,
%it requires time $m^{k/\umin{g} - o(1)}$ assuming the $\umin(g)$-uniform $k$-hyperclique hypothesis.
%\end{theorem}
%(last case is captured by previous theorem, should we state this?)
%\paragraph*{Staircases}
%
%\paragraph*{Sloped staircase}
%see Timo's pictures on time complexities $\mathcal{F} = \{\NAND_r, XOR\}$
%
%
%\paragraph*{Towards a full classification}
%
%$\{\NAND_2,\NAND_3\}$ case:
%We can indeed beat the trivial $\O(n^k)$ algorithm by some $0<\epsilon< 1.$
%
%
%Can you beat essentially brute-force for $\{\NAND_r, EQ\}$?
%
%Capture difference between NOR, EQ, IMPL in combination with $\NAND_r$.
%
\subsection{Detailed Results and Technical Overview} \label{sec:technical-overview}
We highlight our main technical contributions, and provide a more detailed overview of our results. 
We begin by presenting our algorithmic contributions for $k$-Independent Set problem in $3$-hypergraph and developing the core framework that will also be useful in extending to general hypergraphs of higher arity. 
\paragraph*{Algorithm for $k$-Independent Set in Sparse $3$-Hypergraphs}
Let $H$ be a $3$-Hypergraph\footnote{That is, each hyperedge contains \emph{at most} $3$ vertices.} with $n$ vertices and $m$ hyperedges. 
Let $G$ be the underlying graph of $H$, obtained by removing all hyperedges of arity $3$ from $H$. 
Notice that any $k$-independent set in $H$ is also a $k$-independent set in $G$. 
We refer to the $k$-independent sets in $G$ as \emph{potential solutions}.
Formally, let $I(G,k)$ denote the set of all independent sets of size $k$ in $G$ and let $I_{\text{invalid}}(H,k)$ denote the set of all \emph{false solutions}, i.e. the set of all potential solutions that contain a hyperedge in $H$. 
We begin with a simple observation that $H$ has a $k$-independent set if and only if $|I(G,k)| - |I_{\text{invalid}}(H,k)|>0$.

Counting the potential solutions, i.e., the value $|I(G,k)|$ can be done via the classical algorithm of Nešetřil and Poljak \cite{cliqueFast}, yielding an $O(n^{\frac{k\omega}{3}})$ bound (if $k$ is divisible by $3$). 
Hence, the main algorithmic challenge is to efficiently count the invalid solutions, namely those potential solutions that violate at least one hyperedge constraint in $H$. 
We now turn our attention to this task.
For any subset of hyperedges $S$, let $I_S(H,k)$ denote the set of all potential solutions that contain all vertices spanned by $S$; that is, for each $X\in I_S(H,k)$, the induced hypergraph $H[X]$ contains all hyperedges in $S$.
Let $E_{\ge 3}$ denote the set of all hyperedges in $H$ that contain at least three vertices. We can now observe that, by definition of $I_{\text{invalid}}$, we have 
\[
I_{\text{invalid}}(H,k) = \bigcup_{e\in E_{\ge 3}} I_{\{e\}}(H,k),
\]
and for any non-empty set $S\subseteq E_{\ge 3}$
\[
I_S(H,k) = \bigcap_{e\in S}I_{\{e\}}(H,k).
\]
This gives us a natural way to compute the size of $I_{\text{invalid}}$ via the inclusion-exclusion principle as follows. 
\[
|I_{\text{invalid}}(H,k)| =  \sum_{i=1}^{\binom{k}{3}}\sum_{S\in \binom{E_{\ge 3}}{i}}(-1)^{i+1}|I_{S}(H,k)|.
\]
However, naively enumerating all sets $S$ and computing $|I_S(H,k)|$ fails to improve upon trivial $O(n^k)$ brute-force algorithm, even in the sparse regime $m=O(n)$.
To circumvent this issue, we notice that a lot of the ''higher-order'' terms in the inclusion-exclusion formula are either irrelevant, or redundant.
Specifically, such terms either (1) span more than $k$ vertices (irrelevant), or (2) can be obtained by appropriately guessing a ''lower-order'' term (redundant).
Case in point: consider a set $S$ that contains $\binom{k}{3}$ many edges. 
The only way that $S$ spans $k$ vertices is if these vertices induce a clique of size $k$ in $H$.
However, in that case, for any set $S'\subset S$ that contains $k/3$ pairwise non-intersecting hyperedges will span the same vertex set and consequently $I_{S'}(H,k) = I_S(H,k)$

This observation will allow us to drastically reduce the number of inclusion–exclusion terms that must be considered, and form the basis of our improved counting algorithm.
Intuitively, it brings us to the following win-win scenario: let $S\subset E_{\ge 3}$ of size $>k/3$. 
Then either (i) the hyperedges in $S$ are ''clustered together'' and we can get away by guessing a smaller set $S'\subset S$ ($S$ is redundant), or (ii) the hyperedges in $S$ are not ''clustered together'', but they span more than $k$ vertices, so we never have to consider $S$ at all ($S$ is irrelevant). 
However, the issue of ''clustered edges'' cannot be resolved by simply restricting our outer sum to stop at $k/3$, as that would generally lead to a lot of double counting. 
To bypass this problem, the idea is to impose a total ordering $\prec$ on the hyperedges of $H$.
Intuitively, whenever a hyperedge $e\in S$ intersects another hyperedge $e'\in E_{\ge 3}$ with $e'\prec e$, the algorithm recognizes that all the independent sets containing $V(e)\cup V(e')$ have been accounted for in another iteration $S^*$ that contains $e'$, so it avoids counting them again.
%\marvin{algorithm unclear here: do we proceed over hyperedges in this total ordering? Why is the condition asymmetric $e'\prec e$, but we only consider $V(e)\cup V(e')$ (symmetric)?}\mirza{the order in which we iterate over the hyperedges is not important. Note that $e\in S$, $e'\in E$, so we guess the edge $e$ (more precisely, a set $S$ that contains $e$), and then look at all the other edges in the hypergraph. If it intersects an edge with lower priority, we will ignore all the potential solutions that contain both $e$ and $e'$ (once we guess a set $S$ that contains $e'$, we will count this, and by doing it in this way, we avoid double counting)}\marvin{the later aspect is missing in the idea: that the independent sets containing $V(e)\cup V(e')$ are accounted for by a different set $S'$...}
It remains to formalize and implement this idea efficiently, which we address next.
\subparagraph*{Avoiding double counting} 
Fix an arbitrary total order $\prec$ on the hyperedges of $H$.
For any nonempty set of hyperedges $S$, let $I'_S(H,k)$ denote the set of all potential solutions $X$ of size $k$ that satisfy the following two conditions.
\begin{enumerate}
    \item $X$ contains all vertices spanned by the hyperedges in $S$.
    \item For every $e\in S$ and every hyperedge $e'\in E_{\ge 3}$ with $e\cap e'\ne \emptyset$, if $e'\prec e$, then $X$ does not contain $e'$.
\end{enumerate}
%Intuitively, the second condition enforces that among potential solutions with intersecting hyperedges, only the minimum one under $\prec$ is accounted for.
We claim that replacing $I_S$ by $I'_S$ in the inclusion-exclusion formula preserves correctness. In particular, we show that
\[
    I_{\text{invalid}}(H,k)= \bigcup_{e\in E_{\ge 3}} I'_{\{e\}}(H,k)
\]
The rough idea is as follows. Consider any $X\in I_{\text{invalid}}(H,k)$. By definition, the induced subhypergraph $H[X]$ contains at least one hyperedge (and possibly as many as $\binom{k}{3}$). Let $e$ be the minimum hyperedge in $H[X]$ with respect to the ordering $\prec$; that is, for every edge $e'\in E(H[X])\setminus \{e\}$, we have $e\prec e'$.
By construction, $X$ satisfies both defining conditions of $I'_{\{e\}}(H,k)$, and hence $X\in I'_{\{e\}}(H,k)$.
Particularly, this shows that ${I_{\text{invalid}}(H,k)\subseteq \bigcup_{e\in E_{\ge 3}} I'_{\{e\}}(H,k)}$. 
The containment in the other direction follows immediately by noticing that for each $S\subseteq E_{\ge3}$, we have $I'_S(H,k)\subseteq I_S$ and as argued previously, $I_{\text{invalid}}(H,k)= \bigcup_{e\in E_{\ge 3}} I_{\{e\}}(H,k)$.
Together, these inclusions imply the claimed equality.

We further observe that any set $S\subseteq E_{\ge 3}$ containing more than $k/3$ hyperedges either: (1) spans more than $k$ vertices, or (2) contains two hyperedges $e,e'$ such that $e\cap e'\ne \emptyset$.
In the first case, since $|V(S)|>k$, no independent set of size $k$ can contain $S$. Also, in the second case either $e\prec e'$, or $e'\prec e$. In either case, any potential solution in $I_S(H,k)$ containing both $e$ and $e'$ violates the second condition in the definition of $I'_S(H,k)$.
In particular, in both cases the set $I'_S(H,k)$ is empty.
Consequently, all ''higher-order'' inclusion-exclusion terms corresponding to $|S|>k/3$ vanish.
Combining this with the observations above, we have the following refinement of our inclusion-exclusion formula.
\[
    |I_{\textnormal{invalid}}(H,k)| =  \sum_{i=1}^{\lceil\frac{k}{3}\rceil}\sum_{S\in \binom{E_{\ge3}}{i}}(-1)^{i+1}|I'_{S}(H,k)|.
\]
The remaining challenge is that, unlike the original sets $I_S(H,k)$, the sets $I'_S(H,k)$ do not admit an obvious efficient counting algorithm.
We are now going to prove that we can indeed compute their sizes efficiently.

\subparagraph*{Computing $|I'_S(H,k)|$} 
Fix a set $S\subset E_{\ge3}$ and consider any hyperedge $e'\in E_{\ge 3}$. 
Suppose that there exists a hyperedge $e\in S$ such that (i) $e\cap e'\ne \emptyset$, and (ii) $e'\prec e$.
In this case, any potential solution containing $e'$ would violate the second condition in the definition of $I'_S(H,k)$, and we enforce this restriction via the following reduction.

If $|e'\setminus e| = 1$, we delete the unique vertex in $e'\setminus e$ from the hypergraph.
If instead $|e'\setminus e| = 2$, we replace the hyperedge $e'$ by an edge connecting the two vertices in $e'\setminus e$ (see Figure~\ref{fig:intersectiontypes} for an illustration\marvin{move picture here?}).
Let $H'$ be the hypergraph obtained from $H$ by iterating this reduction over all hyperedges $e'\in E_{\ge 3}$.
We show that this transformation preserves exactly the desired solutions, namely that $I'_S(H,k) = I_S(H',k)$.
Recall that the set $I_S(H',k)$ can be computed by running the Nešetřil, Poljak~\cite{cliqueFast} $k$-clique algorithm on an appropriate subgraph of the underlying graph of $H'$.
This yields an efficient procedure for computing $|I'_S(H,k)|$ and consequently an efficient way to compute $|I_{\text{invalid}}(H,k)|$. 
By implementing this carefully, we can bound the running time for computing this value as follows (assuming $k$ is divisible by $3$, otherwise we get a small polynomial overhead). \marvin{Should one say more here here? this bounds assumes iterating only over \emph{matchings} $S$, can this be mentioned here?}
    \begin{align*}
        T_k(n,m) &\le  \mathcal{O}\left(m^{\frac{k}{3}} + \sum_{i=0}^{\frac{k}{3}-1} m^i\left( m + n^{(k-3i)\omega/3} \right) \right)\\
        & \le \mathcal{O}\left(m^{\frac{k}{3}} +  n^{k\omega/3}   \right).
    \end{align*}
For more details on the implementation, as well as the detailed computation on this bound, we refer the reader to Section \ref{sec:k-IS-3-hypergraphs}.

We complement this algorithmic result with matching conditional lower bounds.
Using a technique that we call \emph{sparse embedding}, that maps a small dense instance of $k$-Hyperclique Detection in $3$-uniform hypergraphs into a sparse instance of $k$-Independent Set Detection in $3$-hypergraphs, we show that for any density $m= \Theta(n^\gamma)$, $2\le \gamma\le 3$, no algorithm can run in time $\mathcal O(m^{k/3-\varepsilon})$, unless the $3$-Uniform Hyperclique Hypothesis fails.
Moreover, the lower bound of $n^{k\omega/3 - o(1)}$ is inherited directly from hardness of $k$-Independent Set Detection in graphs.
Combining these two lower bounds establishes the conditional optimality of our algorithm 
across the entire sparsity spectrum. \footnote{When $m= \Theta(n^\gamma)$ for some $\gamma<2$, every such instance is a trivial yes-instance as a consequence of Turan's theorem; see Section~\ref{sec:triviality-cutoff} for more details.}

\paragraph*{Extending the Algorithm to General $h$-Hypergraphs}
Let $H$ be an $h$-Hypergraph for some $h\ge 3$, with $n$ vertices and $m = \sum_{i=2}^{h} m_i$ hyperedges, where each $m_i$ is the number of hyperedges of arity $i$.
The goal is to extend the algorithmic framework of detecting $k$-Independent Sets in $3$-Hypergraphs to this more general setting.
We first show that a straightforward extension of the techniques developed for $3$-hypergraphs yields an improvement over the brute-force $O(n^k)$ running time for general hypergraphs, provided that $m_i\le O(n^{3-\varepsilon_i})$ for every arity $i$ (and consequently $m = O(n^{3-\varepsilon})$). 
More precisely, we prove the following proposition, which generalizes the algorithm for $3$-hypergraphs and provides a unified upper bound for $k$-Independent Set in general $h$-hypergraphs.
\begin{restatable}{proposition}{HigherArityAlgo}
\label{prop:higher-arity-algo}
    Given any $h$-hypergraph $H$ with $n$ vertices and $m$ edges, for every $k$ divisible by $3$, there is an algorithm deciding if $H$ contains a $k$-independent set in time $\mathcal{O}\left(\min\left\{n^k, n^{\frac{k\omega}{3}} + m^{\frac{k}{3}}\right\}\right)$. 
\end{restatable}
On a high level, we are taking the same blueprint as in the $3$-hypergraph case: we count the potential solutions and subtract the number of invalid solutions. Recall,
\[
    |I(H,k)| = |I(G,k)|- |I_{\text{invalid}}(H,k)|,
\]
where $G$ denotes the underlying graph of $H$.
While the value $|I(G,k)|$ can still be computed efficiently, evaluating $|I_{\text{invalid}}(H,k)|$ becomes more complicated as we increase the arity.
As before, we employ the inclusion-exclusion formula, and eliminate the higher order terms in the same way, by using the following formula. 
\[
    |I_{\textnormal{invalid}}(H,k)| =  \sum_{i=1}^{\lceil\frac{k}{3}\rceil}\sum_{S\in \binom{E_{\ge3}}{i}}(-1)^{i+1}|I'_{S}(H,k)|.
\]
However, the main difficulty lies in computing $|I'_S(H,k)|$. 
Recall that for $3$-hypergraphs, this was achieved by imposing an ordering on hyperedges, and applying the two simple reduction rules (see Figure~\ref{fig:intersectiontypes}), thereby reducing the problem to computing a small independent set in an appropriate graph. 
For hypergraphs of higher arity, however, when two hyperedges $e,e'$ intersect, the set $e'\setminus e$ might contain more than two vertices, so simply adding an edge of arity $2$ no longer suffices.
A natural generalization of the ''type $2$'' reduction rule is to replace $e'$ by a hyperedge spanning $e'\setminus e$.
While this preserves correctness of the construction, crucially, it reduces the problem only to detecting an independent set in hypergraphs, and we can no longer apply the matrix-multiplication-based algorithm to improve upon trivial runtime (see \cite{DBLP:conf/soda/LincolnWW18} for a detailed discussion). 
Fortunately, we can circumvent this issue by observing that the new hypergraph has strictly smaller arity than $H$, allowing us to apply induction on the arity, with the $3$-hypergraph algorithm as a base case.

In fact, we improve upon this by using a combination of the three techniques (1) partitioning the hyperedges of $H$ into ''sparse'' and ''dense'' components (2) applying the inclusion-exclusion procedure described above to the sparse part (3) using sparse witness enumeration to handle the dense parts.
With a careful implementation of these techniques, and by addressing several technical obstacles, we can finally prove the following upper bound for the $k$-Independent Set problem on general hypergraphs. %\mirza{Do you guys think this is too vague? Should I provide more details?}
\begin{restatable}[$k$-Independent Set Algorithm for Higher Arity Hypergraphs]{theorem}{MixedArityAlgo}
\label{thm:mixed-arity-algo}
    Given any $h$-hypergraph $H$ with $n$ vertices and $m_i = \Theta(n^{\gamma_i})$ edges of arity $i$ (for each $i\ge 2$), there is an algorithm deciding if $H$ contains a $k$-independent set for any $k$ divisible by $3$ in time
    \[
    \mathcal{O}\left(n^{\frac{k\omega}{3}} + \sum_{i=3}^h \min\left\{ m_i^{\lceil \frac{k-i+3}{3}\rceil}, m_in^{k-i}\right\}\right).
    \]
\end{restatable}
Perhaps surprisingly, we show that this upper bound is essentially tight: unless either the $k$-Clique Hypothesis, or the $3$-Uniform Hyperclique Hypothesis fails, no significantly faster algorithm is possible.
More specifically, by employing appropriate notions of sparse and dense embeddings, we can prove the following theorem.
\begin{restatable}[Arity-Sensitive Lower Bound]{theorem}{MixedArityLowerBound}\label{thm:mixed_arity_lb}
    Let $\varepsilon>0$ be arbitrary and $n$ be the number of vertices of any given $h$-hypergraph.
    For any $3\le i\le h$, let $m_i = \Theta(n^{\gamma_i})$ for any $2\le \gamma_i\le i$ be the number of hyperedges of arity $i$ of the input hypergraph. %Let $r_i = \lfloor \gamma_i\rfloor$.
    Then:
    \begin{enumerate}
        \item \emph{(Clique Lower Bound)} There is no algorithm solving $k$-Independent Set problem in $h$-hypergraphs in time $\mathcal{O}\left(n^{\frac{k\omega}{3}-\varepsilon}\right)$, unless the $k$-Clique Hypothesis fails.
        \item \emph{($3$-Uniform Hyperclique Lower Bound)} For no $i$ such that $2\le \gamma_i\le 3$ is there an algorithm solving the problem in time $\mathcal{O}\left( m_i^{\frac{k-i+3}3-\varepsilon}\right)$, unless the $3$-Uniform Hyperclique Hypothesis fails. 
        \item \emph{($r$-Uniform Hyperclique Lower Bound)} For no $i$ such that $\gamma_i\ge 3$ is there an algorithm running in $\mathcal{O}\left(m_in^{{k-i-\varepsilon}}\right)$, unless the $\left(\lfloor \gamma_i\rfloor\right)$-Uniform Hyperclique Hypothesis fails.
    \end{enumerate}
\end{restatable}

\paragraph*{Beyond Independent Sets: Boolean CSP} To study the effect of sparsity systematically in a more general class of problems, we turn to the class of Boolean Constraint Satisfaction Problems (CSPs). 
Recall that $k$-Independent Set in $h$-hypergraphs can be viewed as a special case of Boolean CSP over the constraint family $\Fam = \{\NAND_2, \dots, \NAND_{h}\}$.
We formally define this class of problems as follows.
\begin{restatable}[Boolean Constraint Satisfaction Problem $(\CSP_k^\gamma)$]{definition}{BooleanCSPDefinition}
    Let $\Fam$ be a finite Boolean constraint family (i.e. a set of functions $f:\{0,1\}^h \to \{0,1\}$). 
    Given a set $\Phi$ of $m$ Boolean constraints $C$ on variables 
    $x_1,\dots,x_n$, each of the form $f(x_{i_1},\dots, x_{i_h})$, where $f \in \Fam$ and $m \in \Theta(n^\gamma)$,
    the problem $\CSP_k^\gamma$ asks for the existence of an assignment $a:\{x_1,\dots, x_n\}\to \{0,1\}$ satisfying $\Phi$, that sets precisely $k$ variables to $1$. 
\end{restatable}
We begin by considering the special case in which the constraint family $\Fam$ consists exclusively of binary constraints.
Perhaps surprisingly, for every such family $\Fam$, and for all values of $\gamma$, we obtain a complete classification of the problem $\CSP_k^\gamma(\Fam)$: we design an algorithm running in time $T_k(m,n)$, and show its conditional optimality under the $k$-Clique Hypothesis.
To obtain such a complete classification across the entire sparsity spectrum, we first look at the instances of $\CSP_k^\gamma(\Fam)$ for values $\gamma<1$, and show that any such instance admits a linear time algorithm: $\mathcal O(n+m)$.
The classification becomes much more interesting when we consider the regime $\gamma\ge 1$. In particular we prove the following classification theorem, which is a more formal version of Theorem~\ref{thm:classification-informal}.
\begin{restatable}[Classification of Boolean CSPs over Binary Families]{theorem}{BinaryCSPClassification}\label{thm:binary-csp-classification}
 Let $\Fam$ be a family of binary Boolean constraints. Then, for any $1\le \gamma\le 2$ we obtain the following classification:
    \begin{enumerate}
        \item (\emph{Linear Regime}) If $\Fam$ contains neither binary $\NAND$ constraint, nor binary implication $\IMPL$, then $\CSP_k^\gamma(\Fam)$ can be solved in linear time $\mathcal O\left( m + n\right)$.
        
        \item (\emph{Subexponential Regime})~\cite{KunnemannM20} If $\Fam$ contains $\IMPL$, but not $\NAND$, then $\CSP_k^\gamma(\Fam)$ can be solved in time $\mathcal{O}\left(n^{4\sqrt k}\right)$, furthermore, unless the $k$-clique hypothesis fails, there is no algorithm solving $\CSP_k^\gamma(\Fam)$ in time $\mathcal O\left( n^{ \omega \sqrt[3]k/6+c-\varepsilon}\right)$, for any $\varepsilon>0$ and some $c>0$ independent of $k$.
        
        \item (\emph{$k$-IS Regime})  If $\Fam = \{\NAND\}$, then for any $\gamma<2$, every $\CSP_k^\gamma(\Fam)$ instance is trivial. For $\gamma = 2$, there is an algorithm solving $\CSP_k^\gamma(\Fam)$ in time $\mathcal O\left( n^{ \omega k/3}\right)$ (for any $k$ divisible by $3$). Furthermore, any algorithm running in $\mathcal O\left( n^{ \omega k/3-\varepsilon}\right)$ would refute the $k$-Clique Hypothesis.
        
        \item (\emph{Clique Regime}) If for some nonempty family $\Fam'$, $\Fam$ can be written as $\Fam = \{\NAND\} \cup \Fam'$, then we can solve $\CSP_k^\gamma(\Fam)$ in time $\mathcal O\left( m^{ \omega (k - c_{\Fam'})/6 + 1}\right)$ (for all sufficiently large $k$ divisible by $3$), where $c_{\Fam'}\in \{0, 1, 2\}$ depends only on $\Fam'$. Moreover, any algorithm running in $\mathcal O\left( m^{ \omega (k - c_{\Fam'})/6-\varepsilon}\right)$ would refute the $k$-Clique Hypothesis.
    \end{enumerate}
\end{restatable}

The algorithm for the Clique Regime turned out to be the most interesting.
We essentially argue that for each family $\Fam$ that belongs to the clique regime we can reduce any instance of $\CSP_k^\gamma(\Fam)$ efficiently to $\CSP_k^\gamma(\{\NAND, \IMPL\})$, so it suffices to construct an efficient algorithm for this particular family.
Our algorithm proceeds in a few steps that when implemented carefully, yield the desired running time.

\noindent (1) (\emph{Reduction to Restricted Instance}) We first leverage the directed reachability structure induced by $\IMPL$ constraints: intuitively, vertices with many descendants are cheap to guess, since including such a vertex in a solution forces many additional variables to be set as well.
By systematically branching on these vertices and propagating their implications, we can reduce to a \emph{restricted} instance, where each vertex has at most a single descendant (besides itself) in the directed graph induced by $\IMPL$ constraints.

\noindent (2) (\emph{Grouping and Sparsity Cutoff}) For such restricted instances, we introduce a grouping technique to partition the variables according to their local $\IMPL$ structure.
Each group corresponds to a small implication neighborhood that can be treated as a single unit.
We then establish a \emph{sparsity cutoff}: if any group is sufficiently large and therefore sufficiently sparse with respect to $\NAND$ constraints, we immediately find a solution using Turan-like arguments, and are done. 
As a consequence, we can bound the size of each group and the total number of groups, reducing the problem to a \emph{core instance}.

\noindent (3) (\emph{Removing the $2$-Cycles and Reduction to Triangle Detection}) For our last step, we want our Implication-induced graph to be acyclic, however, our core instance still may contain cycles of size $2$. 
We then prove that any such cycle can only be locally contained within a small subset of vertices, so by a bounded guessing step, we can efficiently eliminate all cycles from our graph. 
After eliminating all cycles, we show how to distribute the groups in a balanced manner and encode valid group selections as vertices in an auxiliary graph. This construction reduces the problem to Triangle Detection, allowing us to leverage matrix multiplication to achieve the desired running time.

%% Timo: intendet for journal version
%\input{sections/journal_version/triviality_cutoff_for_tech_intro}

\section{Preliminaries}
For any natural number $x \in \mathbb{N}$, we denote with $[x]$ the set of integers $\{1, \dots, x\}$. 
Further, for any set $S$ and any number $d \in \mathbb{N} \cup \{0\}$ we denote with $\binom{S}{d}$ the set of all subsets of $S$ of size $d$. With $\binom{S}{\leq d}$ we denote the set of all subsets of $S$ of size $\leq d$. 
An \emph{$h$-hypergraph} is any hypergraph $H = (V,E)$ such that $E\subseteq \binom{V}{\le h}\setminus\binom{V}{\leq 1}$. 
An \emph{$h$-uniform hypergraph} is any $h$-hypergraph with $E\subseteq \binom{V}{h}$. 

\subsection{Hardness Assumptions}
The $k$-clique problem asks, given a graph $G=(V,E)$ on $n$ vertices, to determine whether there exists a $k$-clique $C \subset V$ of size $k$ such that $\binom{C}{2} \subseteq E$.
Dating back to~\cite{cliqueFast}, one can detect $k$-cliques time $\O(n^{\frac{\omega k}{3}})$ for $k$ divisible by $3$. Here, $\omega < 2.372$ denotes the matrix multiplication exponent, such that we can multiply two $n\times n$ matrices in time $\O(n^\omega)$.

This problem generalizes to $r$-uniform hypergraphs, asking to find a set of vertices such that every $r$ tuple is contained in a hyperedge -- but algebraic approaches such as fast matrix multiplication fail in this setting (for a discussion of this phenomenon, see~\cite{DBLP:conf/soda/LincolnWW18}).
As such, no algorithm with substantial improvements over brute-force time $n^{k - o(1)}$ is known (see~\cite{AbboudFS24} for subpolynomial-factor improvements). 
Even more so, polynomial improvements for $h$-uniform $k$-hyperclique would result in improved algorithms for problems such as $\textsc{Max-$h$-SAT}$~\cite{Wil07} that are widely believed to be hard. 

In fine-grained complexity theory, the $k$-clique problem has been widely used to obtain conditional lower bounds,
such as e.g.~\cite{AbboudBW18,chang:LIPIcs.CPM.2016.13,DalirrooyfardW22}.
Over the recent years, the $r$-uniform $k$-hyperclique problem has seen increasing popularity as well, 
yielding applications for a variety of problems~\cite{DBLP:conf/soda/LincolnWW18,AbboudBDN18,BringmannFK19,KunnemannM20,BringmannS21,DalirrooyfardJW22,MathialaganWX22,FischerKRS25,FuLR25}.

As such we use the following hypothesis concerning $k$-(hyper)clique detection:
\begin{hypothesis}[$r$-Uniform $d$-Hyperclique Hypothesis]
    Let $\epsilon > 0$ and $k > d$.
    \begin{enumerate}
        \item For $d = 2$ there is no $\O(n^{({\omega k}/{3}) - \varepsilon})$-time algorithm detecting a $k$-clique in a graph (also referred to as $k$-clique hypothesis).
        \item For $d\geq 3$, there is no $\O(n^{k - \varepsilon})$-time algorithm detecting a $k$-clique in a $d$-uniform hypergraph.
    \end{enumerate}
\end{hypothesis}

\section{\texorpdfstring{$k$}{k}-Independent Set in Non-Uniform Hypergraphs}
This is our main technical section. 
It is dedicated to constructing algorithms and lower bounds for detecting a $k$-independent set in (non-uniform) $h$-hypergraphs.\footnote{I.e., each hyperedge has arity \emph{at most} $h$.}
Note that throughout this section we treat edges as subsets of vertices and use standard set-theory notation (e.g. $e\cap e'$, $e\setminus e'$, $V \subseteq \bigcup_{e\in E} e$).
We begin by considering the simplest family of $h$-hypergraphs, namely the $3$-hypergraphs. 

\subsection{\texorpdfstring{$k$}{k}-Independent Set in \texorpdfstring{$3$}{3}-Hypergraphs}\label{sec:k-IS-3-hypergraphs}
In this section, we prove the following two main theorems.
\begin{theorem}[Algorithm for $k$-Independent Set Problem in $3$-Hypergraphs]\label{thm:23IS-algo}
    Given any $3$-hypergraph $H$ with $n$ vertices and $m$ edges, there is an algorithm deciding whether $H$ contains a $k$-independent set in time $\mathcal O\left(n^{\frac{k\omega}{3}} + m^{\frac{k}{3}}\right)$  for all $k$ divisible by $3$.
\end{theorem}
Moreover, we show that this running time is essentially optimal, unless at least one of the two established hypotheses fails.
\begin{theorem}[Conditional Lower Bounds for $k$-Independent Set Problem in $3$-Hypergraphs]
    \label{thm:23ISLB}
    There is no algorithm solving $k$-Independent Set problem in $3$-hypergraphs in time: 
    \begin{enumerate}
        \item $\mathcal O\left(n^{\frac{k\omega}{3}-\varepsilon}\right)$, unless the $k$-Clique Hypothesis fails.
        \item $\mathcal O\left(m^{\frac{k}{3}-\varepsilon}\right)$, unless the $3$-Uniform Hyperclique Hypothesis fails. Moreover, this holds even when restricting $m = \Theta(n^\gamma)$, for any $2\le \gamma\le 3$.\footnote{Recall that Turan's theorem implies that if $\gamma<2$, any instance is a trivial yes-instance (see Section \ref{sec:triviality-cutoff}). Hence, this theorem gives us a full landscape of the complexity of the problem in terms of both number of vertices and number of (hyper)edges.}
    \end{enumerate}
    
\end{theorem}

We dedicate the first part of this section to constructing our algorithm and proving Theorem \ref{thm:23IS-algo}. 
Our approach consists of first counting the \emph{potential solutions}, which are in principle all independent sets on the underlying ($2$-uniform) graph of $H$ (intuitively, in the first step we ignore arity-$3$ hyperedges), and then using inclusion-exclusion to count all \emph{false potential solutions}, by which we understand those independent sets of the underlying graph of $H$ which contain an arity-$3$ hyperedge in $H$. 
While this approach yields a correct solution, a naive implementation is unfortunately too slow. 
Intuitively, if many arity $3$-edges are clustered together, we would spend too much time counting the higher-order terms in the inclusion-exclusion formula. 
The last step of our algorithm takes care of this by making sure that we never double-count the solutions from such clustered hyperedges, assuring that we only need to compute the higher order terms of the ''nicely structured'' hyperedges, yielding the desired running time.
We start with the following simple observation.
\begin{observation}\label{obs:invalid-solutions}
    Let $H$ be a $3$-hypergraph and $G$ be the underlying graph of $H$ (obtained by removing all arity-$3$ hyperedges). 
    Let $I(G,k)$ be the set of all $k$-independent sets in $G$, and let $I_{\textnormal{invalid}}(H,k)$ be the set of all independent sets in $G$ that contain a hyperedge in $H$. 
    Then $H$ contains an independent set of size $k$ if and only if
    \[
    |I(G,k)| - |I_{\textnormal{invalid}}(H,k)| >0.
    \]
\end{observation}
Recall that the classical clique counting approach computes $|I(G,k)|$ in time $\mathcal O(n^{k\omega/3})$. 
It remains to argue that we can also compute $|I_{\textnormal{invalid}}(H,k)|$ efficiently. 
The following lemma gives us a way to compute $|I_{\textnormal{invalid}}(H,k)|$ via the standard inclusion-exclusion-based approach. 
In particular, for a subset of hyperedges $S$, we can count how many independent sets in $G$ contain $S$ and then make sure we avoid double-counting. 
For any subset of hyperedges $S$, let $I_S(H,k)$ denote the set of all independent sets of size $k$ in $G$ that contain all vertices spanned by $S$.
\begin{lemma}[Invalid Solutions via the Inclusion-Exclusion Principle]\label{lemma:inclusion-exclusion1}
    Let $H$ be a $3$-hypergraph and let $E_3$ denote the set of all arity-$3$ hyperedges in $H$.
    Then the following equality holds.
    \begin{align*}
        |I_{\textnormal{invalid}}(H,k)| = & \sum_{i=1}^{\binom{k}{3}}\sum_{S\in \binom{E_3}{i}}(-1)^{i+1}|I_{S}(H,k)|.
    \end{align*}
\end{lemma}
\begin{proof}
    The statement follows directly from the inclusion-exclusion principle, by observing that 
    \[
    I_{\textnormal{invalid}}(H,k) = \bigcup_{e\in E_3} I_{\{e\}}(H,k),
    \]
    and that for any non-empty set $S$ we have 
    \[
    I_{\textnormal{S}}(H,k) = \bigcap_{e\in S} I_{\{e\}}(H,k).\qedhere
    \]
\end{proof}
It is easy to see that naively enumerating all sets $I_S$ and computing the term $|I_S(H,k)|$ is infeasible.
However, we can observe that a lot of these sets can be seen as redundant by a more clever implementation. 
For instance, the only way that a set $S$ that contains $\binom{k}{3}$ edges is a part of an independent set of size $k$ in $G$ is that the $k$ vertices that it spans form a hyperclique in the underlying $3$-uniform hypergraph of $H$. 
However, in that case, we can observe that there is a subset $S'$ of $S$ that is of a much smaller size, in particular contains only $\frac{k}{3}$ hyperedges, that spans the same vertex set as $S$. 
Intuitively, if a particular solution contains many clustered hyperedges, by a naive implementation of our inclusion-exclusion approach, we will enumerate many edge sets $S$ that span the same set of vertices, hinting at the fact that we are doing a lot of redundant work.
In particular, this means that we only need to look at the sets $S$ that contain up to $\frac{k}{3}$ many edges, as long as we can guarantee that the independent sets spanned by these clustered hyperedges are never double counted.
We dedicate the rest of this section to formally introducing the notion of the clustered edges and showing how to avoid this redundant work.
\paragraph*{Handling intersecting edges.} 
%\mirza{Hype this part a bit...}
Let us first introduce some useful notation and terminology. 
We will introduce all of the concepts in a more general setting, namely for $h$-hypergraphs for $h\ge 3$, in order to be able to reuse it in the following sections, but note that in this section we will only focus on the special case of $h=3$.
Denote by $E_i$ the set of all hyperedges of arity $i$, and analogously, denote by $E_{\ge i}$ the set of all hyperedges of arity at least $i$.
Let $\phi:E(H)\to [m]$ be any fixed bijection. 
We use $\phi$ to obtain a total ordering of the hyperedges of $H$. That is, we write $e\prec e'$ if $\phi(e)<\phi(e')$.
Let $e,e'$ be edges that intersect (share common vertices) and assume that $e\prec e'$.
Intuitively, our goal is to efficiently remove all independent sets $X$ that contain $e$ from $I_S(H,k)$ for any $S$ that contains $e'$ to avoid the double counting of the sets spanned by the overlapping edges.
We now formally prove that this does not destroy any valid solutions.
For any set of hyperedges $S$, let $I'_S(H,k)$ be the set of all independent sets $X$ of the underlying graph $G$ of size $k$ satisfying the following two conditions.
\begin{enumerate}
    \item\label{cond:spanning} $X$ contains all the vertices spanned by the hyperedges in $S$.
    \item\label{cond:overlaps} \emph{(Handling the overlapping hyperedges)} For any edge $e\in S$ and any edge $e'\in E_{\ge3}$, such that $e\cap e' \ne \emptyset$, if $e'\prec e$, then $X$ does not contain $e'$.
\end{enumerate}
We now prove that we can safely replace the sets $I_S$ in our inclusion-exclusion formula by the sets $I'_S$.
\begin{lemma}
    Let $H$ be an $h$-hypergraph and let $E_{\ge 3}$ denote the set of all arity-$\ge 3$ hyperedges in $H$.
    Then the following equality holds.
    \begin{align*}
        I_{\textnormal{invalid}}(H,k) = \bigcup_{e\in E_{\ge3}} I'_{\{e\}}(H,k).
    \end{align*}
\end{lemma}
\begin{proof}
    We prove this by demonstrating set containment in both sides.
    Notice that one side is straightforward, as for each $e$ we have $I'_{\{e\}}\subseteq I_{\{e\}}$, and hence:
    \begin{align*}
          \bigcup_{e\in E_{\ge 3}} I'_{\{e\}}(H,k) \subseteq \bigcup_{e\in E_{\ge3}} I_{\{e\}}(H,k) = I_{\textnormal{invalid}}(H,k),
    \end{align*}
    where the last equality follows directly from definition of the set $I_{\textnormal{invalid}}(H,k)$.
    To show the other containment, let $X$ be any set in $ I_{\textnormal{invalid}}(H,k)$. 
    It suffices to show that there exists some $e$ such that $X$ is contained in $I'_{\{e\}}$.
    By definition, the subhypergraph $H[X]$ contains at least one hyperedge.
    Let $e$ be the first hyperedge contained in $H[X]$ with respect to the ordering $\prec$, i.e. for any hyperedge $e'$ contained in $X$, we have $e\prec e'$. 
    Then clearly $X$ is contained in $I'_{\{e\}}$.
\end{proof}
The idea of replacing $I$ by $I'$ in our inclusion-exclusion formula is to avoid double counting the clustered independent sets, hence we should expect that many higher order terms will vanish. Indeed, we now prove that this is indeed the case.
\begin{lemma}
    Let $H$ be an $h$-hypergraph and let $E_{\ge 3}$ denote the set of all arity-$\ge 3$ hyperedges in $H$. Let $S\subseteq E_{\ge 3}$ be any set of at least $\frac{k}{3}+ 1$ hyperedges. Then $I'_S(H,k) = \emptyset$.
\end{lemma}
\begin{proof}
    First notice that if for every pair of hyperedges $e,e'\in S$ it holds that $e\cap e' = \emptyset$, then $S$ spans at least $3(\frac{k}{3}+1) > k$ many vertices, and vacuously cannot be contained in any independent set of size $k$ in the underlying graph, hence $I'_S(H,k)\subseteq I_S(H,k) = \emptyset$.
    Now assume that $S$ contains two edges $e,e'$ that share at least one common vertex. 
    Then each independent set in $I_S(H,k)$ violates Condition \ref{cond:overlaps} (since either $e\prec e'$, or vice versa, in either case any independent set that contains both $e,e'$ violates Condition \ref{cond:overlaps}) in the definition of $I'_S(H,k)$ and hence is not contained in $I'_S(H,k)$, implying that $I'_S$ is empty.
\end{proof}
We can now formally rewrite our inclusion-exclusion formula in terms of sets $I'_S$.
\begin{corollary}\label{cor:inclusion-exclusion}
    Let $H$ be an $h$-hypergraph and let $E_{\ge 3}$ denote the set of all arity-$3$ hyperedges in $H$.
    Then the following equality holds.
    \begin{align*}
        |I_{\textnormal{invalid}}(H,k)| = & \sum_{i=1}^{\lceil\frac{k}{3}\rceil}\sum_{S\in \binom{E_{\ge3}}{i}}(-1)^{i+1}|I'_{S}(H,k)|.
    \end{align*}
\end{corollary}
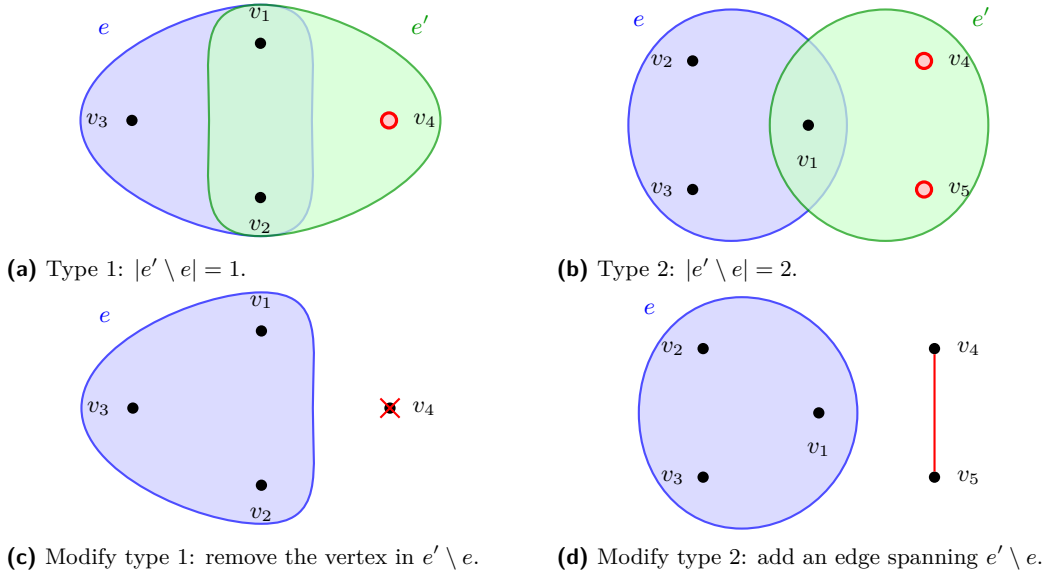
\begin{figure}[h]
    \centering
    % --- Subfigure A: Type 1 Intersection ---
    % |e| = 3, |e'| = 3. focus on e' \ e.
    \begin{subfigure}{0.48\textwidth}
        \centering
        \begin{tikzpicture}[scale=1.7]
            % Vertices
            \node[std_node] (v1) at (0, 0.6) {};
            \node[std_node] (v2) at (0, -0.6) {};
            
            % v3 is in e \ e', now set to standard
            \node[std_node] (v3) at (-1, 0) {}; 
            % v4 is the vertex in e' \ e, now highlighted
            \node[highlighted_node] (v4) at (1, 0) {};
            
            % Vertex Labels
            \node[label_style, left=0.1cm of v3] {$v_3$};
            \node[label_style, right=0.1cm of v4] {$v_4$};
            \node[label_style, above=0.1cm of v1] {$v_1$};
            \node[label_style, below=0.1cm of v2] {$v_2$};

            \begin{scope}[on background layer]
                % Hyperedge e contains {v1, v2, v3}
                \draw[hyperedge_e] 
                    plot [smooth cycle, tension=1] coordinates {(-1.4,0) (0, 0.9) (0.4, 0) (0, -0.9)};
                \node[label_style, blue, above left] at (-1.1, 0.6) {$e$};

                % Hyperedge e' contains {v1, v2, v4}
                \draw[hyperedge_ep] 
                    plot [smooth cycle, tension=1] coordinates {(1.4,0) (0, 0.9) (-0.4, 0) (0, -0.9)};
                \node[label_style, green!60!black, above right] at (1.1, 0.6) {$e'$};
            \end{scope}
        \end{tikzpicture}
        \caption{Type 1: $|e' \setminus e| = 1$.}
        \label{fig:type1}
    \end{subfigure}
    \hfill
    % --- Subfigure B: Type 2 Intersection ---
    % |e| = 3, |e'| = 3. focus on e' \ e.
    \begin{subfigure}{0.48\textwidth}
        \centering
        \begin{tikzpicture}[scale=1.7]
            % Vertices
            \node[std_node] (v1) at (0, 0) {};
            
            % v2, v3 are in e \ e', now set to standard
            \node[std_node] (v2) at (-0.9, 0.5) {};
            \node[std_node] (v3) at (-0.9, -0.5) {};
            
            % v4, v5 are in e' \ e, now highlighted
            \node[highlighted_node] (v4) at (0.9, 0.5) {};
            \node[highlighted_node] (v5) at (0.9, -0.5) {};
            
            % Vertex Labels
            \node[label_style, below=0.2cm of v1] {$v_1$};
            \node[label_style, left=0.05cm of v2] {$v_2$};
            \node[label_style, left=0.05cm of v3] {$v_3$};
            \node[label_style, right=0.1cm of v4] {$v_4$};
            \node[label_style, right=0.1cm of v5] {$v_5$};

            \begin{scope}[on background layer]
                % Hyperedge e contains {v1, v2, v3}
                \draw[hyperedge_e] 
                    plot [smooth cycle, tension=1] coordinates {(-1.4,0) (-0.6, 0.9) (0.3, 0) (-0.6, -0.9)};
                \node[label_style, blue, above left] at (-1.2, 0.7) {$e$};

                % Hyperedge e' contains {v1, v4, v5}
                \draw[hyperedge_ep] 
                    plot [smooth cycle, tension=1] coordinates {(1.4,0) (0.6, 0.9) (-0.3, 0) (0.6, -0.9)};
                \node[label_style, green!60!black, above right] at (1.2, 0.7) {$e'$};
            \end{scope}
        \end{tikzpicture}
        \caption{Type 2: $|e' \setminus e| = 2$.}
        \label{fig:type2}
    \end{subfigure}
        \begin{subfigure}{0.48\textwidth}
        \centering
        \begin{tikzpicture}[scale=1.7]
            % Vertices
            \node[std_node] (v1) at (0, 0.6) {};
            \node[std_node] (v2) at (0, -0.6) {};
            \node[std_node] (v3) at (-1, 0) {}; 
            \node[std_node] (v4) at (1, 0) {};
            
            % Red X over v4 (the difference vertex)
            \draw[red, thick, shorten <=-3pt, shorten >=-3pt] (v4.north west) -- (v4.south east);
            \draw[red, thick, shorten <=-3pt, shorten >=-3pt] (v4.north east) -- (v4.south west);
            
            % Vertex Labels
            \node[label_style, left=0.1cm of v3] {$v_3$};
            \node[label_style, right=0.1cm of v4] {$v_4$};
            \node[label_style, above=0.1cm of v1] {$v_1$};
            \node[label_style, below=0.1cm of v2] {$v_2$};

            \begin{scope}[on background layer]
                \draw[hyperedge_e] 
                    plot [smooth cycle, tension=1] coordinates {(-1.4,0) (0, 0.9) (0.4, 0) (0, -0.9)};
                \node[label_style, blue, above left] at (-1.1, 0.6) {$e$};
            \end{scope}
        \end{tikzpicture}
        \caption{Modify type 1: remove the vertex in $e' \setminus e$.}
        \label{fig:type1_mod}
    \end{subfigure}
    \hfill
    % --- Subfigure B: Type 2 Intersection ---
    % |e' \ e| = 2. We connect the two vertices in the difference.
    \begin{subfigure}{0.48\textwidth}
        \centering
        \begin{tikzpicture}[scale=1.7]
            % Vertices
            \node[std_node] (v1) at (0, 0) {};
            \node[std_node] (v2) at (-0.9, 0.5) {};
            \node[std_node] (v3) at (-0.9, -0.5) {};
            \node[std_node] (v4) at (0.9, 0.5) {};
            \node[std_node] (v5) at (0.9, -0.5) {};
            
            % Red edge between v4 and v5 (vertices in the difference)
            \draw[red, thick] (v4) -- (v5);
            
            % Vertex Labels
            \node[label_style, below=0.2cm of v1] {$v_1$};
            \node[label_style, left=0.05cm of v2] {$v_2$};
            \node[label_style, left=0.05cm of v3] {$v_3$};
            \node[label_style, right=0.1cm of v4] {$v_4$};
            \node[label_style, right=0.1cm of v5] {$v_5$};

            \begin{scope}[on background layer]
                \draw[hyperedge_e] 
                    plot [smooth cycle, tension=1] coordinates {(-1.4,0) (-0.6, 0.9) (0.3, 0) (-0.6, -0.9)};
                \node[label_style, blue, above left] at (-1.2, 0.7) {$e$};
            \end{scope}
        \end{tikzpicture}
        \caption{Modify type 2: add an edge spanning $e' \setminus e$.}
        \label{fig:type2_mod}
    \end{subfigure}
    \caption{Two types of intersections of the hyperedges and the modifications introduced by Algorithm \ref{alg:resolve-intersection}.}
    \label{fig:intersectiontypes}
\end{figure}
We now proceed to argue that we can construct the sets $I'_S$ efficiently.
\paragraph*{Constructing the sets $I'_S$}
Given two hyperedges $e,e'$ such that $e\cap e' \ne \emptyset$, we distinguish between two types of intersections between them.
We say that the intersection of two hyperedges $e,e'$ is \emph{type $i$} if (1) $e'\prec e$, and (2) $|e'\setminus e| = i$ (see Figure \ref{fig:intersectiontypes}).
The idea to construct the sets $I'_S$ is to recursively guess an arity-$3$ edge $e$ and then for all edges $e'$ in $E_3(H)$ such that the intersection $e,e'$ is type $i$, we span the vertex set $(e'\setminus e)$ by a hyperedge of arity $i$ in all the descending branches. 
This ensures that we never consider the independent sets that violate Condition \ref{cond:overlaps} in the definition of the set $I'_S$.
More formally, we consider the following algorithm.
\begin{algorithm}
\begin{algorithmic}[1]
    \Procedure{resolve-intersections}{$H, S \subset E(H)$}
    \For{$e\in S$}
        \State $H' \gets H$
        \For{$e'\in E_3(H)$}\label{line:4}
            \If{$|e\cap e'| = 0$ \textbf{or} $e\prec e'$}
                \State \textbf{continue} 
            \ElsIf{$|e'\setminus e| = 1$}
                \State $v\gets \text{the unique vertex in set } (e'\setminus e)$
                \State $H' \gets H'- v$
            \Else
                \State $e'' \gets \text{hyperedge spanning } (e'\setminus e)$
                \State $H' \gets (H' - e') \cup e''$
            \EndIf
        \EndFor 
    \EndFor
    \Return $H'$
    \EndProcedure
\end{algorithmic}
\caption{}%Algorithm for Partial $k$-Dominating Set.}
\label{alg:resolve-intersection}
\end{algorithm}
\begin{lemma}\label{lemma:resolve-intersections}
    Let $H$ be a $3$-hypergraph and let $S\subset E_3(H)$.
    Let $H'$ be the hypergraph returned by the \textsc{resolve-intersections}$(H,S)$ function in Algorithm \ref{alg:resolve-intersection}.
    Then $I'_S(H,k) = I_S(H',k)$.
\end{lemma}
\begin{proof}
    Let $X$ be any element in $I_S(H',k)$. 
    We only need to prove that $X$ satisfies the Condition \ref{cond:overlaps} in the definition of the set $I'_S$. 
    Let $e'\in E_3(H)$ be any hyperedge in $H$ such that for some $e\in S$, $e'\prec e$ and $e\cap e'\ne \emptyset$.
    We consider different cases based on the value of $|e'\setminus e|$.
    If $|e'\setminus e| = 1$, our algorithm removes the unique vertex that lies in the difference from $H'$, and hence no subset of vertices in $H'$ can contain $e'$.
    If $|e'\setminus e| = 2$, our algorithm adds an arity-$2$ edge between the two vertices $u,v$ that lie in the difference, and so no independent set of the underlying graph $G'$ of $H'$ contains $e'$. 
    The remaining two cases, namely $|e'\setminus e| = 0$ and $|e'\setminus e| \ge 3$, are trivial. 
    We can discard the former case, since it would imply that $e=e'$. 
    Similarly, the latter case implies that $e\cap e' = \emptyset$, which we assumed was not the case.
    This proves that $X\in I'_S(H,k)$, and hence $I_S(H',k)\subseteq I'_S(H,k)$.

    For the other containment, consider any element $X$ in $I'_S(H,k)$. 
    Assume for contradiction that $I_S(H',k)$ does not contain $X$. 
    In particular, this means that $X$ does not form an independent set in the underlying graph of $H'$, but it does form an independent set in the underlying graph of $H$.
    However, this is only possible if the induced subhypergraph $H[X]$ contains an edge $e'$ such that (1) $e\cap e' \ne \emptyset$ and (2) $e'\prec e$ for some $e\in S$. 
    But this violates Condition \ref{cond:overlaps} in the definition of the set $I'_S$, implying that $X\not\in I'_S(H,k)$, yielding a contradiction.
\end{proof}
We can now use Algorithm \ref{alg:resolve-intersection} as a subroutine in our main algorithm.
Recall that in a hypergraph $H$, a \emph{matching} $M$ of size $i$ is a set of $i$ pairwise non-intersecting hyperedges. Also, for simplicity, assume that $k$ is divisible by $3$, and we will handle the remaining cases later in the analysis.
\begin{comment}
\begin{algorithm}
\begin{algorithmic}[1]
    \Procedure{Invalid-Solutions}{$H, k, {\rm sign}$, count}
    \If{$k\le 0$}
            \Return count
    \EndIf
    \For{$e = \{u,v,w\}\in E_3(H)$}
        \State $H'\gets $ \textsc{resolve-intersections}$(H,e)$
        \State $G' \gets$ underlying graph of $H'$.
        \State $R \gets N_{G'}[u]\cup N_{G'}[v]\cup N_{G'}[w]$
        \State $ct \gets $ $\#k$-Independent-Sets in $G'$ that contain $\{u,v,w\}$.
        \State ${\rm count} \gets $ \textsc{Invalid-Solutions}$(H'-R, k-3, -{\rm sign }, {\rm count} + {\rm sign} \cdot ct)$
    \EndFor
    \Return count
    \EndProcedure
\end{algorithmic}
\caption{}%Algorithm for Partial $k$-Dominating Set.}
\label{alg:resolve-intersection}
\end{algorithm}
\end{comment}
\begin{algorithm}
\begin{algorithmic}[1]
    \Procedure{Invalid-Solutions}{$H, k$}
    \State count $\gets 0$
    \State $G\gets$ Underlying graph of $H$
    \For{$i \in [k/3]$}
        \State sign $\gets (-1)^{i+1}$
        \For{matching $S\subset E_{\ge 3}$ of size $i$}\label{line:5}
            \If{$|V(S)| = k$}
            \If{$S\in I(G,k)$}\label{line:6}
                \State ${\rm count}\verb!++!$
            \EndIf
            \State \textbf{continue} 
            \EndIf
            \State $H' \gets$ \textsc{resolve-intersections}$(H,S)$
            \State ${\rm count} \gets  {\rm count} + {\rm sign}\cdot |I_{S}(H',k)|$
        \EndFor
    \EndFor
    \Return count
    \EndProcedure
\end{algorithmic}
\caption{}%Algorithm for Partial $k$-Dominating Set.}
\label{alg:invalid-solutions}
\end{algorithm}
It remains to prove the correctness of this algorithm and analyse the running time.
\begin{lemma}
    Given a $3$-Hypergraph $H$ with $n$ vertices and $m$ hyperedges, Algorithm \ref{alg:invalid-solutions} returns the value $ |I_{\textnormal{invalid}}(H,k)|$ correctly in time $\mathcal O\left( n^{\frac{k\omega}{3} } + m^{\frac{k}{3}} \right)$.
\end{lemma}
\begin{proof}
    Correctness follows directly from Lemma \ref{lemma:resolve-intersections} and Corollary \ref{cor:inclusion-exclusion}.
    We now prove the upper bound on the running time. 
    Recall that there are at most $\mathcal O(m^{i})$ matchings of size $i$, and that each such matching $S$ by definition spans precisely $3i$ many vertices in $3$-hypergraphs. 
    In particular, the set $I_S(H',k)$ consists of $3i$ many ''guessed'' matching vertices and $k-3i$ many vertices outside $S$. 
    If $3i<k$, we can find the remaining $k-3i$ many vertices by running the standard matrix-multiplication algorithm on the underlying graph $G'- N_{G'}[S]$. 
    Constructing $H'$ (and consequently the underlying graph $G'$) takes at most $\mathcal O(m)$ time.
    Note that if $3i = k$, we only need to check that $S$ forms an independent set in the underlying graph $G$, which we can do in $\mathcal O(1)$ time (Line \ref{line:6}).
    This allows us to bound the total running time as follows.
    \begin{align*}
        T_k(n,m) &= \mathcal{O}\left(m^{\frac{k}{3}} + \sum_{i=0}^{\frac{k}{3}-1} m^i\left( m + n^{(k-3i)\omega/3} \right) \right)\\
        & = \mathcal{O}\left(m^{\frac{k}{3}} + \sum_{i=0}^{\frac{k}{3}-1} m^{i+1} + m^in^{(k-3i)\omega/3}  \right)\\
        & = \mathcal{O}\left(m^{\frac{k}{3}} + \sum_{i=0}^{\frac{k}{3}-1} m^in^{(k-3i)\omega/3}   \right)\\
        & = \mathcal{O}\left(\sum_{i=0}^{\frac{k}{3}} m^in^{(k-3i)\omega/3}   \right)\\
        & = \mathcal{O}\left(m^{\frac{k}{3}} +  n^{k\omega/3}   \right).
    \end{align*}
\end{proof}
This Lemma, together with Observation \ref{obs:invalid-solutions} concludes the proof of Theorem \ref{thm:23IS-algo}.
We dedicate the rest of this section towards constructing our conditional lower bounds and moreover proving Theorem \ref{thm:23ISLB}.
\paragraph*{Conditional Lower Bounds} We remark that the first conditional lower bound from Theorem \ref{thm:23ISLB} follows directly from observing that finding a $k$-independent set in a graph is a special case of finding a $k$-independent set in a $3$-hypergraph. 
We proceed to prove that for any $2\le \gamma \le 3$, and $m = \Theta(n^\gamma)$, no algorithm can detect $k$-independent sets in $3$-hypergraphs with $n$ vertices and $m$ edges in time $m^{k/3-\varepsilon}$, unless the $3$-Uniform Hyperclique Hypothesis fails. 
To this end, we construct a simple reduction from detecting a $k$-independent set in a $3$-uniform hypergraph by using a sparse embedding. 
Let $n$ be any large integer and fix any $2\le \gamma\le 3$. 
Let $N:= \lceil n^{\gamma/3}\rceil$. Let $H^* = (V^*,E^*)$ be any given $3$-uniform hypergraph. 
Let $H = (V,E)$ be the $3$-hypergraph constructed as follows. Let $V'$ be a set of $n$ vertices and set $V = V'\cup V^*$. 
Moreover, let $E = E^* \cup \{xy \mid x\in V', y\in V-x \}$, i.e. we add an edge from every vertex in $V'$ to every other vertex in $H$.
\begin{lemma}\label{lemma:3is-lower-bound}
    The $3$-Hypergraph $H$ has $\mathcal O(n)$ many vertices and $m=\mathcal O(n^\gamma)$ many edges.
    Moreover, if there exists an algorithm detecting a $k$-independent set in $H$ in time $\mathcal O(m^{\frac{k}{3}-\varepsilon})$ for some $\varepsilon>0$, then there exists an algorithm detecting a $k$-independent set in $H^*$ in time $N^{k-2\varepsilon}$.
\end{lemma}
\begin{proof}
    We first note that $H$ has
    \[\mathcal O(N+n) = \mathcal O(n^{\gamma/3}+n) = \mathcal O(n)\]
    many vertices.
    Furthermore, we can bound the number of edges in $H$ as follows.
    \begin{align*}
        |E| &= |E^*| + \mathcal{O}(n^2) \\&= \mathcal{O}(N^3 + n^2) \\&= \mathcal{O}(n^\gamma + n^2)\\
        &= \mathcal{O}(n^\gamma).
    \end{align*}
    We now prove that there is a one-to-one correspondence between $k$-independent sets in $H^*$ and $k$-independent sets in $H$. 
    We first observe that for any $k\ge 3$, any independent set in $H^*$ is also an independent set in $H$. 
    This follows directly by noticing that each edge in $E\setminus E^*$ contains at least one endpoint not in $V^*$.
    Conversely, for any $k\ge 3$, any set of $k$ vertices that contains at least one vertex from $V'$ also contains an edge (hence is not an independent set). In particular, any $k$-independent set in $H$ is also a $k$-independent set in $H^*$.
    Finally, assume that there is an algorithm detecting a $k$-independent set in $H$ in time $\mathcal O(m^{\frac{k\omega}{3}-\varepsilon})$. 
    Then we can construct this reduction and detect the $k$-independent set in $H^*$ in time
    \begin{align*}
        T & = \mathcal O(m+m^{\frac{k}{3}-\varepsilon}) \\
        &= \mathcal O(n^\gamma + n^{\gamma \frac{k}{3}-\gamma \varepsilon})\\
        &=  \mathcal O(n^{\gamma \frac{k}{3}-\gamma \varepsilon})\\
        & = \mathcal O(N^{k-\gamma \varepsilon})\\
        & =  \mathcal O(N^{k-2\varepsilon})
    \end{align*}
\end{proof}
Theorem \ref{thm:23ISLB} now follows directly.

\subsection{\texorpdfstring{$k$}{k}-Independent Set in General \texorpdfstring{$h$}{h}-Hypergraphs}
We start this section by showing that our techniques for $3$-hypergraphs can be extended to detect $k$-independent sets in general $h$-hypergraphs. 
Later on, by being somewhat more careful, we can obtain a slightly better running time and proceed to show that, perhaps surprisingly, this is as good as we can hope for, unless one of the established hypotheses in fine-grained complexity fails.
More formally, we start by proving the following proposition.
\HigherArityAlgo*
Let $E_{\ge 3}$ denote the set of all hyperedges of $H$ of arity $\ge 3$.
Recall that $I(G,k)$ denotes the set of all $k$-independent sets in the underlying graph of $H$ (i.e. in graph $G = (V, E\setminus E_{\ge 3})$), and that $I_{\text{invalid}}(H,k)$ denotes the set of all independent sets in $G$ that contain a hyperedge in $H$ (false solutions).
Also, recall that $H$ contains an independent set of size $k$ if and only if
\[
 |I(H,k)| = |I(G,k)| - |I_{\text{invalid}}(H,k)|>0.
\]
The challenge for going from $3$-hypergraphs to general $h$-hypergraphs is that in general Lemma \ref{lemma:resolve-intersections} does not necessarily hold. 
In particular, running Algorithm \ref{alg:resolve-intersection} on general graphs may introduce new hyperedges of arity larger than two (recall that in  $3$-hypergraphs, all the intersections between hyperedges are either type-1, or type-2 intersections), and hence
computing the value $|I'_S|$ is not as simple as running $k$-Independent Set algorithm on the underlying hypergraph of $H'$.
However, we show that we can circumvent this issue by a simple inductive argument.
We now prove a statement analogous to Lemma \ref{lemma:resolve-intersections} for general hypergraphs.
\begin{lemma}\label{lemma:higher-arity-I'}
    Let $H$ be an $h$-hypergraph and let $S\subset E_{\ge 3}$ be a matching in $H$. Let $H'$ be the hypergraph obtained by running Algorithm \ref{alg:resolve-intersection} on $H$ (replacing $E_3$ by $E_{\ge 3}$ in Line \ref{line:4}).
    Define the hypergraph $H''$ as follows. $V(H'') = V(H')$;$E(H'') = E(H')\setminus E_{\ge3}(H)$.
    Note that in case $h=3$, $H''$ is just the underlying graph of $H'$.
    Then 
    \[I'_S(H,k) = I(H'',k) \cap \left\{X\in \binom{V}{k} \mid V(S)\subseteq X\right\}.\]
\end{lemma}
The proof follows the same general approach as the proof of Lemma \ref{lemma:resolve-intersections}.
\begin{proof}
    Write $\mathcal T := I(H'',k) \cap \left\{X\in \binom{V}{k} \mid S\subseteq X\right\}$. Let $X\in \mathcal T$ be arbitrary. Clearly (by definition of $\mathcal T$) $X$ satisfies Condition \ref{cond:spanning} in the definition of the set $I'_S$. 
    We prove that it also satisfies Condition \ref{cond:overlaps}. 
    Let $e'\in E_{\ge3}(H)$ be any hyperedge in $H$ such that for some $e\in S$, $e'\prec e$ and $e\cap e' \ne \emptyset$.  
    We consider different cases based on the value of $|e'\setminus e|$.
    If $|e'\setminus e| = 1$, our algorithm removes the unique vertex that lies in the difference from $H'$, and hence no subset of vertices in $H'$ (hence also $H''$) can contain $e'$.
    If $|e'\setminus e| \ge 2$, our algorithm adds a hyperedge between the vertices spanning the set $e'\setminus e$ in $H'$. 
    Note that this newly added hyperedge is also present in $H''$ by construction, and so no independent set of the hypergraph $H''$ contains $e'$. 
    This proves that $X\in I'_S(H,k)$ and hence $\mathcal T \subseteq I'_S(H,k)$. 
    
    We now prove the other containment.
    Let $Y$ be any set in $I'_S(H,k)$.
    Then by Condition \ref{cond:spanning}, we have that $Y\in \left\{X\in \binom{V}{k} \mid V(S)\subseteq X\right\}$.
    It only remains to prove that $Y$ forms an independent set of size $k$ in $H''$.
    Assume for contradiction that this is not the case. 
    In particular, since $Y\in I'_S(H,k)$, this means that $Y$ forms an independent set in $H$, but not in $H''$. 
    Note that this is only possible if the induced subhypergraph $H[Y]$ contains an edge $e'$ such that (1) $e\cap e' \ne \emptyset$, and (2) $e'\prec e$ for some $e\in S$.
    But this violates Condition \ref{cond:overlaps} in the definition of $I'_S$, implying that $Y\not\in I'_S(H,k)$, yielding a contradiction.
\end{proof}
We now use induction on the arity of the graph to prove that we can obtain the running time of Proposition \ref{prop:higher-arity-algo}. We start by proving the following simple structural properties of graph $H''$.
\begin{observation}\label{obs:properties-of-h''}
    Let $H$ be an $h$-hypergraph ($h\ge 3$) with $n$ vertices and $m$ edges.
    Let $H''$ be the hypergraph constructed as above.
    Then $H''$ satisfies the following conditions.
    \begin{itemize}
        \item Any hyperedge in $H''$ spans $<h$ many vertices (i.e. $H''$ is an $h''$-hypergraph with $h''<h$).
        \item $|E(H'')|\le m$.
    \end{itemize}
\end{observation}
\begin{proof}
    Any hyperedge in $H''$ of arity $\ge 3$ is not present in $H$, hence was added during Algorithm \ref{alg:resolve-intersection}. 
    Note that whenever we added one such hyperedge, we also remove one hyperedge from $H$, hence the second property follows directly. 
    Also, each hyperedge added during this algorithm is a strict subset of a hyperedge in $H$, implying the first property as well.
\end{proof}
We are now ready to prove Proposition \ref{prop:higher-arity-algo}.
\begin{proof}[Proof of Proposition \ref{prop:higher-arity-algo}]
    Let $H$ be an $h$-hypergraph. We proceed by strong induction on $h$.
    If $h\le 3$, the claim follows directly from Theorem \ref{thm:23IS-algo}.
    Assume $h>3$.
    We now prove that if we can detect a $k$-independent set in any $(\le h-1)$-hypergraph $H^*$ in time $O\left(\min\left\{n^k, n^{\frac{k\omega}{3}} + m^{\frac{k}{3}}\right\}\right)$, then we can detect a $k$-independent set in $H$ in the same time.
    Using Corollary \ref{cor:inclusion-exclusion} and Lemma \ref{lemma:higher-arity-I'}, it is sufficient to show that for any set $S\subseteq E_{\ge3}$ that induces a matching in $H$, we can compute the value $|I(H'',k) \cap \left\{X\in \binom{V}{k} \mid V(S)\subseteq X\right\}|$ efficiently.
    Let $S$ be a matching of size $i$. In particular $|V(S)|\ge 3i$. 
    Notice that by construction of hypergraph $H''$, we have 
    \[I(H'',k) \cap \left\{X\in \binom{V}{k} \mid V(S)\subseteq X\right\} = I(H''-N_{G''}(S), k-|V(S)|)\]
    where $G''$ is the underlying graph of $H''$. 
    Moreover, since $H''-N_{G''}(S)$ is a subhypergraph of $H''$, by Observation \ref{obs:properties-of-h''}, it has at most $m$ edges and arity at most $(h-1)$.
    Hence, by the induction hypothesis, we can compute $|I(H''-N_{G''}(S), k-|V(S)|)|$ in time at most $\mathcal O\left(n^{\frac{(k-3i)\omega}{3}} + m^{\frac{k-3i}{3}}\right)$.
    We can thus compute $I_{\text{invalid}}(H,k)$ (by a simple modification of Algorithm \ref{alg:invalid-solutions} as described above) in time 
\begin{equation}\label{eq:running-time-analysis-prop}
\begin{aligned}
    T_k(n,m)
    &= \mathcal O\left(m^{\frac{k}{3}} + \sum_{i=0}^{k/3-1} m^i \left(n^{\frac{(k-3i)\omega}{3}} + m^{\frac{k-3i}{3}}\right)\right)\\
    &= \mathcal O\left(m^{\frac{k}{3}} + \sum_{i=0}^{k/3-1} m^i n^{\frac{(k-3i)\omega}{3}} \right)\\
    &= \mathcal O\left( \sum_{i=0}^{k/3} m^i n^{\frac{(k-3i)\omega}{3}} \right)\\
    &= \mathcal{O}\left(m^{\frac{k}{3}} + n^{k\omega/3}\right).
\end{aligned}
\end{equation}
    Moreover, whenever $m = \Omega(n^3)$, we can use a simple brute-force algorithm to obtain an $\mathcal O(n^k)$ algorithm.
\end{proof}
Moreover, by noticing that for any $h\ge 3$, detecting a $k$-independent set in $h$-hypergraphs is a generalization of detecting a $k$-independent set in $3$-hypergraphs, hence the lower bound from Theorem \ref{thm:23ISLB} naturally transfers to this setting (whenever $2\le \gamma\le 3$) as well, showing conditional optimality of our algorithm.
However, we notice that our algorithm analysis is nevertheless a bit wasteful sometimes (e.g. we bound the size of $V(S)$ by $3i$, whereas if $S$ consists of large hyperedges, it might be way larger), and a natural question that arises is if the total number of edges is a correct parameter to measure the effect of sparsity on this problem. 
A natural set of parameters could perhaps be $m_i$ for all $3\le i\le h$, where each $m_i$ denotes the number of hyperedges of arity $i$.
We remark that none of our conditional lower bounds exclude the possibility of an algorithm running in time e.g. $\mathcal O\left(n^{k\omega /3} + \sum_{i=3}^h m_i^{k/i}\right)$.
So a natural question to ask is if such an algorithm is indeed possible. 
A simple modification of the reduction from Theorem \ref{thm:23ISLB} shows that, while a small improvement is indeed possible, we do not hope to achieve such a running time. 
To that end, we obtain the following distinct lower bounds, and we later show that a simple, slightly more clever analysis of our algorithm implies that we can in fact match each of them.
\MixedArityLowerBound*
We note that the clique lower bound is inherited from the previous section, and we prove the remaining two lower bounds.
In order to prove the last bound, we reduce from the independent set detection in \emph{unbalanced $k$-partite $h$-uniform hypergraphs}. The hardness of this problem as stated by the following lemma follows a rather standard split-and-list approach (see e.g. \cite{bringmann2015quadratic,fischer2024effect,kunnemann2024fine} for analogous proofs.
\begin{lemma}[Hardness of Colorful $k$-Independent Set in Unbalanced Hypergraphs]
    Given a $k$-partite $r$-uniform hypergraph $H = (V_1\cup \dots \cup V_k, E)$, there is no algorithm detecting if there is an independent set $x_1\in V_1, \dots, x_k\in V_k$ in time $\mathcal O\left(\left(\prod_{i=1}^k|V_i|\right)^{1-\varepsilon}\right)$, unless the $r$-Uniform Hyperclique Hypothesis fails.
\end{lemma}
Consider now the following construction. 
Let $3\le i\le h$ be arbitrary and let $m_i = \Theta(n^{\gamma_i})$ for arbitrary $\gamma_i\ge 3$.
Let $r = \lfloor \gamma_i\rfloor$. 
We reduce from the $k$-independent set detection in unbalanced $r$-uniform hypergraph as follows. 
Let $H^* = (V_1\cup \dots \cup V_k,E^*)$ be an arbitrary $k$-partite $r$-uniform hypergraph with $|V_1|=\dots = |V_{k-1}| = n$ and $|V_k| = \lfloor n^{\gamma_i-r}\rfloor$. 
\begin{lemma}\label{lemma:h-is-lower-bound-mn-construction}
    Let $H^*$ be as above. There exists an $h$-hypergraph $H$ satisfying the following conditions.
    \begin{enumerate}
        \item $H$ has $\mathcal O(n)$ vertices and $\mathcal O(m_i)$ hyperedges of arity $i$.
        \item $H$ contains an independent set of size $(k+i-r-1)$ if and only if $H^*$ contains vertices $x_1\in V_1, \dots, x_k\in V_k$ that form an independent set.
        \item $H$ can be constructed in time $\mathcal{O}(n+m_i)$.
    \end{enumerate}
\end{lemma}
\begin{proof}
    Let $V(H) = V_1\cup \dots \cup V_k\cup \{d_1,\dots, d_{i-r-1}\}$ and let $E^*_{(k)}$ be the set of hyperedges of $H^*$ with an endpoint in $V_k$. We now construct the hyperedges of $H$.
    \[
    E(H) = E^*_{(k)} \cup \bigcup_{i=1}^k\binom{V_i}{2} \cup \left\{\{\underbrace{v_k}_{\in V_k},\underbrace{d_1,\dots, d_{i-r-1}}_{\text{dummy nodes}}, \underbrace{v_{j_1}, \dots, v_{j_r}}_{\text{original edge}}\}\mid \{v_{j_1}, \dots, v_{j_r}\}\in E^*, \quad v_k\in V_k \right\}
    \]
    That is, we preserve all the hyperedges from $H^*$ that have an endpoint in $V_k$, and for all other edges we create $|V_k|$ many edges of arity $i$, each obtained by extending the original edge by a vertex $v_k$ from $V_k$, and the dummy nodes $d_1,\dots, d_{i-r-1}$. 
    The remaining edges of arity-$2$ are obtained by making each part $V_i$ of $H^*$ into a clique, in order to guarantee that any independent set in $H$ contains at most one vertex from each $V_i$.
    We first prove that $H$ has the desired number of vertices and edges. 
    It is straightforward to see that $|V(H)| = \mathcal O(kn + n^{\gamma_i-r}) = \mathcal O(n)$. 
    Also, notice that \[|E^*_k|\le \mathcal O(n^{r-1}|V_k|) = \mathcal O(n^{r-1+\gamma_i-r}) = \mathcal O(n^{\gamma_i-1}).\]
    We only have to prove that $H$ contains $\mathcal O(m_i)$ edges of arity $i$ to show both first and the third condition.
    Let $E_i$ denote the set of edges of arity $i$ in $H$. Then we can bound it as follows.
    \begin{align*}
        |E_i| &= \mathcal{O} \left(|V_k| \cdot n^r\right) = \mathcal O(n^{r+\gamma_i-r}) = \mathcal O(n^{\gamma_i}).
    \end{align*}
    It remains to prove the second property. We notice that in $H$ any independent set of size $(k+i-r-1)$ must contain \emph{all} vertices $d_1,\dots, d_{i-r-1}$, since edges $\binom{V_i}{2}$ ensure that any independent set contains at most one vertex from each $V_i$. 
    We prove that $v_1\in V_1,\dots, v_k\in V_k$ is an independent set in $H^*$ if and only if $v_1,\dots, v_k, d_1,\dots, d_{i-r-1}$ is an independent set in $H$. 
    Assume first that $S:=\{v_1,\dots, v_k, d_1,\dots, d_{i-r-1}\}$ is an independent set in $H$ such that (without loss of generality) $v_1\in V_1,\dots, v_k\in V_k$, and assume for contradiction that $v_1,\dots, v_k$ contains a hyperedge $\{v_{j_1},\dots, v_{j_r}\}$ in $H^*$ (hence does not form an independent set).
    Note that this edge cannot be in $E^*_{(k)}$, since this edge set is also contained in $H$. 
    But then consider the edge $\{v_k, d_1,\dots, d_{i-r-1},v_{j_1},\dots, v_{j_r}\}$ in $H$. 
    Clearly each of the vertices of this edge is contained in $S$, and hence $S$ cannot be an independent set, a contradiction.

    Conversely, let $v_1\in V_1,\dots, v_k\in V_k$ form an independent set in $H^*$. We prove that $S:=\{v_1,\dots, v_k, d_1,\dots, d_{i-r-1}\}$ is an independent set in $H$.
    This follows by observing that any edge in $H^*$ is embedded into an edge in $H$, hence if $\{v_1,\dots, v_k, d_1,\dots, d_{i-r-1}\}$ contained an edge, then the corresponding edge would be present in $H^*$ as well.
\end{proof}
We can now prove the third lower bound from Theorem \ref{thm:mixed_arity_lb}.
\begin{lemma}
    For no $i$ such that $\gamma_i\ge 3$ is there an algorithm running in $\mathcal{O}\left(m_in^{{k-i-\varepsilon}}\right)$, unless the $\left(\lfloor \gamma_i\rfloor\right)$-Uniform Hyperclique Hypothesis fails.
\end{lemma}
\begin{proof}
    Assume that such an algorithm exists. 
    Given a hypergraph $H^*$ as above, construct the graph $H$ as in the proof of the last lemma and run the algorithm on $H$ to detect a $k':=(k+i-r-1)$-independent set in time $\mathcal O(m_in^{k'-i-\varepsilon})$ to detect the independent set in $H^*$ in time:
    \begin{align*}
        T_k &= \mathcal{O}\left(m+n+m_in^{k'-i-\varepsilon}\right) \\ 
        & = \mathcal O\left(m_in^{(k+i-r-1)-i-\varepsilon}\right)\\
        & =  \mathcal O\left(m_in^{(k-r-1)-\varepsilon}\right)\\
        & = \mathcal O\left(n^{\gamma_i}n^{k-r-1-\varepsilon}\right)\\ 
        & = \mathcal O\left(n^{k-1} n^{\gamma_i-r-\varepsilon}\right) \\
        & = \mathcal{O}\left((|V_1|\dots |V_{k-1}| |V_k|)^{1-\varepsilon}\right)
    \end{align*}
    refuting the $r$-Uniform Hyperclique Hypothesis.
\end{proof}
We now proceed to show the $3$-Uniform Hyperclique based conditional lower bound as stated in Theorem \ref{thm:mixed_arity_lb}.
We remark that this construction stems from a rather straightforward combination of ideas from the proofs of Lemma \ref{lemma:3is-lower-bound} and Lemma \ref{lemma:h-is-lower-bound-mn-construction}, and we will only sketch the proof.
Let $3\le i\le h$ be arbitrary and let $m_i = \Theta(n^{\gamma_i})$ for arbitrary $2\le \gamma_i\le 3$. 
We reduce from the (colorful) $k$-independent set detection in $k$-partite $3$-uniform hypergraph as follows. 
Let $H^* = (V_1\cup \dots \cup V_k,E^*)$ be an arbitrary $k$-partite $3$-uniform hypergraph with $|V_1|=\dots = |V_{k}| = m_i^{1/3}$. 
\begin{lemma}\label{lemma:h-is-lower-bound-m-construction}
    Let $H^*$ be as above. There exists an $h$-hypergraph $H$ satisfying the following conditions.
    \begin{enumerate}
        \item $H$ has $ \Theta(n)$ vertices and $\mathcal O(m_i) = O(n^{\gamma_i})$ edges of arity $i$.
        \item $H$ contains an independent set of size $(k+i-3)$ if and only if $H^*$ contains vertices $x_1\in V_1, \dots, x_k\in V_k$ that form an independent set.
        \item $H$ can be constructed in time $\mathcal{O}(n+m_i)$.
    \end{enumerate}
\end{lemma}
\begin{proof}[Proof sketch]
    Let $V(H) = V_1\cup \dots \cup V_k\cup \{d_1,\dots, d_{i-3}\}\cup \{p_1,\dots, p_n\}$ and construct the edges of $H$ as follows.
    \[
    E(H) = \left\{\{p_i,v\} \mid i\le n, v\in V(H)\right\} \cup \bigcup_{i=1}^k \binom{V_i}{2}\cup \left\{\{\underbrace{d_1,\dots, d_{i-3}}_{\text{dummy nodes}}, \underbrace{v_{j_1}, v_{j_2}, v_{j_3}}_{\text{original edge}}\}\mid \{v_{j_1}, v_{j_2}, v_{j_3}\}\in E^*\right\}.
    \]
    In particular, vertices $p_i$ make sure that we have $\Theta(n)$ many vertices (sparse embedding) and moreover, we add an edge of arity two between each $p_i$ and every other vertex to make sure that no independent set of size $\ge 2$ contains a $p_i$ vertex.
    Besides this, we note that the construction is very similar to that of Lemma \ref{lemma:h-is-lower-bound-mn-construction}, and all of the properties of the hypergraph $H$ can be proved completely analogous to the proof of that lemma.
\end{proof}
Theorem \ref{thm:mixed_arity_lb} now follows directly.
Finally, we prove that the running time of the algorithm from Proposition \ref{prop:higher-arity-algo} can be improved to match the lower bounds from Theorem \ref{thm:mixed_arity_lb} when we parameterize by the parameters $m_i$.
More precisely, we prove the following theorem.
\MixedArityAlgo*
The strategy towards proving this is to design a heavy-light approach on hyperedges of different arities. 
Specifically, we start similarly as in the previous algorithms by counting all of the independent sets in the underlying graph and then the goal is to subtract the count of the ''false solutions'', that is the independent sets in the underlying graph that contain an edge in the original hypergraph.
To that end we partition the set $[h]$ into two sets $S,D$ (standing for sparse and dense), where $S$ consists of all those $i$ such that $m_i = O(n^{3})$ and $D$ consists of all other $i$'s.
We then use the approach similar to that of Proposition \ref{prop:higher-arity-algo} to efficiently count the false solutions in ''sparse'' underlying hypergraphs, and finally, we show that we can also efficiently enumerate the false solutions in ''dense'' underlying hypergraphs while in parallel making sure to never double count any false solution.

We start by observing that a simple brute-force approach allows us to \emph{enumerate} all $k$-independent sets in the underlying graph that contain a hyperedge of arity $i$.
\begin{observation}\label{lemma:enum-false-sol}
    Let $H$ be an $h$-hypergraph with $n$ vertices and $m_i = \Theta(n^{\gamma_i})$ edges of arity $i$ for each $2\le i\le h$.
    Let $G$ be the underlying graph of $H$.
    Then there exists an algorithm \emph{enumerating} all $k$-independent sets in $G$ that contain a hyperedge of arity $i$ in $H$ in time $\mathcal{O}(m_in^{k-i})$.
\end{observation}
\begin{proof}
     We simply enumerate all sets of size $k$ that contain a vertex of arity $i$ (at most $\mathcal O(m_in^{k-i})$ of them) and for each in constant time check if it contains an independent set in $G$. If it does not, continue, if it does, subtract $1$ from the total count. 
    It is straightforward to verify that this procedure counts precisely $k$-independent sets in $G$ that contain an edge of arity $i$ in $H$.
\end{proof}
Finally, we are missing one more tool before proving our theorem, namely the refinement of the algorithm analysis from Proposition \ref{prop:higher-arity-algo}.
We first need to make a few simple observations.
First, recall that so far we had an arbitrary ordering on the set of hyperedges $\prec$. 
We now want to enforce some structure on the ordering, and in particular we want to make it monotone in arity of the hyperedge. 
More specifically, let $\prec$ be any total ordering on $E(H)$ such that whenever two edges $e,e'\in E(H)$ satisfy $|V(e)|<|V(e')|$, we have that $e\prec e'$.
We now bound the number of hyperedges in the graph $H''$ using the new ordering. 
The following lemma is a refinement of Observation \ref{obs:properties-of-h''}.
\begin{lemma}\label{lemma:properties-of-h''}
    Let $\prec$ be the total ordering of the hyperedges of $H$ as above.
    Let $S\subseteq E_{\ge 3}(H)$ be a set of hyperedges such that each hyperedge has arity at most $i$. 
    Let $H'$ and $H''$ be as before (i.e. $H'$ stems from running Algorithm \ref{alg:resolve-intersection} on $H,S$ and $H'' = H'- E_{\ge 3}(H)$).
    Then $|E_{\ge3}(H'')| \le \sum_{j=3}^i m_j$.
\end{lemma}
\begin{proof}
    Any hyperedge in $H''$ of arity $\ge 3$ is not present in $H$ and was added during Algorithm \ref{alg:resolve-intersection}. 
    In particular, any such hyperedge stems from the intersection of some $e\in S$ and some $e'\in E(H)$ such that $e'\prec e$. 
    However, since $e'\prec e$, this means that the arity of $e'$ is at most arity of $e$, which is at most $i$ since $e\in S$. 
    The bound on the number of hyperedges in $H''$ now follows easily.
\end{proof}
\begin{lemma}
    Given any $h$-hypergraph $H$ with $n$ vertices and $m_i = \Theta(n^{\gamma_i})$ hyperedges of arity $i$ for each $2\le i\le h$, there exists an algorithm computing the number of $k$-independent sets in the underlying graph $G$ of $H$ that contain a hyperedge in $H$ (false solutions) in time 
    \[
    \mathcal O \left(n^{k\omega/3} + \sum_{i=3}^h m_i^{\frac{k-i+3}{3}}\right).
    \]
\end{lemma}
\begin{proof}
    We run the same algorithm from the proof of Proposition \ref{prop:higher-arity-algo}, but analyse in terms of $m_i$'s this time, rather than in terms of $m$.
    Note that the construction and correctness is already argued there, so we only prove the running time, namely we refine the Equation~\ref{eq:running-time-analysis-prop}. 
    Recall that in Equation~\ref{eq:running-time-analysis-prop} we bounded the number of vertices spanned by the matching $S$ by $3|S|$. 
    However, this is too coarse for our current use case, and we rather look at the arity of the largest hyperedge contained in $S$. 
    More specifically, we notice that any matching $S\subseteq E_{\ge3}$ of size $j$ that contains an edge of arity $i$ spans at least $i+3(j-1)$ many vertices. Hence, we can bound the running time of the algorithm as follows.
    We will also assume that $S$ does not contain any edges of arity greater than $i$ (i.e. we first guess the edge of the largest arity and then guess the remaining $j-1$ edges).
    \begin{equation}
        \begin{aligned}\label{eq2-}
            T_k(n,m)
            &= \mathcal O\left(\sum_{i=3}^h\sum_{j=0}^{(k-i)/3} m_i\cdot m^{j} \left(n^{\frac{(k-i-3j)\omega}{3}} + m^{\frac{k-i-3j}{3}}\right)\right)
        \end{aligned}
    \end{equation}
    Before bounding this equation further, let us first take a minute to digest it.
    The intuition is that in the inner-most term, we first guess a matching of size $j+1$ for some $j\ge 0$ that contains a hyperedge of arity $i$ for some $i\ge 3$, and the remaining $j$ hyperedges of arity at most $i$. 
    This contributes the term $m_i\cdot m^{j}$ ($m_i$ is the time required to guess the edge of arity $i$ and there are at most $\mathcal O(m^j)$ matchings of size $j$ that can potentially form a matching of size $j+1$ with the guessed edge). 
    We then use the same induction-based approach from Proposition \ref{prop:higher-arity-algo} to show that we can compute the value $|I(H''-N_{G''}(S), k-|V(S)|)|$ in time at most $\mathcal O\left(n^{\frac{(k-i-3j)\omega}{3}} + m^{\frac{k-i-3j}{3}}\right)$, giving us the desired value. 
    We finally simply sum over all possible values of $i,j$ to cover all cases and get the total bound of the running time of our algorithm.
    
    Note that we can do even slightly better than this simple bound.
    Particularly, for any $i\ge 3$, let $M_i = \sum_{j=3}^i m_j$.
    Now, recall that by our assumption all of the edges in our matching have arity at most $i$.
    Also, note that we are applying our induction on hypergraph $H''$, and second $m$ is the trivial bound on the hyperedges in $H''$ as stated by Observation \ref{obs:properties-of-h''}.  
    However, using the previous lemma, we know that there is a better bound on this number that can be obtained by exploiting the monotonicity of our ordering $\prec$ w.r.t. arity. 
    This means that in principle we can replace both $m$'s in our bound with $M_i$'s.
    We can now show that this bound evaluates to what we want.
    \begin{align*}
        T_k(n,m)
        &= \mathcal O\left(\sum_{i=3}^h\sum_{j=0}^{(k-i)/3} m_i\cdot M_i^{j} \left(n^{\frac{(k-i-3j)\omega}{3}} + M_i^{\frac{k-i-3j}{3}}\right)\right) & \text{\small{(Eq~\ref{eq2-} + arguments above)}}\\
        &= \mathcal O\left(n^{\frac{k\omega}{3}} + \sum_{i=3}^h\sum_{j=0}^{(k-i)/3} m_i\cdot \left(M_i^{\frac{k-i}{3}}\right)\right) & \\
        &= \mathcal O\left(n^{\frac{k\omega}{3}} + \sum_{i=3}^h m_i\cdot \left(M_i^{\frac{k-i}{3}}\right)\right) \\
        & = \mathcal O\left(n^{\frac{k\omega}{3}} + \sum_{i=3}^h m_i\cdot \left(\sum_{j=3}^i m_j\right)^{\frac{k-i}{3}}\right) & \text{\small{(by definition of $M_i$)}}\\
        &= \mathcal O\left(n^{\frac{k\omega}{3}} + \sum_{i=3}^h m_i\cdot \left(\sum_{j=3}^i m_j^{\frac{k-i}{3}}\right)\right) & \text{\small{(using that $k,i\in \mathcal O(1)$)}}\\
        &= \mathcal O\left(n^{\frac{k\omega}{3}} + \left(\sum_{i=3}^h m_i^{\frac{k-i+3}{3}}\right) + \left(\sum_{\substack{3\le i\le h\\ j<i}} m_i\cdot m_j^{\frac{k-i}{3}}\right)\right)\\
        &= \mathcal O\left(n^{\frac{k\omega}{3}} + \sum_{i=3}^h m_i^{\frac{k-i+3}{3}}\right),
    \end{align*}
    where the last line follows from the fact that for each $i\ne j$, we have that $m_im_j^{\frac{k-i}{3}}\le \max\{m_j^{\frac{k-i+3}{3}}, m_i^{\frac{k-i+3}{3}}\}$, and since $j<i$, we have $m_j^{\frac{k-i+3}{3}}\le m_j^{\frac{k-j+3}{3}}$.
    Hence each of the (constantly many) terms from the second sum can be bounded by some term in the first sum.
\end{proof}
We are finally ready to prove Theorem \ref{thm:mixed-arity-algo}.
\begin{proof}[Proof of Theorem \ref{thm:mixed-arity-algo}]
    Let $S,D$ (standing for sparse and dense respectively) be a partition of $[h]$ such that $S= \{i\in [h]\mid m_i^{(k-i+3)/3} \le m_in^{k-i}\}$ and $D = [h]\setminus S$. 
    Intuitively, $S$ consists of all those indices $i$'s such that there are at most (roughly) $\mathcal O(n^3)$ many edges of arity $i$.
    Let $H_S$ be the spanning subgraph of $H$ consisting of all edges of arity $i$ with $i\in S$. 
    More formally, let $E_i(H)$ be the set of hyperedges of $H$ of arity $i$, then we define $H_S = (V, E_S)$ where $V = V(H)$ and $E_S = \bigcup_{i\in S} E_i(H)$.
    By the previous Lemma, we can compute the set $I_{\text{invalid}}(H_S, k)$ in time 
    \[\mathcal O \left(n^{k\omega/3} + \sum_{i\in S} m_i^{\frac{k-i+3}{3}}\right).\] 
    Note that this counts all the false solution that contain a hyperedge of ''sparse arity'' (i.e. hyperedge present in $H_S$). 
    
    Let $H_D$ be defined analogously, by replacing $S$ by $D$ (and additionally preserving the arity-$2$ edges, which are implicitly included in $S$).
    Now, using Observation \ref{lemma:enum-false-sol}, we can enumerate each false solution from $H_D$ in time 
    \[
    \mathcal O(\sum_{i\in D} m_i n^{k-i}).
    \]
    For each such enumerated set, we can in $\mathcal O(1)$ time check if the corresponding vertices contain an edge in $H_S$ (meaning that we already counted this solution), and if it does, we do nothing, otherwise, we increment the count by one.
    This algorithm is clearly correct and runs in the desired time.
\end{proof}

\section{Binary Constraint Families and the Effects of Sparsity}
\label{sec:binary}
% The constraint function $\NAND$ is unique in its behavior concerning sparsity, exhibiting a sharp \emph{phase transition} of hardness.
% We will see that the complexity of binary constraint families that do not contain $\NAND$ on the other hand,
% are not affected by sparsity (except for degenerate cases).
% This motivates the question to consider constraint families
% containing $\NAND$ in \emph{combination} with other constraints.
Towards a generalization of the phenomena seen with the non-uniform $k$-independent set, we recall the setting of Boolean constraint satisfaction problems, parameterized by the number of non-zeros $k$ in a satisfying assignment and the density $\gamma$:
\BooleanCSPDefinition*

The Boolean constraint satisfaction problem is a powerful generalization, capturing a vast class of problems, including the non-uniform $k$-independent set.
In this section, we will provide a classification of $\CSP_k^\gamma(\Fam)$ for \emph{binary} constraint families $\F$ -- 
as already for binary families, considering sparse instances can severely change the complexity of the given problem.
Our theorem is complete with the exception of the family $\Fam = \{\IMPL\}$ previously discussed in \cite{KunnemannM20}. For this family we extend the lower and upper bounds and show that sparsity has no impact.

\BinaryCSPClassification*
\subparagraph*{Preliminaries}
Let us recall and introduce some $\CSP$ specific notations and results.
For a singleton constraint family $\Fam = \{f\}$, we also write $\CSP_k^\gamma(f)$.
Denote by $\text{vars}(\Phi)$ the variables of an instance $\Phi$ of $\CSP_k^\gamma(\Fam)$, then we define its \emph{primal graph} $G(V,E)$ where
$V = \{x_i \mid x_i \in \text{vars}(\Phi)\}$ and $E = \{(x_{i_1},\dots,x_{i_h}) \mid f(x_1,\dots,x_{i_h}) \in \Phi\}$.
Furthermore we call a satisfying assignment $a$ for $\CSP_k^\gamma$ a \emph{solution} if it is of weight exactly $k$. 
For a Boolean constraint function $f$, we use $f(\mathbf{x})$ with $\mathbf{x} = (x_1,\dots,x_h)$ as a shorthand for $f(x_1,\dots,x_h)$.
Finally, let $f_\text{sym}$ denote the symmetrization of $f$, i.e. 
\[f_\text{sym}(x_1,\dots,x_h) = \bigwedge_{\sigma \in S_h} f(x_{\sigma(1)},\dots,x_{\sigma(h)}).\]

A constraint function $f$ is \emph{$0$-valid/invalid} if it is satisfied/violated by the all $0$-assignment.
Moreover, we say a constraint family is $0$-valid if this holds for all its constraints.
We define constants $\umin(f)$ and $\smin(f)$, denoting the weight of the minimum \emph{violating} and \emph{satisfying} assignments respectively. 
\begin{definition}
    For a Boolean constraint function $f$ we let $\umin(f):= \min\left\{\|a\|_1 \mid f(a)=0\right\}$ and $\smin(f):= \min\left\{\|a\|_1 \mid f(a)=1\right\}$. We call $f$ $\umin(f)$-violating and $\smin(f)$-satisfying respectively.
    For a Boolean constraint family $\Fam$ we define 
    \begin{itemize}
        \item $\umin(\Fam) = \min_{f \in \Fam} \umin(f)$ and 
        \item $\smin(\Fam) = \min_{f \in \Fam} \smin(f)$ respectively.
    \end{itemize}
\end{definition}

We generally assume our constraint families to be finite and we also implicitly assume that there are no trivial constraint functions contained, as they can be trivially resolved. 
In the following, we will refer to Boolean constraint functions as known from propositional logic, i.e. $F(x) = \neg x$, $\NAND(x,y) = \neg(x \wedge y)$, $\EQ(x,y) = x\leftrightarrow y$ and so on.

With this out of the way, we use the remainder of this section to first consider the setting of sublinear density and binary constraint families agnostic to sparsity,
then give a linear time algorithm for $\CSP_k^\gamma(\NAND)$ for density $\gamma < 2$ and
subsequently, provide algorithms for binary families containing $\NAND$ together with matching lower bounds.

Intuitively, while $\CSP_k^\gamma(\NAND)$ is either in time $\O(n+m)$ for $\gamma < 2$ or in time $n^{\frac{\omega k}{3}\pm o(1)}$, adding \emph{any} other binary constraint into the mix, makes the time exponent scale with the density exponent $\gamma$, resulting in a running time of 
$\O\left(n^{\gamma{(k - \smin(\Fam))\omega}/{6} + 1}\right)$, while depending on properties of the constraint family, there is a slight offset in $k$.
% \todo{This theorem does not hold! Consider \{\IMPL\}, \{\NAND\} or \{\NAND, \OR\}}
\subsection{General Families with Sublinear Density}
In this section, we briefly show that $\CSP$s with a sublinear number of constraints can be solved quickly, as asymptotically, they contain $\Theta(n)$ isolated vertices that facilitate finding a solution.
First we recall the following result by Marx \cite{Marx05}, that efficiently deals with $0$-invalid constraint functions by reducing them to the $0$-valid case.
\begin{lemma}[Branch and Bound~\cite{Marx05}]
    \label{lem:branch_and_bound}
    Let $\Fam$ be a constraint family that is not $0$-valid. 
    Then we can one-to-many reduce an instance $\Phi$ of $\CSP_k^\gamma(\Fam)$ 
    to at most $f(k)$ many instances $\Phi$ of $\CSP^{\gamma_i}_{k_i}(\Fam^*_i)$ where $k_i < k$ and $\Fam^*_i$ is a $0$-valid constraint family in time $\O(n+m)$.
\end{lemma}
\begin{proof}
    Observe that for a $0$-invalid constraint to be satisfied, at least one more variable has to be set to $1$.
    We can use this observation for a branch and bound approach as given in \cite{Marx05}, 
    where for every violated $0$-invalid constraint $f(x_1,\dots,x_r)$ we branch by setting $x_i$ for $i \in [r]$ to $1$, and replacing $f$ with this constraint on variables $\{x_1,\dots,x_r\} \setminus \{x_i\}$ and adding it to $\Fam^*$. If this constraint is infeasible, we terminate this branch. Note that if there is no feasible branch remaining, then the instance is infeasible.
    We terminate after at most $k$ steps or when arriving at an instance over a constraint family $\Fam^*$ with only $0$-valid constraints remaining.
    The search tree has at most $f(k)$ many leaves, thus we obtain for every leaf $i$ an instance $\Phi_i$ of $\CSP_{k_i}^{\gamma_i}(\Fam_i^*)$ 
    where each $k_i$ is bounded by $k_i \leq k - \smin(\Fam)$.
    The search tree is bounded in depth by $k$ with branching factor $r$ and we can check for constraints violated by the $0$-assignment in time $\O(n+m)$,
    obtaining the desired runtime.
\end{proof}

Using the previous lemma, we now consider all Boolean constraint families for density $\gamma < 1$.
The general intuition is that asymptotically, there exists a large enough set of unused variables, providing us access to a weight $k$ solution.
\begin{theorem}
\label{thm:sub_linear_density}
    Let $\Fam$ be a Boolean constraint family.
    Then any instance of $\CSP_k^\gamma(\Fam)$ for $\gamma < 1$ can be solved in time $\O(n)$.
\end{theorem}
\begin{proof}
    Clearly as there are only $\O(n^{1-\delta})$ many constraints for some $\delta > 0$,
    asymptotically for $k \ll n$ there are at least $k$ variables not appearing in any constraint.
    These present a weight $k$ solution for any $0$-valid $\Fam$, by setting all other variables to $0$.
    For $0$-invalid constraint families, we need to resolve the $0$-invalid constraints first:
    We apply \autoref{lem:branch_and_bound} to obtain $f(k)$ many $0$-valid instances in time $\O(n)$. 
    As branch and bound only considers variables that are contained in a constraint, there are still at least $k$ variables not appearing in any constraint -- using the previous argument, we obtain a solution.
\end{proof}

\subsection{Constraint Families without \texorpdfstring{$\NAND$}{NAND}}
In this section, we will briefly visit binary constraint families, that do not contain $\NAND$, proving part 1 and 2 of \autoref{thm:binary-csp-classification}.
For these families, sparsity does not play a substantial role: Either they are efficiently solved in the dense case already,
or the complexity does not change with respect to the instance density.
With \autoref{thm:sub_linear_density} in mind, we will only consider instances of $\gamma \geq 1$ in the following.

\begin{observation}
\label{obs:easy_constraints}
    We can preprocess an instance of $\CSP_k^\gamma(\Fam \cup \{\NOR,F\})$ to remove all occurrences of $\{\NOR,F\}$ in time $\O(n + m)$ 
    to obtain an instance of $\CSP_k^{\gamma'}(\Fam)$.
\end{observation}
It suffices to remove any variable appearing in either $F$ or $\NOR$, as setting those to $1$ would directly violate some constraint. For completeness sake, notice that this argument also extends to unary constraints on two variables of the form $f(x,y) = F(x)$, as we can simply replace all occurrences by $F(x)$.
While this preprocessing step might yield an instance of higher density $\gamma'$,
it is ensured that $m' \leq m, n' \leq n$ for the resulting instance.

The \emph{easiest} class of problems are those with constraint families that are weakly separable, a condition established in \cite{Marx05}. The authors show that such families are solvable in FPT-time even for the dense case, and in our binary setting this condition matches exactly those families that contain neither $\NAND$ nor $\IMPL$. In fact, such families can be solved in linear time, for any density.

\begin{theorem}
    Let $\Fam$ be a binary constraint family such that $\NAND \not \in \Fam$ and $\IMPL \not \in \Fam$.
    Then for any $1\le \gamma$ we can solve $\CSP^\gamma_k(\Fam)$ in time $\O(n + m)$.
\end{theorem}
\begin{proof}
    Using \autoref{lem:branch_and_bound} we can resolve all $0$-invalid constraints in $\Fam$ 
    and obtain $f(k)$ instances of $\CSP^{\gamma_i}_{k_i}(\Fam^*_i)$ in time $\O(n+m)$.
    An important observation is, that since $\Fam$ is binary, 
    the branch and bound procedure only introduces unary constraints to $\Fam^*_i$ (or resolves them directly).
    As neither $\IMPL$ nor $\NAND$ are contained in $\Fam$, and any other $0$-valid constraint of arity at most $2$ is either $\NOR$, $F$ or trivial,
    applying \autoref{obs:easy_constraints} yields an instance of $\CSP^{\gamma'_i}_{k_i}(\EQ)$ that we need to solve.
    To satisfy a set of $\EQ$ constraints, it is required to always set all variables of a connected component to either true or false.
    More precisely, by determining all connected components $C$ in the primal graph of an instance in $\O(m)$, 
    we can reduce an instance of $\CSP_{k_i}^{\gamma_i}(\EQ)$ to a bounded subset-sum 
    instance. Note that there are at most $n$ components, and every connected component $C_i$ corresponds to a positive weight $w_i = |C_i|$ for $i \in [n]$.
    As our target sum is $k$, we can bound every weight $0 < w_i \leq k$ and 
    further bound the number of weights $w_i$ of the same weight $w \leq k$ by $\lceil k / w\rceil$.
    Finally, this obtains an instance with at most $k^2$ distinct weights and target value $k$.
    Using a fast time algorithm for bounded subset sum (e.g. \cite{PISINGER19991}) each instance can be solved in time $\O(k^3)$.
\end{proof}
Hence the interesting cases are constraint families that \emph{represent} $\IMPL$ and $\NAND$, consider the regimes as classified in \cite{KunnemannM20}.
In the binary setting, the only constraints that represent $\IMPL$ and $\NAND$, are exactly those functions themselves.

Due to the nature of the precedence of $\IMPL$ constraints, there is a simple argument that can be made to reduce any dense instance on $\IMPL$-constraints to a sparse one. 
\begin{lemma}
    An instance $\Phi$ of $\CSP_k^\gamma(\IMPL)$ can be reduced to an instance $\Phi'$ of $\CSP_k^{1}(\IMPL)$ in time $\O(n^\gamma)$.
\end{lemma}
\begin{proof}
    \label{lem:sparse_impl}
    Consider the primal graph $G$ of $\Phi$.
    Observe that setting a variable $x_i$ to $1$ entails setting every descendant of $x_i$ to $1$ as well.
    Thus, only variables with less than $k$ descendants can be part of a weight $k$ solution, allowing us to prune 
    all vertices with outgoing degree of $k$ or larger to obtain an equivalent instance
    $\Phi'$ that has most $\O(n)$ many constraints and thus a density of $\gamma = 1$. 
\end{proof}

Finally, for all binary constraint families that do not contain $\NAND$, but $\IMPL$ we essential obtain the same regime as in the dense case.
\begin{theorem}[Thm. 1.3. \cite{KunnemannM20}]
    Let $\Fam$ be a binary constraint family with $\NAND \not \in \Fam$. 
    There exists a constant $c_\Fam$ such that for any $1 \le \gamma\le 2$,
    an instance $\Phi$ of $\CSP_k^\gamma(\Fam\ \dot\cup\ \IMPL)$ can be solved in time $\O(n^{4\sqrt{k} + c_\Fam})$. 
    Further, there is no $\varepsilon > 0$ such that we can solve $\CSP_k^\gamma(\Fam\ \dot\cup\ \IMPL)$ in time $\O(n^{\sqrt[3]{k}\omega/6 + c_\Fam -\varepsilon})$ assuming the Clique Hypothesis fails.
\end{theorem}
As the hardness proof given by Künnemann and Marx constructs a sparse instance for $\CSP_{k}^\gamma(\IMPL)$ already, it seems unlikely so far that we can radically improve the existing gap between upper and lower bound due to sparsity considerations.

It now remains to investigate families that contain $\NAND$ constraints.
\subsection{Binary Constraint Families containing \texorpdfstring{$\NAND$}{NAND2}}
What differentiates $\NAND$ from other binary constraints regarding density? 
Turán's Theorem \cite{turan1941extremal} answers this question,
asymptotically guaranteeing the existence of a $k$-independent set in a sparse graph for $m \in O(n^\gamma)$ for $\gamma < 2$.

In this section we will give upper and lower bounds for binary constraint families that contain $\NAND$. 
Specifically we will prove parts 3 and 4 of \autoref{thm:binary-csp-classification}.
We begin by giving a linear time algorithm for solving $\CSP^\gamma_k(\NAND)$ for $\gamma <2$.
For \emph{Clique Regime}, we first show how to solve any $0$-valid binary constraint family in combination with $\NAND$, 
then reduce $0$-invalid families to this case.

As a preliminary result, we consider the task of detecting a $k$-independent set in sparse graphs, allowing us to solve $\CSP^\gamma_k(\NAND)$ for $\gamma <2$ in time $\O(n+m)$.
The following corollary shows that, speaking asymptotically, every sparse graph contains a $k$-independent set.
\begin{corollary}[Sparse $k$-Independent Set]
\label{thm:turan-kis}
Let $G$ be a graph on $n$ vertices with $m \leq \frac{n^2}{2k}$ many edges.
Then $G$ contains a $k$-independent set for large enough $n$.
\end{corollary}
\begin{proof}
    Assume there was a $k$-independent set-free graph $G$ on $n$ vertices with $m$ edges. 
    Then the complement graph $\overline{G}$ is $k$-clique-free and the number of edges is bounded from below by 
    \[
    \left|E(\overline{G})\right| \geq \frac{n(n-1)}{2} - \frac{n^2}{2k}\geq \left(1-\frac{1}{k} - o(1)\right) \frac{n^2}{2} > \left(1-\frac{1}{k-1} + o(1)\right)\frac{n^2}{2}
    \]
    By Turán, the largest $k$-clique-free graph on $n$ vertices has at most $\left(1 - \frac{1}{k-1} + o(1)\right)\frac{n^2}{2}$ edges. 
    Thus, for large $n$, we arrive at a contradiction.
\end{proof}

This result can be made constructive, by iteratively picking low degree vertices, removing them and their neighborhood from the graph until we find a $k$-independent set. This procedure neatly ensures that during every step $i$, it is guaranteed that the graph still contains a $(k-i)$-independent set.
\begin{lemma}
\label{lem:constructive_turan}
    If there exists a $k$-independent set by \autoref{thm:turan-kis}, then we can find it in time $\O(n+m)$.
\end{lemma}
\begin{proof}
    Consider some graph $G$ as required by \autoref{thm:turan-kis} on $n$ vertices with $m \leq \frac{n^2}{2k^2} \leq \frac{n^2}{2k}$ for large enough $n$. 
    %\marvin{Cor. 24 already needs to use $n^2/2k$.} 
    For our further considerations this suffices, allowing us to avoid tedious details, but we note that one can make this constructive for the tight bound as well.
    Starting with $G_{0}\coloneqq G$ (define $n_i,m_i$ analogously) we proceed iteratively for $k$ rounds to build an independent set $I$: 
    Pick a vertex $v_i$ of degree $d_i \coloneqq \deg(v_i) \leq \frac{2m_i}{n_i}$, 
    add it to $I$ and continue with $G_{i+1} \coloneqq (G_i \setminus \{v_i\} \cup N(v_i))$ while $i<k$.  
    As we remove the neighborhood of every $v_i$, it is easy to see that the resulting set $I$ is indeed a $k$-independent set.

    By the handshake lemma and an averaging argument, it is guaranteed that $G_i$ contains a vertex $v_i$ of degree $d_i$.
    For correctness, it is left to prove that every iteration (up to the $k$-th one) maintains the graph's density, 
    such that $G_i$ indeed always contains a $(k-i)$-independent set.
    More precisely, we prove that $m_i \leq \frac{n_i^2}{2(k-i)^2}$ for all $i < k$.
    For $m_0$ this is clear. 
    Let $c := k-i$: going from step $i$ to $i+1$, we remove $v_i$ and its neighborhood of size $d_i \leq \frac{n_i}{c^2}$. Thus we can bound the size of the graph $G_{i+1}$ by
    \begin{align*}
        n_{i+1} & \geq n_i - \frac{n_i}{c^2} - 1 \geq \left(\frac{c-1}{c}\right) n_i \text{ for large enough $n_i$ and}\\
        m_{i+1} &\leq m_i
    \end{align*} 
    Assume the invariant holds for all $i < i+1$, then
    \[
        m_{i+1} \leq m_i < \frac{1}{2} \left(\frac{n_i}{c}\right)^2 \leq \frac{1}{2} \left(\frac{n_{i+1}}{c-1}\right)^2.
    \]
    The last step follows as $n_{i+1} \geq \left(\frac{{c-1}}{c}\right) n_i$.
    Thus during step $i$ of the algorithm, \autoref{thm:turan-kis} ensures the existence of a $(k-i)$-independent set in the remaining graph $G_i$.
  
    Finally, the running time is bounded by 
    \[
        \O\left(\sum_{i=0}^k n_i + m_i\right) = \O\left(k \cdot (n_0 + m_0)\right) = \O\left(n + m\right). 
    \]
\end{proof}
\begin{theorem}
    We can solve $\CSP_k^\gamma(\NAND)$ in time $\O(n + m)$ for $\gamma < 2$.
\end{theorem}
%\thinking{Timo:Maybe we should add that such instances are YES-instances and one could return YES immediately and further emphasize that this runtime is essentially optimal }
%\thinking{yeah kinda what i was going for in the paragraph, could be made explicit in the theorem statement i think}
\begin{proof}
    Consider the primal graph $G$ with $m = \Theta(n^\gamma)$: Then a solution is a $k$-independent set in $G$.
    As there exists some $n_0$ such that $m \leq \frac{n^2}{2k}$ for all $n \geq n_0$, a simple application of \autoref{lem:constructive_turan} constructs a solution for $\CSP_k^\gamma(\NAND)$ in time $\O(n + m)$.
\end{proof}

The natural thing to consider is the introduction of further constraint functions. 
Maybe surprisingly, this makes the problem harder to solve, depending on the density.
\subsubsection{Solving \texorpdfstring{$0$}{0}-valid Constraint Families}
\label{sec:NAND_IMPL}
For this section, we give algorithms for $0$-valid constraint families containing $\NAND$.
In particular we will show the following:
\begin{theorem}
\label{thm:0_valid_binary_algo}
    Let $\Fam$ be a $0$-valid binary constraint family.
    Then we can solve $\CSP_k^\gamma(\Fam\ \dot\cup \ \{\NAND\})$ in time $\O(m^{\frac{\omega k}{6} + 1})$.
\end{theorem}
Algorithmically, this boils down to preprocessing all easy constraints first and remain with an instance over constraints $\{\NAND,\IMPL\}$ for which give an algorithm.

For the first part, recall that we can resolve $\NOR,F$ and any ``unary'' constraint on two variables in linear time.
\begin{observation}
\label{obs:0_valid_no_impl}
Let $\Fam$ be binary constraint family not containing $\EQ$ or $\IMPL$.
 $\CSP_k^\gamma(\Fam \ \dot\cup \ \{\NAND\})$ can be solved in time $\O(m^{\frac{\omega k}{6}})$.
\end{observation}
We preprocess the instance by \autoref{obs:easy_constraints} in time $\O(n + m)$.
Depending on the density of the residual instance over constraints from $\Fam$ with $n'\leq n$ vertices and $m' \leq m$, we can then find a solution in linear time if $m' = n^{2 - O(1)}$, or we know that the resulting instance is indeed dense, thus $n' \in \O(\sqrt{m'})$ and we obtain a running time of $\O(m^{\frac{\omega k}{6}})$.

It thus suffices to consider the remaining $0$-valid binary, constraint functions (modulo symmetries) $\IMPL$ and $\EQ$.
As $\EQ$ reduces to $\IMPL$ by replacing each $\EQ$ constraint with two $\IMPL$-constraints,
we only give an algorithm for $\{\NAND,\IMPL\}$:
\begin{theorem}
\label{thm:nand_impl_alg}
    There is an algorithm solving $\CSP_k^\gamma(\{\NAND,\IMPL\})$ in time $\O\left(m^{k\omega/6 + 1}\right)$.%$m^{\omega k/6 + 2 + o(1)}$.
\end{theorem}
    In the following, we will consider the primal graph $G$ of the given instance, distinguishing between $\NAND$ and $\IMPL$ edges in $G$.
    % We begin by preprocessing the instance: First, observe that we can remove all pair of vertices contained in $\NAND$ and $\IMPL$ together.
    With regard to the $\IMPL$-edges, we define for a vertex $v$ its descendants $D(v) = \{ u \mid u \in V, \text{s.t. }u\text{ is reachable from }v \}$ and ancestors $A(v) = \{ u \mid u\in V, \text{s.t. }v\text{ is reachable from }u\}$. In particular we have that $v\in D(v)$ and $v\in A(v)$. We also extend this notation to sets of vertices s.t. $D(V) = \bigcup_{v \in V} D(v)$.
    We compute for every vertex $D(v)$ and $A(v)$. During the computation, if we find two vertices $u,w\in D(v)$ such that there exists a $\NAND$-edge on $u,w$, we can safely remove $v$ and all its ancestors from the graph. Further, we can remove any vertex $v$ such that $D(v)\ge k$.

    After this preprocessing step, we observe the following:
    due to the nature of $\IMPL$ constraints, vertices that have many descendants are \emph{cheap} to guess:
    Consider the vertices $V_{\geq 3} = \{ |D(v)| \geq 3 \mid v \in V\}$. 
    Setting some $v \in V_{\leq 3}$ to $1$ yields at least $2$ more variables $u \in D(v)$, which we need to set to $1$.

    Hence, let us first consider the hard case where for all $v \in V$ it holds that $|D(v)| \leq 2$. 
    We call instances that satisfy this condition \emph{restricted}. 
    After giving an algorithm for such restricted instances, 
    we show that this yields an equally fast algorithm for the general case.
\begin{theorem}[Restricted $\CSP_k^\gamma(\{\NAND,\IMPL\}$]
\label{thm:restricted_nand_impl}
    Let $\Phi$ be an instance of $\CSP_k^\gamma(\{\NAND,\IMPL\})$ such that no two variables share both an \IMPL- and \NAND$_2$-constraint and $|D(v)| \leq 2$ for all $v \in \text{vars}(\Phi)$.
    Then there is an algorithm solving $\Phi$ in time $\O\left(m^{k\omega/6 + 1}\right)$.%$m^{\omega k/6 + 2 + o(1)}$.
\end{theorem}
\begin{proof}
We devise this algorithm in 3 steps. 
First, we exploit the restriction on $\IMPL$-constraints to partition the variables into groups. 
Each group induces a solution within itself, if it is sufficiently large and sparse w.r.t. $\NAND$-edges. 
Otherwise, we can assume that such a group must be dense which allows us to bound its size. 
Then, we will deal with $\IMPL$-edges that form $2$-cycles and describe the main part of the algorithm, 
consisting of guessing and grouping vertices cleverly (depending on $\IMPL$-edges), finally reducing to triangle detection.

    The following observation regarding $\NAND$ sparsity of large vertex sets is crucial for our algorithm:
    \begin{claim}
    \label{lem:small_groups}
            Let $V' \subseteq V$ with $|V'| \ge \sqrt{m}^{1 + \epsilon}$ for all $\epsilon > 0$,
            then there exists an independent set $\{v_1,\dots, v_{k}\}$ w.r.t. $\NAND$-edges. 
    \end{claim}
    \begin{proof}
    Assume $V'$ were dense, i.e. the count of edges is
    \[\O\left(|V'|^2\right) = \O\left(\left(\sqrt{m}^{1+\epsilon}\right)^2\right) = \O\left(m^{1+\epsilon}\right) = \O\left(n^{(\gamma + \gamma\epsilon)}\right)\]
    where $\gamma\epsilon >0$, which is a contradiction as the graph has only $\O\left(n^\gamma\right)$ edges in total.
    Thus we can find a $k$-independent set in $V'$ in linear time by \autoref{lem:constructive_turan}.
    \end{proof}
    In particular, this bounds the size of any $V' \subseteq V$ where $\IMPL$ edges allow us to (almost) freely choose vertices from $V'$ 
    by $|V'| \geq \sqrt{m}^{1 + \epsilon}$. Otherwise we can find a $k$-independent set directly in time $\O(n+m)$.

    \subparagraph*{Removing \texorpdfstring{$2$}{2}-cycles}
    Recall that we are promised a \emph{restricted} instance with $|D(v)| \leq 2$ for all $v \in V$.
    More precisely, every such vertex $v$ has at most a single outgoing $\IMPL$-edge.
    Observe that this entails that $2$-cycles are isolated w.r.t. $\IMPL$-edges.
    Using \autoref{lem:small_groups}, we can bound the size of this set by $\O(\sqrt{m})$. 
    %\timo{we require $k$ to be even for this to true, right?}
    It now suffices to guess a set of at most $\lfloor k/2 \rfloor$ many cycles and solve the remaining instance without cycles.
    Assuming we can solve the problem without $2$-cycles in time $\O\left(m^{k\omega/6 + 1 }\right)$, we obtain the running time 
    \[
    \sum_{j = 0}^{\lfloor k/2 \rfloor} \O({\sqrt{m}}^{j}) \O\left(m + m^{(k - 2j)\omega/6 + 1}\right) = \sum_{j = 0}^{\lfloor k/2 \rfloor} \O\left(m^{ k\omega/6 + 1 - (2j\omega )/6 + j/2}\right) = \O\left(m^{k\omega /6 + 1}\right)
    \] 
    for $\omega \geq 2$.
    
    \subparagraph*{Reducing to Triangle Detection}
    After dealing with $2$-cycles with $\IMPL$-edges, 
    we partition $V$ into sets 
    \begin{itemize}
        \item $V_L = \{ v \mid |A(v)| = 1 \text{ and } |D(v)| = 2\}$,
        \item $V_R = \{v \mid |A(v)| \geq 2\}$ and
        \item $V_0\, = \{ v \mid |D(v)| = |A(v)| = 1\}$.
    \end{itemize} 
    For any vertex  $u \in V_L$ to be part of a solution, there exists some $v \in V_R$ that is necessarily contained as well. 
    % \timo{I dont think this is true since there might be variables with no incident \IMPL-constraint, neither outgoing nor incoming. There can't be too many such variables, but some might exist.}
    % \julian{indeed, those need to be dealth with. I guess we can just use them as another group}
    In this regard, we call the set of ancestors $A(v)$ of $v\in V_R$ and $v$ itself its \emph{group}. 
    Conversely, if we know that some $v \in V_R$ is part of a solution, then including any subset of its group is consistent with the $\IMPL$ constraints. Vertices in $V_0$ are isolated w.r.t. $\IMPL$-edges.
     
    Using \autoref{lem:small_groups}, we can give bounds on the size of groups. 
    If a group is too large, we can easily find a solution fully contained therein.
    Thus, we have the following properties for groups:
    \begin{enumerate}[P1]
        \item $|V_R|\leq f(k)\sqrt{m}$, thus there are at most $\O(\sqrt{m})$ many groups.
        \item Each group has size at most $|A(v)| \leq f(k)\sqrt{m}$
    \end{enumerate}
    To see this, observe that there are no $\IMPL$-edges between vertices in $V_L$, and no $\IMPL$-edges contained in $V_R$, as the restricted instance promises that $|D(V)| \leq 2$ for all $v \in V$.
    Lastly, we also consider $V_0$ as a group as well, where we can bound its size $|V_0| \leq f(k) \sqrt{m}$ by \autoref{lem:small_groups}.

    We exploit these properties by guessing \emph{how many groups} (of which there are at most $\O(\sqrt{m})$) a solution is composed of.
    Further, we guess how many vertices $k_i$ of each group (containing at most $\O(\sqrt{m})$ vertices) are part of the solution, thus it suffices to consider at most $\O(\sqrt{m}^{k_i})$ subsets per group.
    To this end, the algorithm guesses a partition of $k$ into $\ell$ groups of sizes $k = k_1+\dots+ k_\ell$ with $k_1\leq \dots \leq k_\ell$,
    to eventually reduce to an instance of triangle detection.
    First, we work through the case of at most two groups. 
    \subparagraph*{$\mathbf{\ell \leq 2}$}  
    We guess at most two groups in time $\O({\sqrt{m}}^2) = \O(m)$. 
    It remains to find a $(k-2)$-independent set among the $f(k)\sqrt{m}$ many vertices in the groups $A(v_1) \cup A(v_2)$.
    We distribute these $\O\left(\sqrt{m}\right)$ vertices into three parts, where each part corresponds to a choice of $\left\lceil\frac{k-2}{3}\right\rceil\leq \frac{k}{3}$ many vertices and reduce to a $3$-partite triangle instance of size $\O\left(\sqrt{m}^{\frac{k}{3}}\right)$.
    Two vertices in the triangle instance are adjacent, if the corresponding sets of vertices in the original instance form an independent set.
    The total time of guessing the groups, constructing and solving the triangle instance is bounded by $\O\left(m^{k\omega/6 + 1}\right)$.
    % \timo{Im not sure why we need to reduce to triangle and make these bucket arguments. Can't we just solve (k-2)-IndSet on $A(v_1) \cup A_(v_2)$?}
    % \julian{Fair enough, although this frames the algorithm as a reduction to triangle in any case which feels more neat}
    
    \subparagraph*{$\mathbf{\ell \ge 3}$}  We first prove the following claim, stating that we can distribute the groups corresponding to $k_1,\dots, k_{\ell-2}$ in three bins in such a way that the \emph{imbalance} of any pair of bins is at most $k_{\ell-1}$; then we can use the remaining two groups to fix the imbalance:
    \begin{claim}
        Let $k_1\leq \dots \leq k_\ell \in \mathbb{N}$.
        Then there exists a distribution of $k_1,\dots,k_{\ell-2}$ into three bins $S_1, S_2, S_3$ such that the following condition holds 
        \[
        \left|\left|\left(\sum_{k_i\in S_p} k_i\right)\right| - \left|\left(\sum_{k_j\in S_q} k_j\right)\right| \right| \leq k_{\ell-1}
        \]
        for each pair $p,q\in [3]$.
    \end{claim}
    \begin{proof}[Proof (sketch)]
        A greedy approach suffices: At step $i$, put value $k_i$ into the bin whose sum is the smallest after iteration $i-1$. 
        Invariant: after step $i$ the difference between the sum of values of any pair of bins will be at most $k_i$. 
        Hence, after iteration $\ell-2$, the difference between any pair of buckets will be at most $k_{\ell-2} \leq k_{\ell-1}$.
    \end{proof}
    Let $s_1, s_2,s_3$ be the sums of values from the bins $S_1, S_2, S_3$ respectively. 
    We construct the triangle instance by following the construction from the proof of the previous claim. 
    This yields an instance of triangle detection $X_1,X_2,X_3$ where each set $X_i$ corresponds to all valid choices of $s_i$ vertices that adhere to our group assignment choice. Again, two vertices are adjacent if their vertex choices form a valid independent set.
    
    However, this is not yet sufficient, as the imbalance is potentially still large.
    To this end, we use the remaining $k_{\ell-1}+k_\ell$ vertices to balance these partitions.
    We guess the remaining two groups in $\O(m)$ (i.e. assume that the corresponding vertices in $V_R$ are part of the solution) time and can now distribute the vertices of their group,
    where exactly $k_{\ell-1} + k_{\ell} - 2$ many of those vertices have to be part of a solution.
    The new triangle instance will have each node correspond to a subset of $\binom{V}{s_i+r_i}$ vertices of the original instance, 
    where the $s_i$ vertices are obtained by taking the union of groups $1,\dots,\ell-2$ as in the claim, 
    and the $r_i$ vertices come from the remaining two groups and are distributed in such a way to minimize the imbalance.
    
    Indeed, we can reduce the imbalance between buckets to $\max_{i,j\in[3]} |s_i+r_i-s_j-r_j|\leq 1$.
    See that the imbalance after distributing $k_1,\dots,k_{\ell - 2}$ is at most $k_{\ell - 2}$.
    Thus, we distribute the contribution $k_{\ell-1} - 1$ and $k_\ell - 1$ of our guessed groups, 
    i.e. how many vertices of said groups are set to one in a solution, to the imbalanced buckets.
    As we have at least $2(k_{\ell-2})$ many vertices available and the imbalance is at most $k_{\ell - 2}$, we can distribute the contributions up to an imbalance of $1$.
    Thus we obtain buckets of size at most $s_k \leq \lceil\frac{k-2}{3}\rceil \leq \frac{k}{3}$ and the triangle instance is of size $\O(\sqrt{m}^{\frac{k}{3}})$.
    Thus the time required for guessing groups, constructing and solving the triangle instance can be bounded by $\O\left(m^{ k \omega / 6 + 1}\right)$.

    We perform this reduction for every possible partition of $k$, we arrive at a total runtime of $\O\left(m^{ k \omega / 6 + 1}\right)$
    \end{proof}
    It is now left to show, that we can efficiently reduce any given ${\CSP_k^\gamma(\{\NAND,\IMPL\})}$ instance to the restricted case:
\begin{lemma}
\label{lem:restricting_nand_impl}
    Let $\Phi$ be an instance of $\CSP_k^\gamma(\{\NAND,\IMPL\})$. 
    Assume we can solve the restricted case in time $\O\left(m^{c \cdot k + d}\right)$ for constants $c \geq \frac{1}{3}$ and $d$,
    then we can solve $\Phi$ in time $\O\left( m^{c \cdot k + d}\right)$.
\end{lemma}
\begin{proof}
    Let $V_{\geq 3} = \{ v \in V \mid |D(v)| \ge 3\}$ for the primal graph $G$. We refer to those vertices as \emph{heavy}: Recall that heavy vertices are cheap to guess, as they imply at least $2$ more vertices that are contained in a solution.
    As any solution of weight $k$ can contain at most $k/3$ many heavy vertices 
    it suffices to iterate over subsets $S \subseteq V_{\ge 3}$ with $|S| \leq k/3$. 
    Every such set $S$ we want to consider needs to fulfill the following properties:
    \begin{itemize}
        \item There exists no $\NAND(u,v)$-edge with $u,v \in D(S)$,
        \item $|S \cup D(S)| \leq k$, and
        \item $|S \cup D(S)| \geq 3|S|$
    \end{itemize}
    Verifying the validity of a set $S$ can be done in time $\O(|S|)$.

    Given a valid set $S$, we obtain a restricted instance by propagating that $S':=D(S)$ is part of the solution and removing any remaining heavy vertices. 
    More precisely, we delete $S'$ and its neighborhood $N(S')$ w.r.t. $\NAND$-edges, as these vertices can not be part of a solution anymore. 
    Then, we remove all vertices that are still heavy after this step.
    Finally, for every deleted vertex, we need to remove all its ancestors, as they can no longer be part of a solution as well by the $\IMPL$-edges. 
    Note, that as we can compute $D(v)$ for all $v \in V$ in time $\O(m)$,
    this propagation can be done in time $\O(m)$ and obtains a restricted instance. 
    
    Iterating through all subsets of size $s \leq k/3$, checking their validity, propagating and solving the restricted instance
    thus has a total running time of
    \[
    \sum_{s=0}^{k/3} \O\left(s\cdot n^{s}\right) \O(m + m^{c(k-3s) + d}) \leq \sum_{s=0}^{k/3} \O(m^{ck-3cs + s + d}) =\O(m^{ck + d}).
    \]
    as $n \leq m$, $c \geq \frac{1}{3}$ and $|S'| \geq 3|S|$.
\end{proof}
Using \autoref{lem:restricting_nand_impl} and \autoref{thm:restricted_nand_impl} and the fact that $\frac{\omega}{6} \geq \frac{1}{3}$,
this proves \autoref{thm:nand_impl_alg}.
\begin{proof}[Proof of \autoref{thm:0_valid_binary_algo}]
    Let $\Fam$ be a $0$-valid binary constraint family with $\NAND \in \Fam$.
    Given an instance $\Phi$ of $\CSP^\gamma_k(\Fam)$, we can preprocess $\Phi$ in time $\O(n + m)$.
    If $\Phi$ contains $\EQ$ or $\IMPL$ constraints,
    replace any potentially occurring $\EQ$-constraints by two $\IMPL$-constraints and then solve the resulting instance with $\autoref{thm:nand_impl_alg}$ in time $\O(m^{k\omega/6 + 1})$.
    If there are no such constraints, then we can solve the instance in time $\O(m^{k\omega/6})$ by \autoref{obs:0_valid_no_impl}.
\end{proof}

\subsubsection{Solving \texorpdfstring{$0$}{0}-invalid Constraint Families}
Solving $0$-invalid functions boils down to an application of the branch and bound approach of \autoref{lem:branch_and_bound} and solving the resulting $0$-valid instance.
\begin{theorem}
    \label{thm:alg1invalid}
    Let $\Fam$ be a $0$-invalid, binary Boolean constraint family. 
    Then we can solve $\CSP_k^\gamma(\Fam \cup \NAND)$ in time $\O(m^{{\omega(k-\smin(\Fam))}/{6} + 1})$.
    %where $\smin(\Fam) \in \{1,2\}$ depends only on $\Fam$.
\end{theorem}
\begin{proof}
    We apply \autoref{lem:branch_and_bound}. 
    First we argue that to replace a $0$-invalid constraint $f$ with a $0$-valid constraint, 
    at least $\smin(f)$ many $1$s are required to be plugged in. As we prune infeasible branches, we can thus bound each $k_i \leq k - \smin(\Fam)$.
    In the binary case it is easy to see that,
    as $\Fam^*_i$ is $0$-valid but $\Fam$ was $0$-invalid, the only additionally introduced constraint functions are of arity at most $1$ and $0$-valid, these constraints can be directly resolved in linear time.
    Now, it is only left to solve $f(k)$ many instances of $\CSP_{k_i}^{\gamma_i}(\Fam_0)$ for some computable $f(k)$ and $\Fam_0 \subseteq \Fam$ containing only $0$-valid constraints. 
    Each can be solved in time $\O\left(m_i^{\omega(k - \smin(\Fam)) / 6+ 1} \right)$ 
    by \autoref{thm:0_valid_binary_algo}, 
    yielding a runtime of $\O\left(f(k) \cdot m_i^{\omega(k - \smin(\Fam))  / 6 + 1}\right) = \O\left(m^{\omega(k - \smin(\Fam))  / 6 + 1}\right)$.
\end{proof}

\subsubsection{Hardness}
We now go on to prove our main hardness result, showing that the previous algorithms are in fact optimal up to FPT-factors of the form $f(k)n^{\O(1)}$.
It becomes apparent that there is a \emph{smooth} hardness transition depending on the density, as soon as we consider any other binary constraint in combination with $\NAND$.
The reduction is based on the idea of \emph{hiding} a dense, hard core instance of $\NAND$.
By using only few constraints to enforce a large set of auxiliary variables to be $0$, we embed this core into a sparse instance.
We obtain the following lower bound for all binary constraint families containing $\NAND$:

\begin{theorem}
     Let $\Fam$ be a non-empty binary constraint family.
     There exists no $\varepsilon > 0$ such that we can solve $\CSP^\gamma_k(\Fam \ \dot\cup \ \{\NAND\})$ for $1\leq \gamma \leq 2$ in time
    \(
        \O(m^{{(k - \smin(\Fam))\omega}/{6} - \varepsilon})
    \)
\end{theorem}
\begin{proof}
    We reduce from (dense) $k$-independent set, i.e. an instance $\Phi$ of $\CSP_k^2(\NAND)$
    to an instance of $\CSP_k^{\gamma}(\Fam \cup \{\NAND\})$ for any $1 \leq \gamma \leq 2$.
    In the following, for some $f \in \Fam$ with $\smin(f) = \smin(\Fam)$,
    we will construct a gadget formula over $\ell:=n^{2/\gamma}$ many new variables $Y = \{y_1,\dots,y_\ell\}$ 
    plus at most $2$ further auxiliary variables $Z = \{z_1,z_2\}$.
    This gadget should be sufficiently sparse (i.e. have only linearly many constraints) and the only satisfying assignments of weight smaller than $k$
    are of weight exactly $\smin(\Fam)$. Using this gadget, we can then embed the given instance 
    $\Phi$ in a sparse instance to obtain any desired density $\gamma < 2$ while retaining an equivalence between size $k$ and size $k + \smin(\Fam)$ solutions.

    Based on $\smin(\Fam)$, we will construct a gadget constraint $g(x,y)$ that is only satisfied for weight $0$ assignments.
    
    \subparagraph*{$\mathbf{\smin(\Fam) = 0}$}
    Let $g(x,y) \equiv f(x,y) \wedge f(y,x) \wedge \NAND(x,y)$. As we can assume that $f$ is neither trivial nor $\NAND$, $g(x,y)$ is only satisfied by the $0$-assignment.
    
    \subparagraph*{$\mathbf{\smin(\Fam) = 1}$}
    W.l.o.g., let $f(1,0) = 1$.
    We construct 
    \(g'(x) \equiv f(z_1,z_2) \wedge \NAND(z_1,z_2) \wedge \NAND(z_1,x) \wedge \NAND(z_2,x),\)
    which sets $z_1$ or $z_2$ to $1$ but forces $x$ to take value $0$. Thus we get the desired $g(x,y) \equiv g'(x)\wedge g'(y)$.
    
    \subparagraph*{$\mathbf{\smin(\Fam) = 2}$}
    We construct $g(x,y) \equiv f(z_1,z_2) \wedge \NAND(z_1,x) \wedge \NAND(z_1,y)$, which sets $z_1$ and $z_2$ to $1$ but forces $x$ and $y$ to take value $0$.

    \subparagraph*{}
    In every case, it is easy to verify that $g(x,y)$ is only satisfied by setting $x$ and $y$ to $0$.
    We construct a \emph{cycle} using $g$
    \[
        \Phi' := \bigwedge_{i=1}^{\ell-1} g(y_i,y_{i+1}) \wedge g(y_{\ell},y_{1})
    \]
    using only $\O(n^{2/\gamma})$ many constraints.
    $\Phi'$ is only satisfied by the standard assignment, i.e. setting $Y$ to $0$ and $\smin(\Fam)$ variables in $Z$ to $1$.
    
    We set $k' = k + \smin(\Fam)$ and finally construct the instance $\Phi'' \equiv \Phi \wedge \Phi'$ depending on $\smin(\Fam)$.
    By our construction any weight $k'$ satisfying assignment $a$ for $\Phi''$ induces a weight $k$ assignment 
    that is satisfying for $\Phi$, as auxiliary variables need to be set to the standard assignment for $\Phi'$.
    As $\Phi''$ has $\O(n^{2/\gamma})$ many variables but $\O(n^2 + n^{2 / \gamma}) = \O(n^2)$ many constraints, the density is $m = \O(n^\gamma)$.

    Assume we can solve this instance in time $\O(m^{(k' - \smin(\Fam))\omega/
    {6} -\varepsilon})$ 
    for some $\varepsilon > 0$, 
    then we can solve $\CSP_k^2(\NAND)$ in time $\O(n^{{(2 \gamma k \omega)}/{(6\gamma)} -\varepsilon'}) = \O(n^{\frac{\omega k}{3} -\varepsilon'})$ for some $\varepsilon' > 0$
    refuting the $k$-Clique Hypothesis.
\end{proof}

This opens up a new regime when compared to the dense case, where the exponent is linearly dependent on the density of the instance:
While $\CSP_k^\gamma(\NAND)$ exhibits a \emph{sharp} {phase transition} in hardness at $\gamma = 2$, 
if there are additional constraint functions available 
this introduces a \emph{smooth} phase transition for $\gamma <2$.

\section{Hardness at the Triviality Cutoff: \texorpdfstring{$\ge 2$}{>= 2}-violating families}\label{sec:triviality-cutoff}
%was: Beyond binary - $\geq 2$-violating constraint Families
\label{sec:beyond_binary}
After fully classifying $\CSP_k^\gamma$ for binary constraint families, we now turn to higher-arity families.
In particular, we consider the \emph{$c$-violating} constraint families for $c \geq 2$, that is, the families satisfying $\umin(\Fam) \geq 2$. 

This class of constraint functions is structurally closely related to $\NAND_c$, that is clearly $c$-violating.
More importantly, we show that every $c$-violating family can implement the the arity-$h$ constraint $\textsc{LessThan}^h_c$, which is only satisfied by assignments of weight less than $c$.

Using this property, we construct two main ingredients used in our hardness results.
At high level, these gadgets allow us to \emph{embed} instances of $\CSP_k^\gamma$ into instances of different densities.

The main results of this section are twofold.
First, for every $c \geq 3$ we establish a sharp \emph{phase transition} in complexity of $c$-violating constraint families $\Fam$:
depending on the density $\gamma$, the problem $\CSP_k^\gamma(\Fam)$ is either trivial, or requires essentially brute-force time under the $d$-uniform $k$-hyperclique hypothesis.
We then turn to the $c=2$ case. Beyond the corresponding \emph{phase transition}, we revisit constraint families containing $\NAND$ of different arities, which is an equivalent formulation of the non-uniform $k$-independent set problem. 
More formally, we obtain the following main results in this section:
\begin{theorem}
\label{thm:3-vio-hardness}
    Let $\Fam$ be a constraint family with $\umin(\Fam) \ge 3$.
    Then:
    \begin{itemize}
        \item $\CSP_k^\gamma(\Fam)$ is trivial for $\gamma < \umin(\Fam)$.
        \item $\CSP_k^\gamma(\Fam)$ requires time $n^{k-o(1)}$ for any $\gamma \ge \umin(\Fam)$, unless the $\umin(\Fam)$-uniform Hyperclique Hypothesis fails.
    \end{itemize} 
\end{theorem}
Under the $k$-clique hypothesis, we obtain a similar result for families with $\umin(\Fam) = 2$.
\begin{theorem}
\label{thm:2-vio-hardness}
    Let $\Fam$ be a constraint family with $\umin(\Fam) = 2$.
    Then:
    \begin{itemize}
        \item $\CSP_k^\gamma(\Fam)$ is trivial for $\gamma < \umin(\Fam)$.
        \item $\CSP_k^\gamma(\Fam)$ it requires time $n^{\frac{\omega k}{3}-o(1)}$ for $\gamma \ge \umin(\Fam)$, assuming the $k$-Clique Hypothesis.
    \end{itemize} 
\end{theorem}
Constraint families that contain constraint functions of different $\umin$ parameters admit the sparse embedding approach already introduced in the binary case, where a dense hard core instance is embedded into a larger sparse instance.
Consequently, for some constraint families with $\umin(\Fam) = 2$ that additionally contain some $g$ with $\umin(g) > \umin(f)$, we can give a stronger lower bound for specific densities $\gamma$.
\begin{theorem}
\label{thm:sandwich-hardness}
    Let $\Fam$ be a constraint family. 
    If there exist $f,g \in \Fam$ such that 
    \(
       2 \leq \umin(f) < \gamma \le \umin(g),
    \)
    then $\CSP_k^\gamma(\Fam)$ requires time $m^{\frac{k}{\umin(g)} - o(1)}$, unless the $\umin(g)$-uniform $k$-Hyperclique Hypothesis fails. 
\end{theorem}
We note that the combination of \autoref{thm:2-vio-hardness} and \autoref{thm:sandwich-hardness} implies the lower bound previously obtained for the non-uniform independent set in arity $h$-hypergraphs, but it also further extends this lower bound to a broader class of Boolean constraint functions.
While our classification is tight for families with $\umin(\Fam) \geq 3$,
in general, the situation gets considerably more delicate. 
Case in point, as demonstrated in the previous section, only restricting ourselves to the families of mixed-arity $\NAND$ are technically challenging. We leave the open question of whether one can extend the algorithmic approaches used for detecting the $k$-independent set in non-uniform hypergraphs to general constraint families that contain functions with different violation parameters. 
A further complication is that some $2$-violating constraints can \emph{hide} $\NAND$ constraints of higher arity. 
For example, consider the constraint function 
$f(x_1,\dots,x_5) = \NAND_2(x_1,x_2) \wedge \NAND_3(x_3,x_4,x_5)$ which is indeed $2$-violating, although there is an efficient reduction from $3$-uniform $k$-hyperclique. Consequently, our lower bounds do not fully capture the landscape of $\umin(\Fam) = 2$ families. 
Finally, $\leq 1$-violating constraint families introduce more intricacies, as they easily allow to force variables to $1$. This enables reductions from e.g. $\CSP_k^{2.5}(\{\NAND_2,\NAND_3\})$ to $\CSP_{k-1}^{2.5}(\{\IMPL,\NAND_3\})$, although at the cost of a slightly smaller parameter $k-1$.
As such, we leave it open to future work to fully classify Boolean constraint families with $\umin(\Fam)\leq 2$.

\subsection{Algorithms for Trivial Settings}
In this section we give an algorithm proving the triviality cutoff, i.e., that every sufficiently sparse instance is trivial to solve.
This algorithm is a non-uniform generalization of the previously seen $k$-independent set algorithm in sparse graphs.

\begin{theorem}
\label{thm:triviality_cutoff_alg}
    Let $\Fam$ be Boolean constraint family with $\gamma < \umin(\Fam)$.
    Then we can solve $\CSP^\gamma_k(\Fam)$ in time $\O(n + m)$.
\end{theorem}
The approach taken is a generalization of \autoref{lem:constructive_turan}:
the algorithm greedily picks a variable occurring only in ``few'' constraints, adds it to the solution and adjusts the variable's neighborhood in such a way that the instance's sparsity does not meaningfully change. This procedure is carried out $k$ times in total, yielding a valid solution.

Given some family $\Fam$ and an instance of $\CSP_k^\gamma(\Fam)$, we denote by $d_f(v)$ the degree of a variable $v$ with respect to $f \in \Fam$, i.e. be the number of $f$-constraints variable $v$ is contained in.
Using this notion, we now first prove a lemma stating that instances of specific density exponents always contain some variable of low degree. 
For a constraint family $\Fam$ together with density parameters $\gamma_f$ for each $f\in \Fam$, define
\(
\Fam_{\le i} \coloneqq \{f\in \Fam \mid \gamma_f < i\}.
\)
We define $\Fam_{\ge i}, \Fam_{<i}, \Fam_{>i}$ analogously.
\begin{lemma}
    \label{lemma:small_degree_existence}
    Let $\Fam$ be a finite family of Boolean constraint functions with density parameters $\gamma_f$ for each $f\in \Fam$. Consider any instance $\Phi$ of $\CSP_k^\gamma(\Fam)$ with $n$ variables such that for every $f\in \Fam$, the number of $f$-constraints in $\Phi$ is $m_{f} = \Theta(n^{\gamma_f})$.  
    Then if $n$ is sufficiently large, there exists a variable $v$ with $d_g(v) = 0$ for all $g \in \Fam_{< 1}$ and $d_f(v) \le F m_f/n =  O(n^{\gamma_f-1})$ for all $f \in \Fam_{\geq1}$.
\end{lemma}

% \timo{For sure plenty of stuff can be improved regarding form, structure and style. I personally believe however that this proof is technically correct. Please let me know if you think otherwise.}
% \marvin{conditions like $d_f(v) \in \omega(n)$ are formally questionable. Isn't this just $d_f(v) \le m_f/n$?}
% \marvin{Thinking more about this, can't we prove $d_f(v) \le |\Fam| m /n $? }
\begin{proof}
    Assume for contradiction that for every variable $v$, there is some $f \in \Fam_{\geq 1}$ such that $d_f(v) > F m_f / n$ or there is some $f \in \Fam_{<1}$ such that $d_f(v) \geq 1$. For any $v$, pick an arbitrary such constraint $f_v$. By averaging, there exists $f^*\in \Fam$ that is picked at least $n/F$ times, i.e., $|\{v \mid f_v = f^*\}| \ge n/F$.

    If $f^*\in \Fam_{\geq 1}$, each variable $v$ with $f_v= f^*$ contributes more than $Fm_f/n$ constraints, thus the total number of $f^*$-constraints is larger than $(n/F)(Fm_{f^*}/n) = m_{f^*}$, a contradiction. If $f^* \in \Fam_{< 1}$, then each variable $v$ with $f_v=f^*$ contributes at least one constraint, so that the number of $f^*$-constraints is at least $n/F = \Theta(n)$, which is a contradiction to $f^*\in \Fam_{<1}$.

\end{proof}

With this argument we can now also generalize the algorithm to arbitrary families with 0-valid constraint functions of higher arity.

\begin{theorem}
    \label{thm:0valid_sparse_general_algo}
    Let $\Fam$ be a finite family of non-trivial constantly many 0-valid Boolean functions and denote by $q_f$ the arity of $f \in \Fam$.
    Let $\Phi$ be an instance of $\CSP_k^\gamma(\Fam)$ on $n$ variables and $m = \Theta(n^\gamma)$ constraints with constraint counts $m_f$ for each $f \in \Fam$.

    Then there exists some $n_0$ such that for all $n \geq n_0$ the following holds:
    If for all $f \in \Fam$ there is some $\varepsilon_f > 0$ such that $m_f \in \O(n^{\umin(f) - \varepsilon_f})$, $\Phi$ is a YES-instance and we can find a solution in time $\O(n+m)$.
\end{theorem}

\begin{proof}
    % Let $\Phi$ be an instance of $\CSP_k^\gamma(\Fam)$ on $n$ variables and with constraint counts $m_f$ for each $f \in \Fam$ with $m_f \in O(n^{\umin(f) -\varepsilon_f})$ for some $\varepsilon_f > 0$. 
    We note that $\umin(f) \geq 1$ for all $f \in \Fam$ as $\Fam$ is 0-valid.
    % $m$ denotes the total constraint count with $\gamma$ such that $m \in \Theta(n^\gamma)$. 
    We partition the family $\Fam$ into two parts: 
    \begin{align*}
    \Fam_{=1} &\coloneqq \{f \in \Fam \mid \umin(f) = 1\}\\
    \Fam_{\geq 2} &\coloneqq \{f\in \Fam \mid \umin(f) \geq 2 \} = \Fam \setminus \Fam_{=1}.
    \end{align*} Importantly, functions in $\Fam_{\geq 2}$ are satisfied by any assignment of weight $\leq 1$ whereas functions in $\Fam_{=1}$ might be violated by some weight-1-assignment.

    \textbf{Algorithm Description}
    Our algorithm works recursively: For $k = 0$ we can assign every variable to false and obtain a satisfying assignment, as every $f \in \Fam$ is 0-valid.
    
    For $k \geq 1$, we apply \cref{lemma:small_degree_existence} to find some variable $v$ with $d_f(v) \le|\Fam|m_f/ n \in O(n^{\umin(f) - 1 -\varepsilon_f})$ for all $f \in \Fam_{\geq 2}$ and $d_f(v) = 0$ for all $f \in F_{=1}$. 
    We now set this variable $v$ to true adjusting every constraint containing $v$. Let $c$ be an $f$-constraint $c = f(x_1, \dots, x_{i-1}, v, x_{i+1}, \dots,x_{q_f})$ for some $f \in \Fam_{\geq 2}$. We now define a new function $f_i$ on $q_f - 1$ variables: 
    \begin{equation*}
        f_i(v_1, \dots, v_{q_f -1}) \coloneqq f(v_1, \dots, v_{i-1}, 1, v_i, \dots, v_{q_f-1})
    \end{equation*}
    Adding $f_i$ to the family of constraint functions, we replace the $f$-constraint $c$ with the $f_i$-constraint $c' \coloneqq f_i(x_1, \dots, x_{i-1}, x_{i+1}, \dots, x_{q_f})$.
    If $f_i$ turns out to be trivially satisfied, we can drop it. This procedure is performed for every constraint containing $v$ and we obtain a new family $\Fam' \supseteq \Fam$. We can now remove the isolated $v$ from the initial instance and recurse on the newly obtained instance $\Phi'$ of $\CSP_{k-1}^\gamma({\Fam'})$.

    \textbf{Runtime}
    In each recursion step we can find a low-degree variable in time $\O(n+m)$. 
    Each step introduces at most $2^{q_f} \in \O(1)$ new constraint functions for every $f \in \Fam$, 
    thus $\Fam'$ is of size $|\Fam'| \le |\Fam| \cdot 2^{q} \in \Theta(|\Fam|) = \Theta(1)$ where $q \coloneqq \max_{f \in \Fam} q_f$. 
    Replacing incident constraints takes at most $O(m)$ time. 
    As we recurse only $k$ times, the algorithm therefore runs in time $\O(n+m)$ which is essentially optimal as the input is of size $\Theta(n+m)$.
    
    \textbf{Correctness}
    To argue the correctness of our algorithm we need to show that variables added to our solution do not violate any constraint. 
    Further we need show that sparsity is maintained for each constraint function and that $\Fam'$ is  $0$-valid in every step.

    % Let $\Phi$ be the instance of $\CSP_k^\gamma(\Fam)$ on $n$ variables and $m$ constraints with constraint counts $m_f$ for all $f \in \Fam$ and $\gamma$ and $\gamma_f$, respectively. 
    Let $v$ be the low degree variable chosen to be set to true. Denote by $\Phi'$ the instance of $\CSP_{k-1}^{\gamma'}(\Fam')$ obtained by the algorithm on $n'$ variables and $m',m_f', \gamma', \gamma_f'$, respectively.
    Setting $v$ to true does not violate any constraint, as by choice of $v$ it is only contained in constraints $f \in\Fam_{\ge 2}$ with $\umin(f)\geq 2$. Additionally, this ensures that the new function $f_i$ has at least $\umin(f_i) \geq 1$, thus $\Fam'$ is $0$-valid.

    It is immediate that $n' = n - 1$ and $m' \le m$. 
    It is left to prove that for all $f \in \Fam'$ it holds that $m'_f \in \O\left( n'^{\umin(f) - \varepsilon_f} \right) = \O\left( n^{\umin(f) - \varepsilon'_f} \right)$ for some $\varepsilon'_f > 0$.

    Consider some $f_i \in \Fam'$. There are two cases to consider.
    If $f_i \in \Fam$, but we did not replace any constraints by $f_i$-constraints, there is nothing to show.
    Thus, assume the constraint function $f_i$ replaced some $f \in \Fam$ during the procedure. 
    Recall that we replaced only $f$-constraints containing $v$ with $d_{f}(v) \leq |\Fam|m_{f}/n$.
    Thus we introduce at most 
    \[
      d_f(v) \leq |\Fam| m_{f}/n = \O\left( n^{\umin(f) - 1 - \varepsilon_{f}}\right) = \O\left(n^{\umin(f_i) - \varepsilon_{f}} \right)
    \]
    many new constraints for some $\varepsilon_{f} > 0$ as $\umin(f_i) \geq \umin(f) - 1$.
    If $f_i \not \in \Fam$, this suffices, as $m'_{f_i} = d_f(v)$.
    If $f_i \in \Fam$, then $m'_{f_i} = m_{f_i} + d_f(v) = \O\left(n^{\umin({f_i}) - \varepsilon'_{f_i}}\right)$ for $0 < \varepsilon'_{f_i} \leq \max(\varepsilon_{f},\varepsilon_{f_i})$.
    Thus the sparsity properties are maintained such that for every $f_i \in \Fam'$ we have that $m'_{f_i} = \O\left(n^{\gamma_{f_i} - \varepsilon'_{f_i}}\right)$ for some $\varepsilon'_{f_i} > 0$.
\end{proof}

This concludes the generalization of the motivating observation, that $k$-Independent Set instances on $O(n^{2-\varepsilon})$ edges always contain a solution for sufficiently large $n$. 

Finally, we show that for any constraint function $f$, $\CSP_k^{\umin(f)}(f)$ is non-trivial,
i.e. it is not guaranteed that we always have a YES- or NO-instance for any given instance.
\begin{theorem}
    \label{lemma:non-triviality-proof}
    Let $f$ be any constraint function. Then $\satg{k}{f}{\umin(f)}$ is non-trivial.
\end{theorem}
\begin{proof}
    We divide the proof into three cases based on the constant parameter $\umin(f)$ that is either $0, 1$ or $\geq 2$. Let $q \geq 1$ be the arity of $f$. 

    \subparagraph*{Case $\umin(f) = 0$}
    In this case, there is a violating assignment of weight 0 for $f$, i.e. $f$ is 0-invalid. Let $n \in \mathbb{N}$ with $n \geq (k+1)\cdot q$ and we now construct a NO-instance $\Phi$ on $n$ variables and $O(n^0) = O(1)$ constraints.
    Let $V = \{v_1, \dots, v_n\}$ be the set of variables of $\phi$. Then let $V_1, V_2, \dots V_{k+1} \subset V$ be some arbitrary disjoint subsets of $V$. Such sets exists as we chose $n \geq (k+1)\cdot q$. We now place the constraints $f(V_i)$ for all $V_i$. 
    This resulting instance uses precisely $k+1$ $f$-constraints and must be a NO-instance, as some variable in each $V_i$ must be assigned to true since $f$ is 0-invalid. But since there are $k+1$ such disjoint sets, at least $k+1$ variables must be assigned to true and no size-$k$ satisfying assignment exists.

    \subparagraph*{Case $\umin(f) = 1$}
    Using \autoref{lem:atMostConstruction}, we can construct the $\textsc{LessThan}^K_1$ gadget. By adding a total of at most $O(n)$ many gadgets, we can ensure that every variable appears in at least one such constraint.
    As such, all variables must be set to false by the property of $\textsc{LessThan}^K_1$. This is therefore a NO-instance, proving this case.
    
    \subparagraph*{Case $\umin(f) \geq 2$}
    This follows directly from \autoref{thm:2-vio-hardness}, as we reduce from the $k$-independent set problem, which indeed is a non-trivial problem for $\gamma = 2$, proving our claim.
\end{proof}
\subsection{Hardness}
First, we show how to employ $c$-violating constraint functions to create a dense or sparse \emph{embedding} of a given $\CSP_k^\gamma$ instance, which is the crucial idea behind our lower bounds. Then we will give conditional lower bounds for $\geq 2$-violating constraint families.

\subsubsection{Embeddings}
We introduce three constructions in this part. 
The first one allows us to construct the constraint $\textsc{LessThan}^K_{c}$ on $K$ variables, that is satisfied if and only if less than $c$ of those $K$ variables are set to $1$, from any $c$-violating family.
Building on this, we give our \emph{embedding constructions}, 
allowing to embed an instance of density $\gamma$ into an instance of increased or decreased density $\gamma'$.
Our constructions are subject to two main considerations: First, they need to be sufficiently sparse, as to not blow up the density. Secondly, we do not allow to replicate variables in a single constraint, due to the sparsity setting. 
We avoid the latter by introducing auxiliary variables, which we then enforce to be set to false for any satisfying assignment.
\begin{lemma}
    \label{lem:atMostConstruction}
     Let $f$ be $0$-valid and $c$-violating of arity $h$ with $c \geq 1$. 
     Then we can construct a gadget formula $\textsc{LessThan}^K_{c}$ for $K = k + h$ in $\CSP_k^\gamma(f)$.
\end{lemma}
\begin{proof}
    % Note that for $h = c$ the function $f$ is equivalent to $\NAND_h \equiv \textsc{LessThan}^h_h$, which we can pad to obtain the desired construction.
    
    Let $f^* = f_\text{sym}$. Note that $f^*$ is still $c$-violating. 
    We construct the gadget $\textsc{LessThan}^K_c$ on $K = k+h$ many new variables $Y = \{y_1,\dots,y_K\}$.
    For every $\mathbf{x} \in \binom{Y}{h}$ add the constraint $f^*(\mathbf{x})$.
    
    By $c$-violation of $f^*$, any assignment of weight less than $c$ is satisfying.
    We claim that every satisfying assignment with weight at least $c$, 
    is of weight at least $k+1$:
    Assuming there was a satisfying assignment $a$ of weight $c'$ where $c < c'  \leq k$, 
    consider tuple $C' = (y_{i_1},\dots,y_{i_{c'}})$ such that $a(y_{i_j}) = 1$ 
    and let $C = (y_{i_1},\dots,y_{i_{c}})$.
    By a counting argument, $K-(h-c)$ distinct variables are needed to extend $C$ to an $h$-tuple,
    but there are only $c' - c$ variables set to $1$ under $a$.
    As $K-(h-c) = k + c \geq k + 1 > c'$ , there exists at least one $h$-tuple $\mathbf{y}$ 
    such that exactly $c$ variables are set to $1$ under $a$, leading to a contradiction as we included a constraint $f^*$ on every $h$-tuple.
\end{proof}
\begin{lemma}[Dense Embedding]
\label{lem:dense_embed}
    Let $f$ be $c$-violating and of arity $h$ with $c \ge 2$ and $c \leq \gamma \leq h$.
    Given an instance $\Phi$ of $\CSP_k^{c}(\NAND_{c})$ on $n$ variables,
    we can construct an equivalent instance $\Phi'$ of $\CSP_k^{\gamma}(f)$ on $\Theta(n)$ variables.
\end{lemma}
% \timo{I fixed it for general $\delta$, but im not entirely sure if this is correct and the technicality of this just sucks in general. We could make it easier if we let $\delta = c, h > c$, which is pretty much the only thing we actually need. If we let this, the first case disappears.}
% \julian{Considering the case for which this would be useful is not really important, let's just roll it back.}
\begin{proof}
    Note that for $c = h$, the constraint $f$ is exactly $\NAND_c$, and we there is nothing to show.
    Otherwise let $c < h$ and by \autoref{lem:atMostConstruction} can construct $\textsc{LessThan}^K_{c}$.
    Let $Y  = \{y_1,\dots,y_\ell\}$ be a set of new variables for $\ell  = n^{(\gamma-c)/(h-c)} \leq n$.
    We construct $\Phi'$ in the following manner:

    For every constraint $\NAND_{c}(\mathbf{x}) \in \Phi$ and all $\mathbf{y} \in \binom{Y}{K-c}$ we introduce a constraint $\textsc{LessThan}^K_{c}(\mathbf{x},\mathbf{y})$. Since we constructed $\textsc{LessThan}^K_c$ out of multiple $f$-constraints the number of created constraints is $\O\big(n^c \cdot |Y|^{h-c}\big)$.
    We obtain $\Theta\big(n^c \cdot n^{((\gamma-c)/(h-c))\cdot(h-c)}\big) = \Theta\big(n^\gamma\big)$. As $\ell \leq n$, $\Phi'$ consists of $\Theta(n)$ variables and $\Theta(n^\gamma)$ constraints, obtaining the desired density.

    Finally, we show that the resulting instance is equivalent.
    If there exists a solution $a$ to $\Phi$, then extending $a$ by setting all variables $Y$ to $0$ is a solution to $\Phi'$ as well, as $\textsc{LessThan}^K_{c}$ is satisfied by any assignment of weight less $c$ and thus consistent with the $\NAND_c$.
    Conversely, assume there was a satisfying assignment $a$ of weight $k$ for $\Phi'$. 
    We claim that if any variable in $y \in Y$ is set to $1$ by $a$, 
    then we can instead set an arbitrary $x \in \text{vars}(\Phi)$ to $1$ and $y$ to $0$. 
    Clearly any constraint that contains $x$ and $y$ or $y$ only is still satisfied. 
    Thus consider any constraint $f(\mathbf{x},\mathbf{y})$ such that $x \in \mathbf{x}$ but $y \not \in \mathbf{y}$. 
    As every constraint $f(\mathbf{x},\mathbf{y'})$ with $y \in \mathbf{y'}$ was satisfied by $a$ originally, 
    at most $c-2$ many variables of $\mathbf{x}$ are set to $1$ by $a$. 
    Therefore $a \cup \{x\} \setminus \{y\}$ is also satisfying.
    Exhaustive application obtains a satisfying weight $k$ assignment for $\Phi$ where every variable in $Y$ is set to $0$.
    \end{proof}
% With this result, let us swiftly prove that we can extend its application to families of mixed $\NAND$-constraints in a straightforward manner.
% \begin{theorem}
% \label{lem:dense_embed_families}
%     Let $\Fam = \{\NAND_c \mid c \in C\}$ for some set $C \subseteq [k - 1]$ and let $\Fam' = \{f_c\mid c \in C\}$ where $f_c$ is $c$-violating.
%     Then we can reduce an instance $\Phi$ of $\CSP^\gamma_k(\Fam)$ with $m_c = \O(n^{\gamma_c})$ many $\NAND_c$ constraints for $c \in C$
%     to an instance $\Phi'$ of $\CSP^\gamma_k(\Fam')$ such that $n' = \Theta(n)$ and $m_{f_c} = \O(n^{\gamma_c})$ for all $c \in C$.
% \end{theorem}
% \begin{proof}
%     We prove this, by applying \autoref{lem:dense_embed} piecewise for every $\NAND_c$ and $f_c$ for $c \in C$.
%     More precisely, we consider the instances $\Phi_c$ on $\text{vars}(\Phi)$, but only containing $\NAND_c$-constraints, i.e. instances of $\CSP_k^{\gamma_i}(\NAND_c)$.
%     We then apply \autoref{lem:dense_embed} to embed $\Phi_c$ of density $\gamma_c$ into an instance $\CSP_k^{\gamma}(f_c)$ of the same density.
%     The final instance then is $\Phi' = \bigwedge_{c \in C} \Phi'_c$ on $\Theta(n)$ variables and $m_{f_c} = \O(m_c) = \O(n^{\gamma_c})$.
%     Correctness can be verified by the equivalence between weight $k$ satisfying assignments of $\Phi_c$ and $\Phi'_c$ respectively.
% \end{proof}

\begin{lemma}[Sparse Embedding]
\label{lem:sparse_embed}
    Let $f$ be $d$-violating of arity $h$ where $1 \leq d \leq \gamma $. Given an instance $\Phi$ of $\CSP_k^\delta(\Fam)$ for $\delta > \gamma$,
    we can construct an equivalent instance $\Phi'$ of $\CSP_{k}^\gamma(\Fam \cup \{f\})$ on $\Theta(n^{\delta/\gamma})$ variables.
\end{lemma}
\begin{proof}
    For this construction, we crucially employ the fact that $d \le \gamma$, which allows us to place $\O(n^d)$ many constraints.
    We create $l:= n^{\delta/\gamma}$ many new variables $Y = \{y_1,\dots,y_l\}$, as to construct a sparse number of constraints, such that all variables in $Y$ are forced to $0$.
    By \autoref{lem:atMostConstruction} we can  construct $\textsc{LessThan}_{d}^{K}$ using $f$.
    To this end define two formulas: 
    \begin{align*}
    \Psi_{<d} &\coloneqq \bigwedge_{\mathbf{y} \in \binom{Y}{d}}\textsc{LessThan}_{d}^{K}(\mathbf{y},\underbrace{\makebox[5em]{\dots}}_{\text{variables } y' \in Y})\\ 
    \Psi_{<} &\coloneqq \bigwedge_{\mathbf{x}  \in \binom{\text{vars}(\Phi)}{d-1}} \bigwedge_{y \in Y}\textsc{LessThan}_{d}^{K}(\mathbf{x},y,\underbrace{\makebox[5em]{\dots}}_{\text{variables } y \in Y}).
    \end{align*}
    where we pad the remaining positions with arbitrary, distinct variables $y' \in Y$.
    $\Psi_{<d}$ has $\O(n^{\delta / \gamma \cdot d})$ many constraints and $\Psi_{<}$ consists of at most $\O(n^{(\delta / \gamma) + d - 1})$ constraints.
    Note that for all $1 \leq d < \gamma < \delta$ it holds that $\delta/\gamma + d - 1 \leq \delta$.
    The resulting instance $\Phi' \equiv \Phi \wedge \Psi_{<d} \wedge \Psi_<$ has a density of $m = \O(n^\delta) = \O(N^{\gamma})$.
    Let $a$ be a solution to $\Phi$, then extending $a$ by setting all variables $Y$ to $0$ is a satisfying weight $k$ assignment for $\Phi'$.
    Conversely, assume there was a solution $a$ for $\Phi'$. 
    Assume $0 < d'$ many variables in $Y$ are set to $1$ with $d' < d$, otherwise $a$ would violate $\Psi_{<d}$ already.
    Then $a$ sets $k-d'$ many variable from $\text{vars}(\Phi)$ to $1$.
    As there exists at least one $(d - d') $-tuple of variables set to $1$ in $\text{vars}(\Phi)$, 
    by construction of $\Psi_{<}$ all variables in $Y$ therefore must be set to $0$.
Hence we obtain an equivalent instance of  $\CSP_{k}^\gamma(\Fam \cup \{f\})$ on $\Theta(n^{\delta / \gamma})$ many variables and a density exponent of $\gamma$.
\end{proof}
\subsubsection{Lower Bounds}
Equipped with this, we can now state our main hardness results. 
\begin{lemma}
\label{lem:c-violating-hardness}
Let $\Fam$ be a constraint family with $\umin(\Fam) \le \gamma$ and $\varepsilon > 0$.
If $\umin(\Fam) \geq 3$, then $\CSP_k^\gamma(\Fam)$ cannot be solved in time $\O\left(n^{k-\varepsilon}\right)$ assuming the $\umin(\Fam)$-uniform hyperclique hypothesis.
If $\umin(\Fam) = 2$, then $\CSP_k^\gamma(\Fam)$ cannot be solved in time $O\left(n^{\frac{\omega k}{3}-\varepsilon}\right)$ assuming the $k$-Clique Hypothesis.
\end{lemma}

\begin{proof}
    Fix some $f \in \Fam$ with $\umin(f) = \umin(\Fam)$.
    We reduce from an instance $\Phi$  of $\CSP_k^{\umin(\Fam)}(\NAND_{\umin(\Fam)})$.
    Applying \autoref{lem:dense_embed} obtains an equivalent instance $\Phi'$ of $\CSP_k^{\gamma}(f)$ on $\Theta(n)$ variables.
    For $\umin(\Fam) \geq 3$, assume there was an $\varepsilon > 0$ such that $\Phi'$ can be solved in time $\O(n^{k - \varepsilon})$.
    Then one could solve $\Phi$ in time $\O(n^{k - \varepsilon'})$ for some $\varepsilon' > 0$, refuting the $\umin(\Fam)$-uniform $k$-hyperclique hypothesis. 
    For $\umin(\Fam) = 2$, we get a lower bound of $\O(n^{\omega/3 k -\varepsilon'})$ for some $\varepsilon' > 0$ under the $k$-clique hypothesis instead.
\end{proof}
Intuitively, any instance that is dense enough allows a reduction from dense $k$-independent set.

We can derive the following lower bounds depending on the density of the instance, 
if there are further constraints available that allow us to employ sparse embedding.
Just as for binary constraint families, 
we blow up the instance size to hide a dense instance that is sufficiently hard.
\begin{lemma}
\label{lem:d-c-violating-hardness}
    Let $f,g \in \Fam$ where $1 \leq  \umin(f) < \umin(g)$ and $\umin(f) \leq \gamma \leq \umin(g)$.
    Then there is no $\varepsilon > 0 $ such that 
    we can solve $\CSP_k^\gamma(\Fam)$ in time
    \begin{itemize}
        \item $\O(m^{k/\umin(g) - \varepsilon})$ if $\umin(g) \geq 3$
        \item $\O(m^{\omega k/6 - \varepsilon})$ if $\umin(g) = 2$
    \end{itemize}
    conditioned on the $\umin(g)$-uniform $k$-Hyperclique Hypothesis and $k$-Clique Hypothesis respectively.
\end{lemma}
\begin{proof}
    Let $c:= \umin(g)$.
    We reduce from an instance $\Phi$ of $\CSP_{k}^c(\NAND_c)$.
    Applying \autoref{lem:dense_embed} and \autoref{lem:sparse_embed} successively obtains a sparse embedding of $\Phi$ in $\CSP_k^{\gamma}(\{f,g\})$ on $N = \Theta(n^{c/\gamma})$ many variables, which we denote by $\Phi'$.
    Assuming there was an $\varepsilon > 0$ such that $\Phi'$ can be solved in time $\O(m^{k/c - \varepsilon})$,
    then one could solve $\Phi$ in time $\O\left(N^{\frac{\gamma k}{c}k -\varepsilon'}\right) = \O\left(n^{k - \varepsilon'}\right)$ for some $\varepsilon' > 0$, refuting the $c$-uniform $k$-hyperclique hypothesis.
    For $c = 2$, we get a running time of $\O(m^{\omega k /6 - \varepsilon})$ for some $\varepsilon > 0$ refuting the $k$-clique hypothesis instead.
\end{proof}
Consequently, this implies the lower bounds we saw for the non-uniform $k$-independent set problem in arity at most $3$-hypergraphs.

\bibliography{article}

% \appendix 

% \input{sections/triviality_cutoff_proof}
% \input{sections/nandr_eq_regime}
% \input{sections/kernelized_nand_impl}

\end{document}